%% file: main.tex
\documentclass[11pt,a4paper]{article}

\usepackage{fullpage}		 % smaller margins
\usepackage[utf8]{inputenc} % allow utf-8 input
\usepackage[T1]{fontenc}    % use 8-bit T1 fonts
\usepackage{lmodern}			 % Fixes font error
\usepackage{hyperref}       % hyperlinks
\usepackage{url}            % simple URL typesetting
\usepackage{booktabs}       % professional-quality tables
\usepackage{amsthm}			 % theorem environment
\usepackage{amsmath}			 % maths from AMS
\usepackage{amssymb}			 % More maths symbols
\usepackage{amsfonts}       % blackboard math symbols
\usepackage{mathtools}      % useful tools for mathematical typesetting
\usepackage{nicefrac}       % compact symbols for 1/2, etc.
\usepackage{microtype}      % microtypography
\usepackage{subfiles}       % loading tex files
\usepackage{xcolor}         % enable coloured text
\usepackage{bm}             % for bold math symbols
\usepackage{enumitem}		 % for compact tables
\usepackage[numbers]{natbib}  % natbib with round parentheses
\usepackage{siunitx}

\usepackage{xpatch}
\makeatletter
\xpatchcmd{\@thm}{\thm@headpunct{.}}{\thm@headpunct{}}{}{}
\makeatother

\makeatletter
\def\th@plain{%
  \thm@notefont{}% same as heading font
  \itshape % body font
}
\def\th@definition{%
  \thm@notefont{}% same as heading font
  \normalfont % body font
}
\makeatother

\newtheorem{theorem}{Theorem}
\newtheorem{definition}{Definition}
\newtheorem{lemma}[theorem]{Lemma}

\newtheorem{assumption}{Assumption}

\definecolor{myred}{HTML}{880000}
\definecolor{mygreen}{HTML}{008800}
\definecolor{myblue}{HTML}{000088}
\definecolor{linkblue}{HTML}{0000BB}

\hypersetup{
  colorlinks,
  linkcolor={myred},
  citecolor={myblue},
  urlcolor={linkblue}
}

\newcommand{\R}{\mathbb{R}}
\newcommand{\E}{\mathbb{E}}
\newcommand{\x}{\mathbf{x}}
\newcommand{\X}{\mathbf{X}}
\newcommand{\A}{\mathbf{A}}
\newcommand{\Eins}{\mathrm{\textbf{1}}}
\newcommand{\gauss}{\mathcal{N}}
\renewcommand{\P}{\mathbb{P}}
\renewcommand{\S}{\mathcal{S}}

\title{Nearly Minimax-Optimal Rates for Noisy Sparse Phase Retrieval via Early-Stopped Mirror Descent}
\author{
  Fan Wu,
  Patrick Rebeschini \\
	Department of Statistics, University of Oxford
}

\begin{document}

\maketitle

\begin{abstract}%
	\input{abstract.tex}
\end{abstract}

%--------------------------------- Sub-files ----------------------------------
%------------------------------------------------------------------------------
\input{introduction.tex}
\input{sparse_phase_retrieval.tex}
\input{mirror_descent_algorithms.tex}
\input{main_results.tex}
\input{numerical_study.tex}
\input{proofs.tex}
\input{conclusion.tex}

%-------------------------------- Bibliography --------------------------------
%------------------------------------------------------------------------------
\bibliography{references}
\bibliographystyle{plainnat}

%---------------------------------- Appendix ----------------------------------
%------------------------------------------------------------------------------
\clearpage
\appendix

{\LARGE\textbf{Appendix}}
\vspace{3mm}

The Appendix is organized as follows. We prove Theorem \ref{theorem_discrete} in Appendix \ref{supp:proof_theorem4}. In Appendix \ref{supp:proof_supporting_lemmas}, we prove Lemmas \ref{lemma:support1}--\ref{lemma:support3} from Section \ref{proof:theorem3}. Further technical lemmas are provided in Appendix \ref{supp:technical_lemmas}.

\input{proof_theorem4.tex}
\input{proof_supporting_lemmas.tex}
\input{technical_lemmas.tex}

\end{document}

%% file: abstract.tex
This paper studies early-stopped mirror descent applied to noisy sparse phase retrieval, which is the problem of recovering a $k$-sparse signal $\x^\star\in\R^n$ from a set of quadratic Gaussian measurements corrupted by sub-exponential noise. We consider the (non-convex) unregularized empirical risk minimization problem and show that early-stopped mirror descent, when equipped with the hyperbolic entropy mirror map and proper initialization, achieves a nearly minimax-optimal rate of convergence, provided the sample size is at least of order $k^2$ (modulo logarithmic term) and the minimum (in modulus) non-zero entry of the signal is on the order of $\|\x^\star\|_2/\sqrt{k}$. Our theory leads to a simple algorithm that does not rely on explicit regularization or thresholding steps to promote sparsity. More generally, our results establish a connection between mirror descent and sparsity in the non-convex problem of noisy sparse phase retrieval, adding to the literature on early stopping that has mostly focused on non-sparse, Euclidean, and convex settings via gradient descent. Our proof combines a potential-based analysis of mirror descent with a quantitative control on a variational coherence property that we establish along the path of mirror descent, up to a prescribed stopping time.

%% file: introduction.tex
\section{Introduction}
\label{section:introduction}
In many fields of science and engineering, one is tasked with phase retrieval, the problem of recovering an $n$-dimensional signal $\x^\star$ from phaseless measurements
\begin{equation}
\label{eq:obs_model}
Y_j = (\A_j^\top\x^\star)^2 + \varepsilon_j, \qquad j=1,...,m,
\end{equation}
where the sensing vectors $\A_1,...,\A_m$ are observed, and the random variables $\varepsilon_1,...,\varepsilon_m$ model a possible contamination of the measurements with noise. This problem arises naturally in applications such as optics, X-ray crystallography and astronomy, where it is often easier to build detectors which measure intensities and not phases \cite{SEC+14}.
In many applications of interest, the underlying signal $\x^\star$ is naturally sparse \cite{DF87, JEH16, M90}. (Noisy) sparse phase retrieval represents a cornerstone in high-dimensional statistics. It is more challenging than sparse linear regression due to the missing phase information, which often leads to non-convex problems. A geometric analysis of a natural least-squares formulation of phase retrieval (without sparsity and noise) shows that, with high probability, there are no spurious local minimizers, provided the number of Gaussian measurements $m$ is at least of order $n$ modulo logarithmic terms \cite{SQW18}. However, recovery of a sparse signal is possible from potentially far fewer measurements $m\ll n$, in which case the least-squares formulation studied in \cite{SQW18} has potentially exponentially many local minimizers.

Over the years, a variety of methods have been proposed to tackle the problem of sparse phase retrieval. These methods typically enforce sparsity either by adding explicit regularization via penalty terms \cite{CCG15, OYDS12}, sparsity inducing priors \cite{SR15} or sparsity constraints \cite{JH19, NJS15, SBE14}, or by using algorithmic principles such as thresholding steps \cite{CLM16, WZGAC18, ZWGC18}. 
The noisy model has been studied and stability results have been established in \cite{CCG15} for a convex relaxation-based formulation, which includes an $\ell_1$-penalty term to promote sparsity, and in \cite{JH19} for an alternating minimization-based approach that uses an explicit sparsity constraint to enforce sparsity. However, the optimal dependence of the estimation error rate on the sample size, the sparsity of the signal and the noise level was not considered in aforementioned works. Empirical risk minimization with sparsity constraint has been shown to achieve a nearly minimax-optimal rate of convergence in \cite{LM15}, though it does not lead to a tractable algorithm due to the non-convex sparsity constraint. Thresholded Wirtinger flow (TWF) \cite{CLM16} is a non-convex optimization-based algorithm that relies on thresholding steps to enforce sparsity, and has been shown to achieve a nearly minimax-optimal rate of convergence for noisy sparse phase retrieval, provided the number of measurements $m$ is of order $k^2$ modulo logarithmic terms. 

Mirror descent was first introduced by Nemirovski and Yudin \cite{NY83} for solving large-scale convex optimization problems, and has since gained in popularity in various optimization and machine learning settings. An appealing property of mirror descent is the fact that it can be adapted to the (possibly non-Euclidean) geometry of the problem at hand by choosing a suitable strictly convex function, the so-called mirror map. In convex optimization, the algorithm often admits a potential-based convergence analysis in terms of the Bregman divergence associated to the mirror map, see e.g.\ \cite{BT03, B15, NJLS09}. There is a large body of literature that analyzes mirror descent in optimization (both convex and non-convex) and online learning settings, see e.g.\ \cite{ABL13, GL13, S11}. In particular, an algorithm based on online mirror descent has been studied for an online version of phase retrieval (without sparsity) \cite{KN19}. In the case of \emph{noiseless} sparse phase retrieval, an approach based on continuous-time mirror descent has been recently shown to lead to recovery of the signal up to an arbitrary precision, without requiring explicit regularization terms or thresholding steps \cite{WR20b}. These results, however, do not apply to noisy measurements where perfect recovery of the signal is no longer possible and where the iterates of mirror descent may cease to be (approximately) sparse as time increases.

In the present work, we build upon the framework considered in \cite{WR20b} and 
we present a full analysis of \emph{early-stopped} mirror descent (both continuous-time and discrete-time), equipped with the hyperbolic entropy mirror map, applied to minimize the (non-convex) unregularized empirical risk to solve \emph{noisy} sparse phase retrieval. We establish the lower bound $\frac{\sigma}{\|\x^\star\|_2^2}\sqrt{k/m}$ for the minimax-optimal rate of convergence for noisy sparse phase retrieval under sub-exponential noise for signals $\x^\star$ with minimum non-zero entry (in modulus) on the order of $\|\x^\star\|_2/\sqrt{k}$, where $\sigma$ describes the noise level. We prove that early-stopped mirror descent achieves this rate up to a factor $\sqrt{\log n}$, provided the sample size is on the order of $k^2$ (modulo logarithmic term) and the sensing vectors $\A_1,...,\A_m$ are independent standard Gaussian vectors. Up to logarithmic terms, the sample size requirement matches that of existing results in the literature on (noisy) sparse phase retrieval \cite{CLM16, CCG15}.

The proposed approach yields a simple algorithm for noisy sparse phase retrieval, which, unlike most existing algorithms, does not rely on added regularization terms or algorithmic principles such as thresholding steps to enforce sparsity of the estimates. Running mirror descent involves tuning at most two parameters, namely the step size $\eta$ (only for the discrete-time version of the algorithm) and the mirror map parameter $\beta$. Our analysis shows that $\eta$ should be chosen smaller than a quantity depending on the signal size $\|\x^\star\|_2$, while the choice of $\beta$ depends on $\|\x^\star\|_2$ and the (known) ambient dimension $n$. As the signal size $\|\x^\star\|_2$ can be easily estimated by considering the average size of the observations \cite{CLS15, WGE17}, the tuning of $\eta$ and $\beta$ only involves known or easily estimated quantities, and, in particular, does not require knowledge (or estimation) of the sparsity $k$ or the noise level $\sigma$. The optimal stopping time, whose knowledge is not required a priori to run mirror descent, is the only parameter that depends on the unknown quantities $k$ and $\sigma$. Our numerical simulations attest that a commonly-used data-dependent stopping rule based on cross-validation, where no knowledge of $k$ or $\sigma$ is used, yields results that validate our theoretical findings.

The main idea behind the proof of our results is to combine a generic potential-based analysis of mirror descent in terms of the Bregman divergence with a quantitative control of a variational coherence property that we establish along the path of mirror descent, when properly initialized, up to a prescribed stopping time. Variational coherence has been previously used as an \emph{a-priori} assumption to establish convergence results for mirror descent in non-convex optimization. As defined in \cite{ZMBBG20}, variational coherence corresponds to the requirement that the inner product $\langle\nabla F(\x), \x - \x' \rangle$ is non-negative for any vector $\x$ in a certain region, where $F$ is the objective function to be minimized and $\x'$ is a local minimizer of $F$. Two versions of variational coherence have been considered in \cite{ZMBBG20}: a global property that holds for any vector $\x$ and any global minimizer $\x'$, and a local property that only needs to be satisfied for any vector $\x$ in an open neighborhood of a local minimizer $\x'$. The global property precludes the existence of local minimizers (that are not also global minimizers) and saddle points, and is not satisfied in the problem of sparse phase retrieval when $F$ is taken to be the unregularized empirical risk. Further, both the global and local versions of variational coherence are formulated to yield convergence towards local minimizers of the objective function, which is not the goal in noisy sparse phase retrieval as, almost surely, the signal $\x^\star$ does not coincide with a local minimizer $\x'$ of the empirical risk. Moreover, in these formulations, variational coherence only yields convergence results but it does not yield a direct control on the speed of convergence.
In contrast, our analysis unveils and exploits a \emph{quantitative} control on variational coherence within a region where the trajectory of mirror descent lies up to the stopping time, upon proper initialization. The object of our analysis is given by the inner product $\langle\nabla F(\X(t)), \X(t) - \x^\star \rangle$, where $\X(t)$ represents the iterate of mirror descent at time $t$ and where the signal $\x^\star$ now replaces the local minimizer $\x'$ of $F$. In our analysis, we establish a lower bound for this inner product that is strictly positive and directly depends on the iterates of the algorithm. This adaptive lower bound yields a control on the rate of convergence of mirror descent, which is key to investigate early stopping and establish our results.
In existing gradient-based approaches to sparse phase retrieval, the inclusion of thresholding steps confines the algorithm to the low-dimensional subspace of sparse vectors, see e.g.\ \cite{CLM16, WZGAC18, YWW19}. On the other hand, our analysis shows that, when properly initialized, the iterates of mirror descent have negligibly small off-support coordinates until a certain stopping time, allowing to focus on the subspace of (approximately) sparse vectors.

Our work adds to the literature on early stopping for iterative methods, which has been primarily developed in the context of ridge regression and kernel methods \cite{RWY14, WYW19, YLC07, ZY05}, i.e.\ convex problems based on the Euclidean geometry which, in particular, do not involve sparsity. An exception is the work \cite{VKR19}, where a connection between early stopping and sparsity has been established for gradient descent with Hadamard parametrization applied to the (convex) problem of sparse recovery. This result was later recovered in \cite{VKR20} using mirror descent with the hyperbolic entropy mirror map. Our results extend this line of research by establishing a connection between early stopping and sparsity for mirror descent applied to the non-convex problem of noisy sparse phase retrieval. Our approach shows how to apply the generic potential-based analysis of mirror descent to non-convex problems where a quantitative control on variational coherence is exhibited along the path traced by the algorithm.

The remainder of the paper is organized as follows. In Section \ref{section:sparse_phase_retrieval}, we give the required background and a brief literature review on sparse phase retrieval. In Section \ref{section:the_algorithm}, we describe our approach to solving sparse phase retrieval using mirror descent. In Section \ref{section:main_results}, we present our main results for mirror descent in continuous time and in discrete time. In Section \ref{section:numerical_simulation}, we perform a numerical study to verify that the dependence of the estimation error achieved by mirror descent with respect to the noise level $\sigma$, sample size $m$ and sparsity $k$ can match the behavior prescribed by our upper bounds. In Section \ref{section:proofs}, we present the proof of the results for the continuous-time algorithm. The proof of the results for the discrete-time algorithm involves further technicalities, and we defer it to the Appendix.

\textbf{Notation.} We use boldface letters for vectors and matrices, normal font for real numbers, and, generally, uppercase letters for random and lowercase letters for deterministic quantities. For any number $n\in\mathbb{N}$, we write $[n] = \{1,...,n\}$. For any vector $\x\in\R^n$ and $p\ge 0$, we write $\|\x\|_p = (\sum_{i=1}^n|x_i|^p)^{1/p}$, which is the standard $\ell_p$-norm for $p\ge 1$. For a random variable $X$, we write $\|X\|_{\psi_1} = \sup_{p\ge 1}p^{-1}(\E[|X|^p])^{1/p}$ for its sub-exponential norm. We denote by $\mathcal{S}=\{i\in [n]:x^\star_i\neq 0\}$ the support of $\x^\star$. For any $\x\in\R^n$ and subset of coordinates $\mathcal{C}\subseteq [n]$, we write $\x_\mathcal{C} = (x_i)_{i\in \mathcal{C}}\in\R^{|\mathcal{C}|}$, and, similarly for $i\in [n]$, we write $\x_{-i} = (x_j)_{j\in[n]\backslash\{i\}}\in \R^{n-1}$. The notation $f(n) \lesssim g(n)$ (resp.\ $f(n) \gtrsim g(n)$, $f(n)\asymp g(n)$) means that there are constants $c_1, c_2>0$ such that $f(n) \le c_1g(n)$ (resp.\ $f(n) \ge c_2g(n)$, $c_2g(n)\le f(n) \le c_1g(n)$).

%% file: sparse_phase_retrieval.tex
\section{Sparse phase retrieval}
\label{section:sparse_phase_retrieval}
The goal in sparse phase retrieval is to reconstruct an unknown $k$-sparse signal vector $\x^\star\in\R^n$ from a set of (possibly noisy) quadratic measurements 
$Y_j = (\A_j^\top\x^\star)^2 + \varepsilon_j$,
$j=1,...,m.$
We consider the standard setting where the sensing vectors $\A_j \sim \gauss(0,\mathbf{I}_n)$ are independent standard Gaussian vectors, and the noise terms $\varepsilon_j$ are independent centered sub-exponential random variables with maximum sub-exponential norm $\sigma\ge 0$, that is $\sigma = \max_{j\in [m]}\|\varepsilon_j\|_{\psi_1}$.\footnote{For the clarity of the analytical results, we focus on noisy sparse phase retrieval with real signal and real Gaussian measurement vectors. Our results extend naturally to the complex case.}

We follow the well-established approach to estimating the signal $\x^\star$ based on non-convex optimization \cite{CLM16, WZGAC18, YWW19, ZWGC18}. Given observations $\{\A_j,Y_j\}_{j=1}^m$, we consider the following (non-convex) unregularized empirical risk minimization problem:
\begin{equation}\label{eq:objective_function}
F(\x) = \frac{1}{4m}\sum_{j=1}^m\big((\A_j^\top\x)^2 - Y_j\big)^2.
\end{equation}
It is worth mentioning that a risk function based on amplitude measurements $|\A_j^\top\x^\star|$ has also been considered \cite{WZGAC18, ZWGC18}. However, the objective function becomes non-smooth in that case, and an analysis via the mirror descent framework appears more challenging.

Without any restrictions, the non-convex function $F$ in (\ref{eq:objective_function}) could potentially have exponentially many local minimizers and saddle points if $m\lesssim n$, cf.\ \cite{SQW18}.
In the noiseless case, it has been shown in \cite{LV13} that $m\ge 4k-1$ Gaussian measurements suffice for $\{\pm\x^\star\}$ to be the sparsest minimizer of $F$ with high probability, that is 
\begin{equation}
\label{eq:problem}
\{\pm \x^\star\}=\operatorname{argmin}_{\x:\|\x\|_0\le k} F(\x).
\end{equation}
However, solving (\ref{eq:problem}) is challenging due to the non-convexity of the objective $F$ and the combinatorial nature of the sparsity constraint $\|\x\|_0\le k$.
Existing gradient-based methods such as thresholded Wirtinger flow (TWF) \cite{CLM16}, sparse truncated amplitude flow (SPARTA) \cite{WZGAC18}, compressive reweighted amplitude flow (CRAF) \cite{ZWGC18} and sparse Wirtinger flow (SWF) \cite{YWW19} minimize the non-convex empirical risk by employing a non-trivial diagonal thresholding, orthogonality-promoting or spectral initialization scheme, which produces an initial estimate close to the signal (up to a global sign) $\{\pm\x^\star\}$, inside the so-called basin of attraction. This is followed by a refinement procedure via thresholded gradient descent iterations, where the thresholding steps are necessary to enforce sparsity of the iterates.

Numerous other algorithms have been developed to solve sparse phase retrieval, which rely on a strategy other than thresholding steps to exploit sparsity. One such strategy is to include an $\ell_1$-regularization term to the objective function, which is commonly found in convex relaxation-based approaches such as compressive phase retrieval via lifting (CPRL) \citep{OYDS12} and SparsePhaseMax \citep{HV16}, but also in the generalized message passing algorithm PR-GAMP \cite{SR15}, which includes a sparsity promoting prior. Alternatively, one can directly restrict the search to $k$-sparse vectors, either with a preliminary support recovery step as in the alternating minimization algorithm SparseAltMinPhase \cite{NJS15}, or by updating the estimated support in every iteration of the search in a greedy fashion as in GESPAR \cite{SBE14}.   

Hadamard Wirtinger flow (HWF) \cite{WR20} is an algorithm that performs gradient descent on the unregularized empirical risk (\ref{eq:objective_function}) using the Hadamard parametrization, and it can be recovered as a first-order approximation to mirror descent equipped with the hyperbolic entropy mirror map \cite{VKR20, WR20b}. 
HWF has been empirically studied in \cite{WR20}, and in some settings it has been shown to exhibit favorable sample complexities compared to other gradient-based approaches to sparse phase retrieval in numerical simulations.

%% file: mirror_descent_algorithms.tex
\section{Mirror descent algorithms}
\label{section:the_algorithm}
Our proposed approach to solving noisy sparse phase retrieval consists of running unconstrained mirror descent equipped with the hyperbolic entropy mirror map to minimize the empirical risk (\ref{eq:objective_function}). In particular, when initialized as described below, we require neither an added regularization term nor thresholding steps to enforce sparsity of the estimates.
We begin by giving a brief overview of the main quantities and definitions for mirror descent. The key object defining the geometry of mirror descent is the \textit{mirror map}.
\begin{definition}
Let $\mathcal{D}\subseteq \R^n$ be a convex open set. We say that $\Phi:\mathcal{D}\rightarrow \R$ is a mirror map if it is strictly convex, differentiable and $\{\nabla \Phi(\x):\x\in\mathcal{D}\} = \R^n$.
\end{definition}
A central quantity in the analysis of mirror descent is the \textit{Bregman divergence} associated to a mirror map $\Phi$, which is a measure of distance between two points $\x, \mathbf{y}\in \mathcal{D}$ given by 
\begin{equation*}
D_{\Phi}(\x, \mathbf{y}) = \Phi(\x) - \Phi(\mathbf{y}) - \nabla \Phi(\mathbf{y})^\top(\x-\mathbf{y}).
\end{equation*}
As a special case, we obtain the squared Euclidean distance by choosing the mirror map $\Phi(\x) = \|\x\|_2^2$, in which case mirror descent coincides with standard gradient descent. 

As we consider unconstrained mirror descent, we have $\mathcal{D} = \R^n$. Let $F:\R^n\rightarrow \R$ be the objective function we seek to minimize. The trajectory of mirror descent is characterized by the mirror map $\Phi$ and an initial point $\X(0)$, and is, in continuous time, defined by the ordinary differential equation \cite{NY83}
\begin{equation}
\label{eq:mirror_descent_continuous}
\frac{d}{dt}\X(t) = -\big(\nabla^2\Phi(\X(t))\big)^{-1}\nabla F(\X(t)).
\end{equation}
In discrete time, the mirror descent update is given by 
\begin{equation}
\label{eq:mirror_descent_discrete}
\nabla\Phi (\X^{t+1}) = \nabla\Phi (\X^t) - \eta \nabla F(\X^t),
\end{equation}
where $\eta>0$ is the step-size. Throughout, we write $\X(t)$ and $\X^t$ for the iterates of mirror descent in continuous and discrete time, respectively.

For the mirror map, we choose the hyperbolic entropy mirror map \cite{GHS20} given by
\begin{equation}
\label{eq:mirror_map}
\Phi(\x) = \sum_{i=1}^n\biggl(x_i \operatorname{arcsinh}\Bigl(\frac{x_i}{\beta}\Bigr) - \sqrt{x_i^2 + \beta^2}\biggr),
\end{equation}
for some parameter $\beta>0$. We provide a discussion on the choice of the parameter $\beta$ in Section \ref{section:main_results}. This choice of mirror map is motivated by the fact that, equipped with this mirror map, mirror descent is approximated by gradient descent with Hadamard parametrization \cite{GHS20, VKR20}, which has been studied in the recovery of low-rank structures in sparse recovery \cite{H17, VKR19, ZYH19} and matrix factorization \cite{ACHL19, GWBNS17, LMZ18}.

For the initialization, we follow the approach outlined in \cite{WR20}, which estimates a single coordinate on the support of the signal $\x^\star$ and initializes all other coordinates to zero. 
\begin{align}
\label{eq:initialization}
X_i(0) \equiv X_i^0 = \begin{cases}
\frac{1}{\sqrt{3}} \sqrt{\frac{1}{m}\sum_{j=1}^mY_j} \enskip & i=I_0\\
0 & i\neq I_0
\end{cases},
\qquad
I_0 \in \operatorname{arg}\max_i\Biggl\{\frac{1}{m}\sum_{j=1}^mY_jA_{ji}^2\Biggr\}.
\end{align}
The term $(\frac{1}{m}\sum_{j=1}^mY_j)^{1/2}$ is an estimate of the magnitude $\|\x^\star\|_2$, see e.g.\ \cite{CLS15, WGE17}, and Lemma 1 of \cite{WR20} guarantees in the noiseless case that $x^\star_{I_0}\neq 0$ with high probability, provided $m\gtrsim (x^\star_{max})^{-2}\max\{k\log n,\, \log^3n\}$, where $x^\star_{max} = \max_i|x^\star_i|/\|\x^\star\|_2$. The following is an immediate extension of Lemma 1 of \cite{WR20} to the noisy observation model.
\begin{lemma}
\label{lemma}
\begin{sloppypar}
Let $\x^\star\in\R^n$ be any $k$-sparse vector, and let the observations $\{Y_j, \A_j\}_{j=1}^m$ be given as in (\ref{eq:obs_model}), where $\A_j\sim\gauss(\mathbf{0}, \mathbf{I}_n)$ i.i.d.\ and $\{\varepsilon_j\}_{j=1}^m$ are independent centered sub-exponential random variables with maximum sub-exponential norm $\sigma = \max_{j}\|\varepsilon_j\|_{\psi_1}$. Let $I_0 \in \operatorname{arg}\max_i \{\frac{1}{m}\sum_{j=1}^mY_jA_{ji}^2\}$. There exist universal constants $c_s,c_p>0$ such that the following holds. If $m\ge c_s(1 + \frac{\sigma^2}{\|\x^\star\|_2^4})(x^\star_{max})^{-2}\max\{k\log n,\, \log^3n\}$,  then, with probability at least $1-c_pn^{-10}$ we have $|x^\star_{I_0}|\ge \frac{1}{2}x^\star_{max}\|\x^\star\|_2$.
\end{sloppypar}
\end{lemma}
\begin{proof}
By Lemma \ref{lemma:tech3} in the Appendix \ref{supp:technical_lemmas}, $|\frac{1}{m}\sum_{j=1}^m\varepsilon_jA_{ji}^2|\le c\sigma \sqrt{\frac{\log n}{m}}$ holds with probability $1-c_pn^{-10}$, where $c>0$ is an absolute constant. 
The rest of the proof follows the same steps as the proof of Lemma 1 in \cite{WR20}, and we omit the details. 
\end{proof}

%% file: main_results.tex
\section{Main results}
\label{section:main_results}
We show that, with high probability, early-stopped mirror descent recovers any $k$-sparse signal $\x^\star\in\R^n$ with $x^\star_{min} \gtrsim 1/\sqrt{k}$ from $(1 + \frac{\sigma^2}{\|\x^\star\|_2^4})k^2$ (modulo logarithmic term) Gaussian measurements, where we write $x^\star_{min} = \min_{i:x^\star_i\neq 0} |x^\star_i|/\|\x^\star\|_2$.

We begin by characterizing the relationship between the Bregman divergence $D_{\Phi}(\x^\star,\x)$ associated to the hyperbolic entropy mirror map $\Phi$ in (\ref{eq:mirror_map}) and the $\ell_2$-norm $\|\x-\x^\star\|_2$.
\begin{lemma}[\normalfont{\cite{WR20b}}] \label{lemma2} 
Let $\x^\star\in \R^n$ be any $k$-sparse vector with  $x^\star_{min} \ge c_\star/\sqrt{k}$ for a constant $c_\star>0$. Let $\mathcal{S}=\{i\in [n] :x^\star_i\neq 0\}$ be its support, and let $\Phi$ be as in (\ref{eq:mirror_map}) with parameter $\beta>0$.
\begin{itemize}
\item For any vector $\x\in \R^n$, we have
\begin{equation}\label{eq:lemma2lb}
\|\x-\x^\star\|_2^2 \le 2\sqrt{\max\{\|\x\|_{\infty}^2,\|\x^\star\|_\infty^2\}+\beta^2} \, D_{\Phi}(\x^\star,\x).
\end{equation}
\item Let $\x\in \R^n$ be any vector with $x_ix^\star_i\ge 0$ (no mismatched sign) and $|x_i| \ge \frac{1}{2}|x^\star_i|$ for all $i=1,...,n$. Then, we have
\begin{equation}\label{eq:lemma2ub}
D_{\Phi}(\x^\star, \x) \le \frac{\sqrt{k}}{c_\star\|\x^\star\|_2}\|\x_{\mathcal{S}}-\x^\star_{\mathcal{S}}\|_2^2 + \|\x_{\mathcal{S}^c}\|_1.
\end{equation}
\end{itemize}
\end{lemma}
A proof of Lemma \ref{lemma2} appears in the appendix of \cite{wu2020continuoustime}. For completeness, we include a proof of Lemma \ref{lemma2} in Section \ref{proof:lemma2} below. Lemma \ref{lemma2} implies that, if we are interested in convergence with respect to the $\ell_2$-norm, we can consider the Bregman divergence $D_{\Phi}$ as a proxy for it. As it is impossible to distinguish $\x^\star$ from $-\x^\star$ using phaseless measurements, we consider the $\ell_2$-distance and Bregman divergence from the solution set $\{\pm \x^\star\}$ given by $\operatorname{dist}(\x^\star, \x) = \min\{\|\x-\x^\star\|_2, \|\x+\x^\star\|_2\}$ and $\operatorname{dist}_{\Phi}(\x^\star,\x) = \min\{D_{\Phi}(\x^\star,\x), D_{\Phi}(-\x^\star,\x) \}$, respectively. Applying the bound in (\ref{eq:lemma2lb}) to both $\x^\star$ and $-\x^\star$, we find the upper bound 
\begin{equation*}
\operatorname{dist}(\x^\star,\x)^2 \le  2\sqrt{\max\{\|\x\|_{\infty}^2,\|\x^\star\|_\infty^2\}+\beta^2}\, \operatorname{dist}_{\Phi}(\x^\star,\x).
\end{equation*}

The setting for our convergence results is summarized in the following assumption.
\begin{assumption}\label{assumption}
For some universal constants $c_\star,c_s>0$, the following holds.
The signal $\x^\star\in \R^n$ is $k$-sparse with $x^\star_{min} \ge c_\star/\sqrt{k}$. The observations $\{Y_j, \A_j\}_{j=1}^m$ are i.i.d.\ given as in (\ref{eq:obs_model}), with $\A_j\sim \gauss(\mathbf{0}, \mathbf{I}_n)$ and $\varepsilon_j$ centered sub-exponential with maximum sub-exponential norm $\sigma = \max_{j}\|\varepsilon_j\|_{\psi_1}$. The sample size is at least $m\ge c_s(1 + \frac{\sigma^2}{\|\x^\star\|_2^4})\max\{k^2\log^2n,\;\log^5n\}$.
\end{assumption}
We now formulate our results for continuous-time mirror descent in terms of the Bregman divergence $D_{\Phi}$. Lemma \ref{lemma2} allows to translate the upper bound on $\operatorname{dist}_{\Phi}(\x^\star, \X(t))$ into an upper bound on the $\ell_2$-distance $\operatorname{dist}(\x^\star,\X(t))$. Recall that we denote the support of the signal $\x^\star$ by $\mathcal{S}=\{i\in [n]:x^\star_i\neq 0\}$ .
\begin{theorem}\label{theorem}
Let Assumption \ref{assumption} hold. There exist universal constants $c_p,c,c_1,c_2>0$ such that the following holds. Let $\X(t)$ be defined by the continuous-time mirror descent algorithm given by (\ref{eq:mirror_descent_continuous}) with mirror map (\ref{eq:mirror_map}) and initialization (\ref{eq:initialization}) with $\beta\le c_1 \|\x^\star\|_2/n^3$. Let $\delta = \sqrt{n\beta/\|\x^\star\|_2}$, and define
\begin{equation*}
T_2 = \inf\biggl\{t>0: \frac{\operatorname{dist}(\x^\star, \X(t))}{\|\x^\star\|_2} \le c\max\biggl\{\sqrt{c_\star\sqrt{k}\delta}, \; \frac{\sigma}{\|\x^\star\|_2^2} \sqrt{\frac{k\log n}{m}} \biggr\}\biggr\}.
\end{equation*}
Then, there exists $T_1\le \frac{c_2}{\|\x^\star\|_2^3}k\log(\frac{\|\x^\star\|_2}{\beta}) \log (k\log \frac{\|\x^\star\|_2}{\beta})$ such that
\begin{equation}
\label{thm:convergence}
\frac{\operatorname{dist}_{\Phi}(\x^\star,\X(t))}{\|\x^\star\|_2} \le \frac{6\sqrt{k}}{c_\star} \exp\biggl(-\frac{c_\star\|\x^\star\|_2^3}{4\sqrt{k}} (t - T_1)\biggr) \quad \text{for all } T_1\le t\le T_2,
\end{equation}
with probability at least $1-c_pn^{-10}$.
Furthermore, for all $t\le T_2$, we have
\begin{equation}
\label{thm:off_support}
\|\X_{\mathcal{S}^c}(t)\|_1 \le \delta \|\x^\star\|_2.
\end{equation}
\end{theorem}

Analogous results also hold in discrete time.
\begin{theorem}\label{theorem_discrete}
Let Assumption \ref{assumption} hold. There exist universal constants $c_p,c,c_1,c_2,c_3>0$ such that the following holds. Let $\X^t$ be defined by the discrete-time mirror descent algorithm given by (\ref{eq:mirror_descent_discrete}) with mirror map (\ref{eq:mirror_map}) and initialization (\ref{eq:initialization}) with $\beta\le c_1 \|\x^\star\|_2/n^3$ and $\eta\le c_3/\|\x^\star\|_2^3$. Let $\delta = \sqrt{n\beta/\|\x^\star\|_2}$, and define
\begin{equation*}
T_2 = \inf\biggl\{t>0: \frac{\operatorname{dist}(\x^\star, \X^t)}{\|\x^\star\|_2} \le c\max\biggl\{\sqrt{c_\star\sqrt{k}\delta}, \; \frac{\sigma}{\|\x^\star\|_2^2} \sqrt{\frac{k\log n}{m}} \biggr\}\biggr\}.
\end{equation*}
Then, there exists $T_1\le \frac{c_2}{\eta \|\x^\star\|_2^3}k\log(\frac{\|\x^\star\|_2}{\beta}) \log (k\log \frac{\|\x^\star\|_2}{\beta})$ such that
\begin{equation}
\label{thm:convergence_discrete}
\frac{\operatorname{dist}_{\Phi}(\x^\star,\X^t)}{\|\x^\star\|_2} \le \frac{6\sqrt{k}}{c_\star} \biggl(1 - \frac{c_\star\eta\|\x^\star\|_2^3}{8\sqrt{k}}\biggr)^{t-T_1} \quad \text{for all } T_1\le t\le T_2,
\end{equation}
with probability at least $1-c_pn^{-10}$.
Furthermore, for all $t\le T_2$, we have
\begin{equation}
\label{thm:off_support_discrete}
\|\X_{\mathcal{S}^c}^t\|_1 \le \delta \|\x^\star\|_2.
\end{equation}
\end{theorem}
In the noiseless case, when $\sigma = 0$, Theorem \ref{theorem} recovers the results in \cite{WR20b} and Theorem \ref{theorem_discrete} shows that any desired relative error $\varepsilon > 0$ can be achieved by running mirror descent with the choice $\beta \asymp \varepsilon^2\|\x^\star\|_2 / n^3$ for $k\log \frac{\|\x^\star\|_2}{\beta}$ (modulo double-logarithmic term) iterations. 

We now discuss some of the main implications of Theorem \ref{theorem} and Theorem \ref{theorem_discrete}. For notational simplicity, we focus on discrete-time mirror descent throughout the remainder of this section, but the same considerations also apply to the continuous-time algorithm.

\textbf{Nearly minimax-optimal rate}. 
Theorem \ref{theorem_discrete} implies that mirror descent achieves a relative mean-squared error $\frac{\operatorname{dist}(\X^{T_2}, \x^\star)}{\|\x^\star\|_2}$ on the order of $\frac{\sigma}{\|\x^\star\|_2^2} \sqrt{\frac{k\log n}{m}}$, provided the mirror map parameter $\beta$ is chosen smaller than $\sigma^4k\log^2 n / (m^2n\|\x^\star\|_2^7)$. This is identical to the estimation error rate achieved by thresholded Wirtinger flow \cite{CLM16}, and nearly matches the lower bound on the minimax-optimal rate which we derive below in Theorem \ref{thm:minimax}. As we only consider $k$-sparse signals with minimum non-zero entry on the order of $\|\x^\star\|_2/\sqrt{k}$, our lower bound on the minimax rate of convergence differs by a factor $\sqrt{\log (en/k)}$ from the bound of $\frac{\sigma}{\|\x^\star\|_2^2} \sqrt{\frac{k\log (en/k)}{m}}$ for $k$-sparse vectors, which has been stated as Theorem 3.2 in \cite{CLM16}.
\begin{theorem}
\label{thm:minimax}
\hspace{-0.5em} Let $\Theta (k,n,r) = \{\x\in\R^n: \|\x\|_2 = r, \|\x\|_0=k, \min_{i:x_i\neq 0}|x_i|\ge c_\star r/\sqrt{k}\}$ for some constant $c_\star>0$. Let $\{\A_{j}\}_{j=1}^m$ be a collection of i.i.d.\ $\gauss (0,\mathbf{I}_n)$ random vectors, $\{\varepsilon_j\}_{j=1}^m$ be a collection of i.i.d.\ $\gauss (0,\sigma^2)$ random variables, and let $\{Y_j\}_{j=1}^m$ be given as in (\ref{eq:obs_model}). There exist universal constants $c_s,c>0$ such that if $m\ge c_s(1 + \frac{\sigma^2}{R^4})k\log (en/k)$, then
\begin{equation}
\label{eq:minimax_rate}
\inf_{\widehat{\x}}\sup_{\x^\star\in \Theta (k,n,r)} \P_{(\A,\mathbf{Y}|\x^\star)}\biggl[\frac{\operatorname{dist}(\x^\star, \widehat{\x})}{r} \ge \frac{c\sigma}{r^2}\sqrt{\frac{k}{m}} \biggr] \ge \frac{1}{5},
\end{equation}
where the infimum is taken over all procedures $\mathbf{\widehat{x}}(\A,\mathbf{Y})$.
\end{theorem}
\begin{proof}
Theorem \ref{thm:minimax} follows from an application of Theorem C in \cite{LM15}. We first introduce some notation used for Theorem C. For any vector $\x\in\R^n$ and $r>0$, we write $\mathcal{B}_2(\x,\rho)\subseteq\R^n$ for the $\ell_2$-ball with center $\x$ and radius $\rho$. For any set $\mathcal{X}\in \R^n$ and $\rho>0$, we denote the packing number with respect to the $\ell_2$-norm by
\begin{equation*}
M(\mathcal{X},\rho) = \sup\{n: \exists \x_1,...,\x_n\in \mathcal{X} \text{ s.t } \forall i\neq j, \|\x_i-\x_j\|_2>\rho\}.
\end{equation*}
To apply Theorem C, we need to bound the packing number of $\Theta(k,n,r)\cap \mathcal{B}_2(\x_0, c_0\rho)$, where $\x_0\in\Theta(k,n,r)$ is an arbitrary vector. More precisely, we will show that, for all $\rho>0$ and $c_0\ge 2$,
\begin{equation}
\label{eq:packing_number_estimate}
\sup_{\x_0\in \Theta(k,n,r)}\log^{1/2}M(\Theta(k,n,r)\cap \mathcal{B}_2(\x_0, c_0\rho), \rho) \asymp \sqrt{k}.
\end{equation}
Theorem \ref{thm:minimax} then follows from an application of Theorem C of \cite{LM15} (c.f.\ Corollary 6.2 in \cite{LM15}).

In the minimax rate given in Corollary 6.2 in \cite{LM15} and Theorem 3.2 in \cite{CLM16}, the set $\Theta(n,k,r)$ is defined without the restriction $\min_{i:x_i\neq 0}|x_i|>c_\star r/\sqrt{k}$. In that case, standard estimates on the packing number yield 
\begin{equation*}
\log^{1/2} M(\widetilde{\Theta}(k,n,r)\cap \mathcal{B}_2(\x_0, c_0\rho), \rho) \asymp \sqrt{k\log (en/k)},
\end{equation*}
where we write $\widetilde{\Theta}(k,n,r) = \{\x\in\R^n: \|\x\|_2 = r, \|\x\|_0=k\}$, see e.g.\ Lemma 1.4.2 and Lemma 2.2.17 in \cite{CGLP12}). On the other hand, with the restriction $\min_{i:x_i\neq 0}|x_i|>c_\star r/\sqrt{k}$, when we fix any $\x_0\in \Theta(k,n,r)$, then any vector $\x\in \Theta(k,n,r)\cap \mathcal{B}(\x_0, c_0\rho)$ must be supported on the support of $\x_0$ for $\rho<\sqrt{2}c_\star r/(c_0\sqrt{k})$, because any non-zero coordinate of $\x_0$ and $\x$ has absolute value at least $c_\star r/\sqrt{k}$, and therefore $\|\x-\x_0\|_2\ge \sqrt{2}c_\star r/\sqrt{k}$ if $\x$ and $\x_0$ have different supports. This leads to a packing number that scales exponentially in $k$ (as opposed to a scaling of $\binom{n}{k}$ in the case of $\widetilde{\Theta}(n,k,r)$), and therefore the estimate in (\ref{eq:packing_number_estimate}).
\end{proof}
 
\textbf{On early stopping and sparsity.} 
The bound (\ref{thm:off_support_discrete}) in Theorem \ref{theorem_discrete} controls the magnitude off-support coordinates can attain while mirror descent runs, up to a time $T_2$. To show the bound (\ref{thm:off_support_discrete}), we crucially use the assumption $x^\star_{min}\gtrsim 1/\sqrt{k}$, which is needed to show that all support coordinates grow faster than any off-support coordinate, which then leads to $\|\X^t_{\S^c}\|_1<\delta\|\x^\star\|_2$ for a sufficiently small $\delta$; see Section \ref{subsection:proof_ideas} for an outline of the proof idea, and Stage (i), part (b) in Section \ref{proof:theorem3} for details. The fact that mirror descent iterates stay approximately $k$-sparse allows us to show bound (\ref{thm:convergence_discrete}): after an initial ``warm-up'' period of length $T_1$, mirror descent converges linearly to the solution set $\{\pm\x^\star\}$ up to a time $T_2$, at which a precision determined by the mirror map parameter $\beta$ and the quantity $\frac{\sigma}{\|\x^\star\|_2^2}\sqrt{\frac{k\log n}{m}}$ is reached. The necessity for early stopping is explained by the fact that, as elaborated in the proof sketch below, we need to control the size of off-support variables $\|\X_{\S^c}^t\|_1$, which is guaranteed to be sufficiently small up to the time $T_2$. 

Using the fact that the Bregman divergence decreases linearly for $T_1\le t\le T_2$, a quick calculations shows that $T_2-T_1 \lesssim \sqrt{k}\log \frac{k\|\x^\star\|_2}{\beta}$, see Stage (ii), part (a) in Section \ref{proof:theorem3} for details. Therefore, an upper bound for $T_2$ is given by $k\log\frac{\|\x^\star\|_2}{\beta}$ (modulo double-logarithmic term). In applications, we can use a data-dependent stopping rule such as the hold-out method outlined in \cite{RWY14} to estimate $T_2$: we can run mirror descent using a fraction $\alpha\in (0,1)$ of the data and return the iterate $\X^t$ which minimizes the empirical risk $F$ evaluated on the remaining $(1-\alpha)$ of the data. Section \ref{section:numerical_simulation} presents a numerical study that also investigates this data-dependent stopping rule.\footnote[1]{Establishing theoretical guarantees for data-dependent stopping rules is outside the scope of our contribution.}

\textbf{Role of the mirror map parameter.}
In Theorem \ref{theorem_discrete}, choosing a small parameter $\beta$ is needed to guarantee that off-support variables stay sufficiently small in the bound (\ref{thm:off_support_discrete}). The quantity $\delta$ in the bound (\ref{thm:off_support_discrete}) depends polynomially on $\beta$, while choosing a small $\beta$ leads to the length of the initial warm-up period $T_1$ increasing as $\log\frac{\|\x^\star\|_2}{\beta}$.
In practice, we can therefore simply choose a very small $\beta$ (e.g.\ $10^{-20}$), so that convergence is linear up to a precision on the order of $\frac{\sigma}{\|\x^\star\|_2^2} \sqrt{\frac{k\log n}{m}}$. A similar trade-off between computational cost and statistical accuracy with respect to the size of the initialization has previously been observed in \cite{VKR19} in the setting of noisy sparse linear regression.

\subsection{Proof ideas}
\label{subsection:proof_ideas}
We give a high-level overview of the main ideas behind the proofs of Theorem \ref{theorem} and Theorem \ref{theorem_discrete}. We begin with the proof of Theorem \ref{theorem}. 
Without loss of generality, let us assume that the initialization in (\ref{eq:initialization}) satisfies $x^\star_{I_0}>0$, so that $\X(0)^\top\x^\star>0$. In this case, we will show that $\X(t)$ converges to $\x^\star$. Otherwise, we can show convergence to $-\x^\star$.

Employing a potential-based analysis in terms of the Bregman divergence, we will bound the Bregman divergence $D_\Phi(\x^\star,\X(t))$, which then also yields a bound on the $\ell_2$-distance $\|\x^\star -  \X(t)\|_2$ via Lemma \ref{lemma2}. A quick calculation yields
\begin{equation}
\label{eq:breg_derivative}
\frac{d}{dt}D_{\Phi}(\x',\X(t)) = -\big\langle\nabla F(\X(t)), \X(t) - \x' \big\rangle,
\end{equation}
where $\x'\in\R^n$ is any reference point. If $F$ were to be convex, then choosing $\x'$ in equation (\ref{eq:breg_derivative}) to be any global minimizer of $F$ would imply
\begin{equation*}
\big\langle\nabla F(\X(t)), \X(t) - \x' \big\rangle \ge F(\X(t)) - F(\x'),
\end{equation*}
which shows that continuous-time mirror descent strictly decreases the Bregman divergence as long as $F(\X(t)) > F(\x')$. As the Bregman divergence is bounded from below by zero by construction due to the convexity of the mirror map, this implies that $F(\X(t))$ must converge to $\min_\x F(\x)$, which in turn implies that $\X(t)$ converges to a global minimizer of $F$, provided the function $F$ is continuous. For non-convex objective functions, the inner product in (\ref{eq:breg_derivative}) has been used to define the notion of \emph{variational coherence} \cite{ZMBBG20}, namely the assumption that
\begin{equation}\label{eq:variational_coherence}
\big\langle \nabla F(\x),\x-\x'\big\rangle \ge 0 \qquad \text{for all}\enskip \x\in\R^n,\; \x'\in\mathcal{X}^\star,
\end{equation}
where $\mathcal{X}^\star = \operatorname{arg}\min_{\mathbf{z}}F(\mathbf{z})$, and that there exists an $\x'\in \mathcal{X}^\star$ such that $\langle \nabla F(\x),\x-\x'\rangle = 0$ implies that also also $\x\in \mathcal{X}^\star$.
Under this assumption, convergence of a stochastic version of mirror descent towards the set of minimizers of $F$ has been shown in \cite{ZMBBG20}. In the batch setting, convergence of mirror descent towards a global minimizer of $F$ can be shown as in the convex case under the assumption of variational coherence. A local version of variational coherence, only assuming that inequality (\ref{eq:variational_coherence}) holds for a local minimizer $\x'$ and for $\x$ belonging to an open neighborhood of $\x'$, was also considered in \cite{ZMBBG20}, yielding a local version of the convergence results.

In order to study convergence rates, a more refined notion, namely quantitative control on a strictly positive lower bound in (\ref{eq:variational_coherence}), would be necessary. For example, if the function $F$ were to satisfy a lower bound of the form
\begin{equation*}
\big\langle \nabla F(\\X(t)),\X(t)-\x'\big\rangle \ge cD_\Phi(\x', \X(t))
\end{equation*}
for some constant $c>0$ and minimizer $\x'\in \mathcal{X}^\star$, then this would imply exponential convergence towards $\x'$, since we have
\begin{equation*}
\frac{d}{dt}D_\Phi(\x', \X(t)) \le -cD_\Phi(\x',\X(t)) \quad \Rightarrow \quad D_\Phi(\x',\X(t))\le e^{-ct}D_\Phi(\x',\X(t)).
\end{equation*}
By design, both the global and local version of variational coherence can only be used to establish convergence towards local minimizers of the objective function. As in our case we are interested in establishing convergence towards the statistical signal $\x^\star$, we consider the inner product $\langle\nabla F(\X(t)), \X(t) -\x^\star \rangle$ and bound this quantity from below within a region where the trajectory of mirror descent is confined to stay up to the prescribed stopping time, upon proper initialization. In order to obtain results on the rate of convergence, it is crucial to quantify the lower bound in (\ref{eq:variational_coherence}). In particular, we derive a lower bound of the form
\begin{equation}
\label{eq:proof_idea_bound}
\big\langle\nabla F(\X(t)), \X(t) -\x^\star \big\rangle \ge \begin{cases}
\frac{c_1}{\sqrt{k}} & \text{if } \|\X(t)\|_2^2\le \frac{2}{5}\|\x^\star\|_2^2 \\
c_2\|\X(t) - \x^\star\|_2^2 & \text{else}
\end{cases}
\end{equation}
for some constants $c_1, c_2 > 0$. The analysis is divided into two stages: (i) an initial warm-up period of length $T_1$, where the Bregman divergence decreases first at a constant rate $1/\sqrt{k}$, after which we can relate the quantity $\|\X(t)-\x^\star\|_2^2$ in the second bound in (\ref{eq:proof_idea_bound}) to the Bregman divergence $D_\Phi(\x^\star, \X(t))$ by counting the number of ``small'' coordinates for which $|X_i(t)| < \frac{1}{2}|x^\star_i|$, and (ii) a subsequent stage where we show linear convergence by relating the Bregman divergence to the quantity $\|\X(t)-\x^\star\|_2^2$ in the lower bound (\ref{eq:proof_idea_bound}) via Lemma \ref{lemma2}.

The proof of the lower bound (\ref{eq:proof_idea_bound}) involves three main steps.
In the first step, writing $f(\x) = \E[F(\x)]$ for the population risk, we assume that we have access to the population gradient, which, by the dominated convergence theorem, can be computed as
\begin{equation*}
\nabla f(\x) = \E[\nabla F(\x)] = \big(3\|\x\|_2^2 - \|\x^\star\|_2\big)\x - 2\big(\x^\top\x^\star\big)\x^\star.
\end{equation*}
In this setting, the lower bounds stated in (\ref{eq:proof_idea_bound}) for the inner product $\langle\nabla f(\X(t)), \X(t) - \x^\star \rangle$ can be obtained via algebraic manipulations: the first bound in (\ref{eq:proof_idea_bound}) relies on Lemma \ref{lemma:support3} and the assumption $x^\star_{min}\ge c_\star/\sqrt{k}$ to bound the inner product $\X(t)^\top\x^\star$ from below by $c_1/\sqrt{k}$, while the second bound only uses the Cauchy-Schwarz inequality.

In the second step, we need to control the deviation of the empirical quantity
\begin{equation}
\label{eq:deviation}
\langle\nabla F(\X(t)), \X(t) -\x^\star \rangle = \frac{1}{m}\sum_{j=1}^m\big((\A_j^\top\X(t))^2 - Y_j\big)\big(\A_j^\top\X(t)\big)\big(\A_j^\top(\X(t) - \x^\star)\big)
\end{equation}
from the population quantity $\langle\nabla f(\X(t)), \X(t) - \x^\star \rangle$. The deviation is bounded in Lemma \ref{lemma:support2}, and we summarize the main proof ideas in the following. 
If the iterate $\X(t)$ were to be independent from the measurement vectors $\{\A_j\}_{j=1}^m$ and noise terms $\{\varepsilon_j\}_{j=1}^m$, then, conditionally on $\X(t)$, the term in (\ref{eq:deviation}) would be a fourth order Gaussian polynomial, which can be controlled using standard concentration bounds. More precisely, we could use a truncation argument and split it into a bounded and an unbounded part, where the unbounded part is non-zero only with small probability, and use Lipschitz concentration and Hoeffding's bounds to control the bounded part, and Chebyshev's inequality for the unbounded part.
However, the assumption of independence is not reasonable. Instead, we use an $\epsilon$-net argument followed by above concentration argument for each fixed vector $\x$ in the $\epsilon$-net. If we considered an $\epsilon$-net in an unit ball in $\R^n$, our argument would require a sample size of $m\gtrsim nk$, with which plenty of algorithms provably recover the signal $\x^\star$ regardless of the assumption on sparsity, e.g.\ \cite{CLS15, CCG15, NJS15}. To obtain our sample requirement of $k^2$ (modulo logarithmic terms), we show that mirror descent exhibits algorithmic regularization in the sense that $\X(t)$ stays ``almost $k$-sparse'' (more precisely, we show that $\|\X_{\S^c}(t)\|_1<\delta\|\x^\star\|_2$), and apply a union bound to control the deviation of the term in (\ref{eq:deviation}) from its expectation for all bounded vectors $\x$ supported on $\S$.

In the third step, in order to show that off-support coordinates stay negligibly small, we consider the trajectory of a coordinate $X_i(t)$, which is defined by (cf.\ (\ref{eq:mirror_descent_continuous}))
\begin{equation*}
\frac{d}{dt}X_i(t) = -\sqrt{X_i(t)^2 + \beta^2}\; \nabla F(\X(t))_i.
\end{equation*}
If we imagined that $\beta=0$ and $\nabla F(\X(t))_i=-c$ was a constant, then the trajectory would simply be given by $X_i(t) = e^{ct}$.
Using the assumption $x^\star_{min}\ge c_\star/\sqrt{k}$ along with the concentration argument outlined above, we can show that, for any $i\in \S$ and $j\notin \S$, 
\begin{equation*}
|\nabla F(\X(t))_j| \le \gamma|\nabla F(\X(t))_i|
\end{equation*}
holds with high probability for some $\gamma<1$. Informally, the behavior of support and off-support coordinates is comparable to $\beta(1+c)^t$ and $\beta(1+\gamma c)^t$, respectively, and the gap can be made large by choosing $\beta$ small and $t$ large. A detailed proof can be found in Section \ref{proof:theorem3}.

The main ideas of the proof of Theorem \ref{theorem_discrete} are similar to those of Theorem \ref{theorem}. To show that the Bregman divergence is decreasing in discrete time, we need to bound the difference 
\begin{equation}
\label{eq:breg_difference}
D_{\Phi}(\x^\star, \X^{t+1}) - D_{\Phi}(\x^\star, \X^t) = -\eta \big\langle \nabla F(\X^t), \X^t-\x^\star\big\rangle + D_{\Phi}(\X^t, \X^{t+1}).
\end{equation}
The first term in (\ref{eq:breg_difference}) can be bounded
as in continuous time, while the second term $D_{\Phi}(\X^t, \X^{t+1})$ measures the distance between consecutive iterates in terms of the Bregman divergence and is due to the discretization and can be bounded using the strong convexity of the mirror map $\Phi$ in (\ref{eq:mirror_map}). Unlike in continuous time, we need to establish coordinate-wise upper and lower bounds for how much each coordinate $X_i^t$ changes in one iteration in the discrete-time case in order to characterize the region to which the trajectory of mirror descent is confined. This is not required in continuous time, since in that case we can use the continuity of $\X_i(t)$. The details of the proof can be found in Appendix \ref{supp:proof_theorem4}.

%% file: numerical_study.tex
\section{Numerical study}
\label{section:numerical_simulation}
In this section, we provide numerical experiments demonstrating how the minimum relative estimation error achieved after $t_{max}$ iterations, that is
\begin{equation}\label{eq:min_rel_err}
\min_{1\le t\le t_{max}} \frac{\operatorname{dist}(\x^\star, \X^t)}{\|\x^\star\|_2},
\end{equation}
depends on the noise-to-signal ratio $\sigma/\|\x^\star\|_2^2$, the sample size $m$ and the sparsity level $k$, and how the length of the warm-up period as defined in Section \ref{proof:theorem3}, that is
\begin{equation*}
T_1 = \inf\biggl\{t>0: \min_{i\in\S}\frac{|X_i^t|}{|x^\star_i|}> \frac{1}{2} \biggr\},
\end{equation*}
depends on the mirror map parameter $\beta$ and the sparsity level $k$. The main purpose of this section is to validate how the theoretical bounds in Theorem \ref{theorem_discrete} on the estimation error and $T_1$, as well as the convergence behavior suggested by Theorem \ref{theorem_discrete}, can be representative of the actual performance of mirror descent. 

The exponentiated gradient algorithm EG$\pm$ without normalization \cite{KW97} is an equivalent formulation of mirror descent equipped with the hyperbolic entropy map (Theorem 24 \cite{GHS20}), which is given by
\begin{equation}
\label{eq:update_eg}
\begin{gathered}
\X^t = \mathbf{U}^t - \mathbf{V}^t, \\
\mathbf{U}^{t+1} = \mathbf{U}^t\odot \exp\bigl(-\eta\nabla F(\X^t)\bigr), \quad \mathbf{V}^{t+1} = \mathbf{V}^t\odot \exp\bigl(\eta\nabla F(\X^t)\bigr),
\end{gathered}
\end{equation}
where $\odot$ denotes the elementwise Hadamard product, and the initialization (\ref{eq:initialization}) becomes
\begin{align*}
U_i(0) = \begin{cases} \frac{\hat{\theta}}{2\sqrt{3}} + \sqrt{\frac{\hat{\theta}^2}{12}+ \frac{\beta^2}{4}} \quad &i=I_0 \\ \frac{\beta}{2} & i\neq I_0 \end{cases}, \qquad V_i(0) = \begin{cases} -\frac{\hat{\theta}}{2\sqrt{3}} + \sqrt{\frac{\hat{\theta}^2}{12}+ \frac{\beta^2}{4}} \quad &i=I_0 \\ \frac{\beta}{2} & i\neq I_0 \end{cases},
\end{align*}
where we write $\hat{\theta} = (\sum_{j=1}^mY_j/m)^{1/2}$ for the estimate of the signal size $\|\x^\star\|_2$.
In our simulations, we found EG$\pm$ to be numerically more stable than the update (\ref{eq:mirror_descent_discrete}), especially for small values of $\beta$ due to the division by $\beta$ in (\ref{eq:mirror_descent_discrete}).

In all experiments, we generate a $k$-sparse vector $\x^\star\in\R^n$ by sampling $x^\star_i$ uniformly from $[-1,-0.15]\cup[0.15,1]$ for $i=1,\dots,n$ and setting $(n-k)$ random entries of $\x^\star$ to zero. We sample $m$ i.i.d.\ measurement vectors $\A_j\sim\gauss (0, \mathbf{I}_n)$ and noise terms $\varepsilon_j\sim\gauss (0, \sigma^2)$ (note that we have $\|\varepsilon_j\|_{\psi_1} = \sqrt{2/\pi}\sigma$), and generate observations $Y_j = (\A_j^\top\x^\star)^2 + \varepsilon_j$. For mirror descent, we find that the step size $\eta = 0.3 / (\sum_{j=1}^mY_j/m)^{3/2}$ works well, which is consistent with the value prescribed by Theorem 4.

\subsection{Estimation error}
First, we examine how the relative estimation error of early-stopped mirror descent depends on the noise-to-signal ratio $\sigma/\|\x^\star\|_2^2$, the sample size $m$ and the sparsity level $k$. Theorem \ref{theorem_discrete} provides an upper bound for the relative estimation error of order $\frac{\sigma}{\|\x^\star\|_2^2}\sqrt{\frac{k\log n}{m}}$, provided $\beta$ is chosen small enough so that $\sqrt{c_\star\sqrt{k}\delta} \le \frac{\sigma}{\|\x^\star\|_2^2} \sqrt{\frac{k\log n}{m}}$. In the following experiments, we fix $n=2000$ and $\beta = 10^{-20}$. We run mirror descent for $t_{max}=5000$ iterations and consider the minimum relative estimation error defined in (\ref{eq:min_rel_err}), averaged across $100$ independent Monte Carlo trials.

Evaluating the minimum relative error defined in (\ref{eq:min_rel_err}) requires oracle knowledge of the underlying signal $\|\x^\star\|_2$.
For the sake of completeness, in our simulations we also consider the hold-out method to implement early stopping, even if this data-dependent rule goes beyond the reach of the theoretical results that we have developed (in this, we follow the same approach considered in prior related literature \cite{RWY14}). The hold-out method only uses observed data to select a stopping time. More precisely, we run mirror descent using $90\%$ of the data, and select the stopping time $T_{stop}\in\{1,...,t_{max}\}$ for which the iterate $\X^{T_{stop}}$ minimizes the risk $F$ evaluated on the remaining $10\%$ of the data.
\begin{figure}[!t]
\centering
\includegraphics[width=\textwidth]{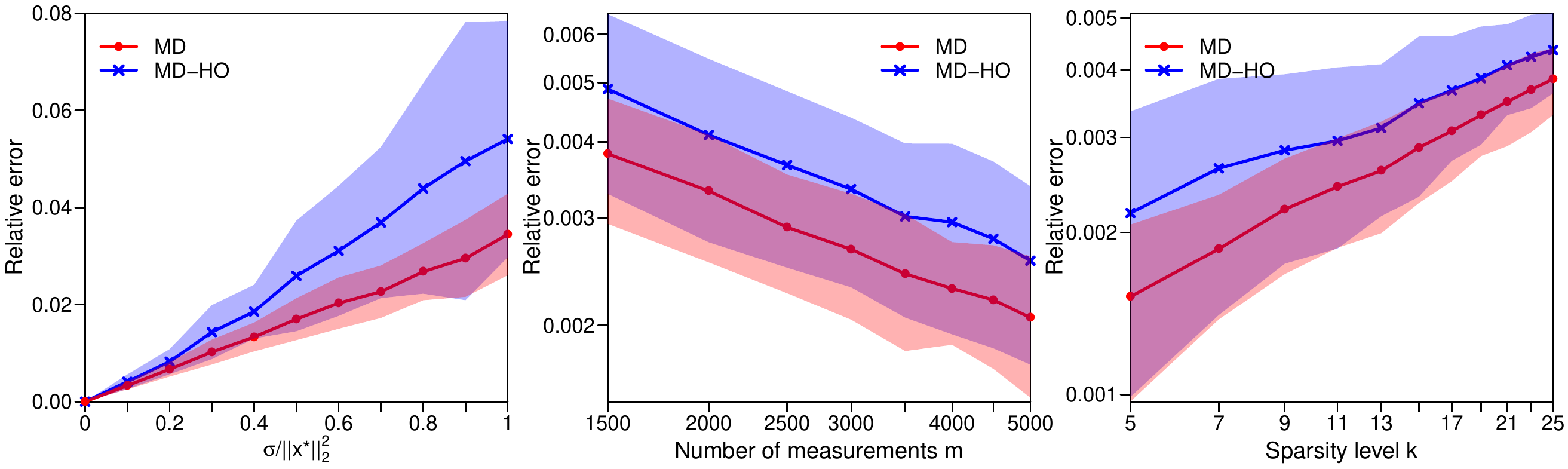}
\caption{Relative error $\operatorname{dist}(\x^\star, \X^{T_{stop}}) / \|\x^\star\|_2$ plus/minus one standard deviation (shaded area) of early-stopped mirror descent, with the stopping time $T_{stop}$ selected with oracle knowledge in red (circles), and using the hold-out method in blue (crosses). \textbf{Left:} relative error against noise-to-signal ratio $\sigma/\|\x^\star\|_2^2$, with $n=m=2000$, $k=10$. \textbf{Center:} log-log plot of relative error against sample size $m$, with $n=2000$, $k=10$, $\sigma/\|\x^\star\|_2^2 = 0.1$. \textbf{Right:} log-log plot of relative error against sparsity level $k$, with $n=2000$, $m=4000$, $\sigma/\|\x^\star\|_2^2 = 0.1$.}
\label{fig1}
\end{figure}

In the first experiment, we fix $k=10$, $m=2000$, and vary the noise-to-signal ratio $\sigma/\|\x^\star\|_2^2$. The left plot in Figure \ref{fig1} verifies the linear relationship between the relative estimation error and $\sigma/\|\x^\star\|_2^2$ suggested by the upper bound in Theorem \ref{theorem_discrete}. 

Next, we fix $k=10$, $\sigma/\|\x^\star\|_2^2 =0.1$, and increase $m$ from $1500$ to $5000$. The log-log plot (Figure \ref{fig1}, center) shows a slope of $-0.5137$, which validates the prediction of Theorem \ref{theorem_discrete} that the estimation error decreases as $1/\sqrt{m}$ as the sample size increases.

Finally, we fix $m=4000$, $\sigma/\|\x^\star\|_2^2 =0.1$, and vary the sparsity $k$. The log-log plot (Figure \ref{fig1}, right) shows a slope of $0.5775$, which, considering the standard deviation, is in line with the upper bound of Theorem \ref{theorem_discrete} which increases as $\sqrt{k}$ as the sparsity level increases. 

In all experiments, the relative error and its standard deviation are larger when the stopping time is selected using the hold-out method compared to a stopping time selected using oracle knowledge of the signal $\x^\star$, but the same trends are exhibited.

\subsection{Warm-up time \texorpdfstring{$T_1$}{}} We note that the asymptotic computational complexity of the algorithm is controlled by the growth of $T_1$, since $T_2-T_1 \lesssim \sqrt{k}\log \frac{k\|\x^\star\|_2}{\beta} \lesssim T_1$, see Stage (ii) part (a) in Section \ref{proof:theorem3}. As a result, we focus on $T_1$ in the following simulations, where we examine how the length of the warm-up period $T_1$ depends on the mirror map parameter $\beta$ and the sparsity level $k$. 
Theorem \ref{theorem_discrete} provides an upper bound of order $k\log\frac{\|\x^\star\|_2}{\beta}$ (modulo double-logarithmic terms).

We first fix $n=2000$, $m=1500$, $k=10$, $\sigma/\|\x^\star\|_2^2 = 0.1$ and decrease $\beta/\|\x^\star\|_2$ from $10^{-4}$ to $10^{-40}$. The left plot in Figure \ref{fig2} validates the upper bound in Theorem \ref{theorem_discrete}, the length of the warm-up period $T_1$ increases linearly with $\log\frac{\|\x^\star\|_2}{\beta}$. 

In the second experiment, we set $n=2000$, $m=4000$, $\sigma/\|\x^\star\|_2^2 = 0.1$, $\beta =10^{-20}$ and increase $k$ from $5$ to $25$. Again, the right plot in Figure \ref{fig2} shows that $T_1$ increases linearly with the sparsity level $k$, validating the upper bound in Theorem \ref{theorem_discrete}.
\begin{figure}[!t]
\centering
\includegraphics[width=\textwidth]{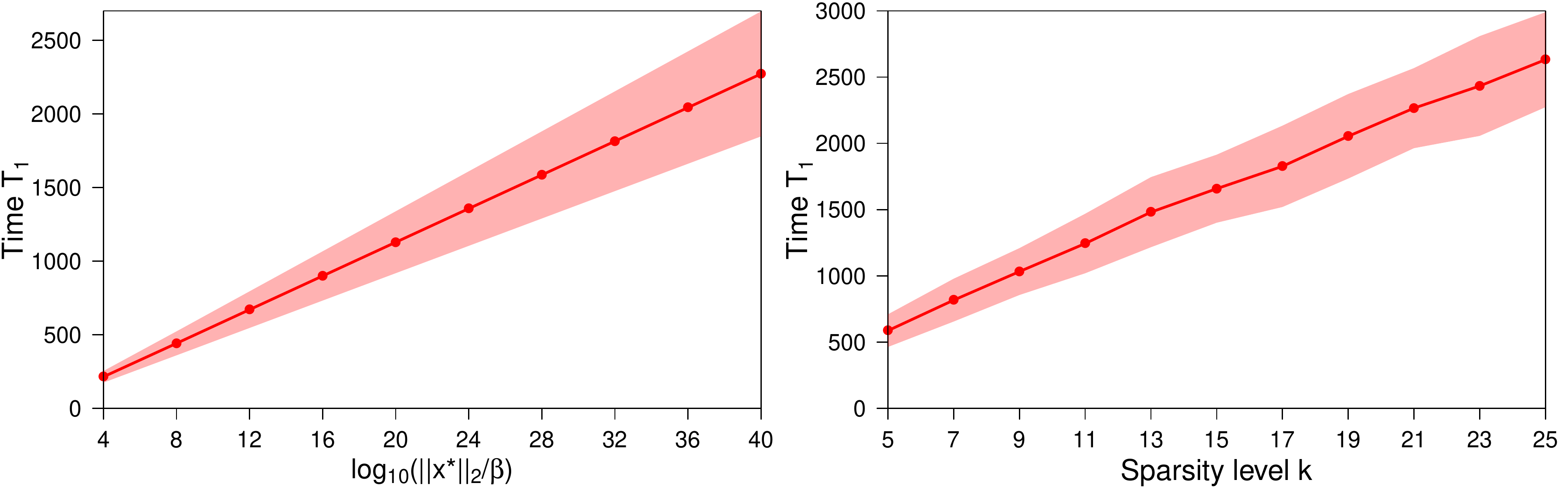}
\caption{Length of warm-up period $T_1$ plus/minus one standard deviation (shaded area) of early-stopped mirror descent. \textbf{Left}: $T_1$ against $\log(\|\x^\star\|_2/\beta)$, with $n=1000$, $m=1500$, $k=10$, $\sigma/\|\x^\star\|_2^2=0.1$. \textbf{Right}: $T_1$ against sparsity level $k$, with $n=2000$, $m=4000$, $\sigma/\|\x^\star\|_2^2=0.1$, $\beta/\|\x^\star\|_2=10^{-20}$.}
\label{fig2}
\end{figure}

\subsection{Convergence behavior of mirror descent}
The next experiment examines the convergence behavior of mirror descent and its dependence on the parameter $\beta$. We only present the $\ell_2$-error $\operatorname{dist}(\x^\star, \X^t)$. The Bregman divergence term displays the same qualitative behavior. We set $n=50000$, $m=1000$, $k=10$ and run mirror descent for different values of $\beta$. In both the noiseless ($\sigma/\|\x^\star\|_2^2 =0$) and noisy ($\sigma/\|\x^\star\|_2^2 = 0.5$) case, we observe the expected dependence on $\beta$: as $\beta$ decreases from $10^{-6}$ to $10^{-14}$, the length of the initial warm-up period increases as $\log \frac{1}{\beta}$, while we have linear convergence up to an improved accuracy that depends polynomially on $\beta$. In the noisy case, the precision up to which we have linear convergence barely improves as we decrease $\beta$ from $10^{-10}$ to $10^{-14}$, which we attribute to the noise term $\frac{\sigma}{\|\x^\star\|_2^2}\sqrt{\frac{k\log n}{m}}$ dominating the achievable accuracy. Figure \ref{fig3} suggests that early stopping is necessary in the presence of noise.
\begin{figure}[!t]
\centering
\includegraphics[width=\textwidth]{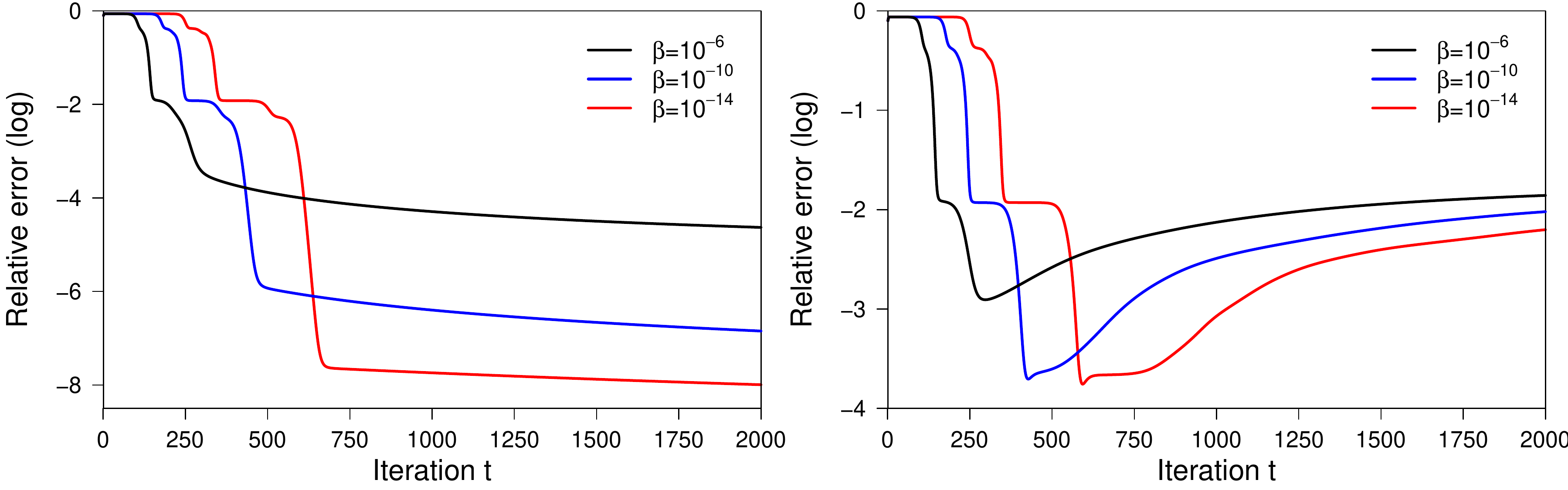}
\caption{Relative error $\operatorname{dist}(\x^\star, \X^t) / \|\x^\star\|_2$ (log-scale) of mirror descent for $n=50000$, $m=1000$, $k=10$, and different values of $\beta$. \textbf{Left}: noiseless case with $\sigma/\|\x^\star\|_2^2 = 0$. \textbf{Right}: noisy case with $\sigma/\|\x^\star\|_2^2 = 0.5$.}
\label{fig3}
\end{figure}

%% file: proofs.tex
\section{Proofs}
\label{section:proofs}
In this section, we prove the results stated in Section \ref{section:main_results}. Since the proof of Theorem \ref{theorem_discrete} largely follows the same steps as the proof of Theorem \ref{theorem}, with a few technical difficulties arising due to the discretization, we defer the proof of Theorem \ref{theorem_discrete} to Appendix \ref{supp:proof_theorem4}. 

\subsection{Proof of Lemma \ref{lemma2}}
\label{proof:lemma2}
The proof of Lemma \ref{lemma2} relies on the fact that the Bregman divergence associated to the hyperbolic entropy mirror map can be decomposed into a sum of functions, each of which depends only on a single coordinate $x_i$ and is (locally) bounded from above and below by quadratic functions. 
\begin{proof}[Proof of Lemma \ref{lemma2}]
Using the identiy $\operatorname{arcsinh}(x) = \log (x + \sqrt{1+x^2})$, we can write
\begin{align*}
D_{\Phi}(\x^\star,\x) &= \sum_{i=1}^n g_i(x_i),
\end{align*}
where 
\begin{equation*}
g_i(x_i) = \sqrt{x_i^2 + \beta^2} - \sqrt{(x_i^\star)^2 + \beta^2} - x^\star_i \log\frac{x_i + \sqrt{x_i^2 + \beta^2}}{x_i^\star + \sqrt{(x_i^\star)^2 + \beta^2}}. 
\end{equation*}
We begin by showing the bound (\ref{eq:lemma2ub}). For any $i\in [n]$, we have $g_i(x^\star_i) = 0$ and
\begin{align}\label{eq:gderiv}
g_i'(x_i) &=\frac{x_i}{\sqrt{x_i^2+\beta^2}} - x^\star_i\frac{1 + x_i/\sqrt{x_i^2+\beta^2}}{x_i+\sqrt{x_i^2 + \beta^2}} = \frac{x_i-x^\star_i}{\sqrt{x_i^2+\beta^2}}.
\end{align}
As we assume $x_ix^\star_i\ge 0$ and $|x_i|\ge \frac{1}{2}|x^\star_i|$, we have $\sqrt{z^2 + \beta^2} \ge \frac{x^\star_{min}\|\x^\star\|_2}{2} \ge \frac{c_\star\|\x^\star\|_2}{2\sqrt{k}}$ for any $z\in (x_i,x^\star_i)$ (or $(x^\star_i, x_i)$ if $x^\star_i < x_i$). Hence, using the convention $\int_a^bf(x)dx = -\int_b^af(x)dx$ when $a>b$, we can bound
\begin{equation*}
g_i(x_i) = g_i(x^\star_i) + \int_{x^\star_i}^{x_i}g'(z)dz < \int_{x^\star_i}^{x_i}\frac{2\sqrt{k}}{c_\star\|\x^\star\|_2}(z - x^\star_i)dz= \frac{\sqrt{k}}{c_\star\|\x^\star\|_2}(x_i-x^\star_i)^2.
\end{equation*}
Summing over all $i\in \S$, this gives
\begin{equation*}
\sum_{i\in \S} g_i(x_i) \le \frac{\sqrt{k}}{c_\star\|\x^\star\|_2}\|\x_\S - \x^\star_\S\|_2^2
\end{equation*}
as claimed. On the other hand, if $x^\star_i=0$, we have 
\begin{equation*}
g_i(x_i) = \sqrt{x_i^2+\beta^2} - \beta \le |x_i|,
\end{equation*}
which gives
\begin{equation*}
\sum_{i\notin \S} g_i(x_i) \le \|x_{\S^c}\|_1,
\end{equation*}
and completes the proof of the bound (\ref{eq:lemma2ub}) in Lemma \ref{lemma2}. The other bound (\ref{eq:lemma2lb}) can be shown similarly using (\ref{eq:gderiv}). We have $\sqrt{z^2 + \beta^2} \le \sqrt{\max\{\|\x\|_\infty^2, \|\x^\star\|_\infty^2\}+\beta^2}$ for all $z\in(x_i,x_i^\star)$ (or $(x^\star_i,x_i)$ if $x^\star_i<x_i$), which gives
\begin{align*}
g(x_i) = g(x^\star_i) + \int_{x_i}^{x_i^\star}g'(z)dz &> \int_{x_i}^{x_i^\star}\left(\max\{\|\x\|_\infty^2, \|\x^\star\|_\infty^2\}+\beta^2\right)^{-\frac{1}{2}}(z - x^\star_i)dz \\
&= \frac{1}{2}\big(\max\{\|\x\|_\infty^2, \|\x^\star\|_\infty^2\}+\beta^2\big)^{-\frac{1}{2}}(x_i-x^\star_i)^2.
\end{align*}
Summing over all $i\in [n]$ gives the bound (\ref{eq:lemma2lb}) of Lemma \ref{lemma2}.
\end{proof}

\subsection{Proof of Theorem \ref{theorem}}
\label{proof:theorem3}
In this section, we will make the following assumptions, which are purely for notational simplicity.
\begin{itemize}
\item We will assume $\|\x^\star\|_2 = 1$. The general case $\|\x^\star\|_2 \neq 1$ immediately follows by writing $\x^\star = \frac{\x^\star}{\|\x^\star\|_2}\|\x^\star\|_2$, $\x = \frac{\x}{\|\x^\star\|_2}\|\x^\star\|_2$ and $\varepsilon_j = \frac{\varepsilon_j}{\|\x^\star\|_2}\|\x^\star\|_2$ in what follows.

\item We will assume $\X(0)^\top\x^\star\ge 0$, in which case we show convergence towards $\x^\star$. Otherwise, we can show convergence towards $-\x^\star$.
\end{itemize}  
The gradient of the objective function (\ref{eq:objective_function}) is given by
\begin{equation*}
\nabla F(\x) = \frac{1}{m}\sum_{j=1}^m\bigl((\A_j^\top\x)^2 - (\A_j^\top\x^\star)^2 - \varepsilon_j\bigr)\bigl(\A_j^\top\x\bigr)A_j.
\end{equation*}
A brief calculation gives the following expression for its expectation, the population gradient.
\begin{equation*}
\nabla f(\x) = \E[\nabla F(\x)] = \big(3\|\x\|_2^2 - 1\big)\x - 2\big(\x^\top\x^\star\big)\x^\star.
\end{equation*}

First, we provide three supporting lemmas characterizing the behavior of mirror descent. We defer the proofs of the following lemmas to Appendix \ref{supp:proof_supporting_lemmas}.
Lemmas \ref{lemma:support1} and \ref{lemma:support2} show that the gradient $\nabla F(\x)$ is close to its expectation $\nabla f(\x)$.
\begin{lemma}\label{lemma:support1}
Let $\x^\star\in \R^n$ be a $k$-sparse vector, and let $\S = \{i\in [n]: x^\star_i\neq 0\}$ be its support. Let $\{\A_j\}_{j=1}^m$ be a collection of i.i.d.\ $\gauss(\mathbf{0},\mathbf{I}_n)$ random vectors and $\{\varepsilon_j\}_{j=1}^m$ a collection of independent centered sub-exponential random variables with maximum sub-exponential norm $\sigma = \max_{j}\|\varepsilon_j\|_{\psi_1}$. There exist universal constants $c,c_p>0$ such that for any constant $\gamma >0$, there is a constant $c_s(\gamma)>0$ such that the following holds. For any $0 \le \delta \le \frac{1}{4}$, let 
\begin{equation*}
\mathcal{X}= \{\x\in\R^n : \|\x\|_2\le \|\x^\star\|_2, \|\x_{\S^c}\|_1\le \delta \|\x^\star\|_2 \}.
\end{equation*}
If $m\ge c_s(\gamma) (1 + \frac{\sigma^2}{\|\x^\star\|_2^4})\max\{k^2\log^2 n, \;\log^5n\}$, then, with probability at least $1-c_pn^{-10}$,
\begin{align*}
\frac{\| \nabla F(\x) - \nabla f(\x)\|_{\infty}}{\|\x^\star\|_2^2} &\le \gamma\min\biggl\{\frac{\|\x_\S\|_1}{k},\; \frac{\|\x_\S-\x_\S^\star\|_2}{\sqrt{k}}\biggr\} + c\delta\|\x^\star\|_2 \\
&\quad + \frac{c\sigma}{\|\x^\star\|_2^2}\sqrt{\frac{\log n}{m}}\min\Bigl\{\|\x\|_1,\; (1+\delta)\|\x^\star\|_2 +  \sqrt{k}\|\x_\S-\x^\star_\S\|_2 \Bigr\}
\end{align*}
holds for all $\x\in\mathcal{X}$.

In particular, if $\frac{\|\x - \x^\star\|_2}{\|\x^\star\|_2} \ge \max\Bigl\{\frac{c\sqrt{k}\delta}{\gamma}, \frac{c(1+\delta)\sigma}{\gamma\|\x^\star\|_2^2}\sqrt{\frac{k\log n}{m}}\Bigr\}$, $\|\x_\S\|_1 \ge \frac{1}{2}\|\x^\star\|_2$, $\delta \le \frac{\gamma}{2ck}$ and $c_s(\gamma)\ge \frac{c^2}{\gamma^2}$, then we have the simplified bound
\begin{align*}
\frac{\|\nabla F(\x) - \nabla f(\x)\|_{\infty}}{\|\x^\star\|_2^2} \le 4\gamma\min\biggl\{\frac{\|\x_\S\|_1}{k},\; \frac{\|\x-\x^\star\|_2}{\sqrt{k}}\biggr\} \quad \text{for all } \x\in\mathcal{X}.
\end{align*}
\end{lemma}

\begin{lemma}\label{lemma:support2}
Let $\x^\star\in \R^n$ be a $k$-sparse vector, and let $\S = \{i\in [n]: x^\star_i\neq 0\}$ be its support. Let $\{\A_j\}_{j=1}^m$ be a collection of i.i.d.\ $\gauss(\mathbf{0},\mathbf{I}_n)$ random vectors and $\{\varepsilon_j\}_{j=1}^m$ a collection of independent centered sub-exponential random variables with maximum sub-exponential norm $\sigma = \max_{j}\|\varepsilon_j\|_{\psi_1}$. There exist universal constants $c,c_p>0$ such that for any constant $\gamma >0$, there is a constant $c_s(\gamma)>0$ such that the following holds. For any $0 \le \delta \le \frac{1}{4}\min\{\frac{\gamma}{c}, 1\}$, let
\begin{equation*}
\mathcal{X}= \{\x\in\R^n : \|\x\|_2\le \|\x^\star\|_2, \|\x_{\S^c}\|_1\le \delta \|\x^\star\|_2 \}.
\end{equation*}
If $m\ge c_s(\gamma) (1 + \frac{\sigma^2}{\|\x^\star\|_2^4})\max\{k^2\log^2 n, \;\log^5n\}$, then, with probability at least $1-c_pn^{-10}$,
\begin{align*}
\frac{\big|\langle \nabla F(\x) - \nabla f(\x), \x-\x^\star\rangle\big|}{\|\x^\star\|_2^3}
&\le \gamma\min \biggl\{\frac{\|\x_S\|_1}{\sqrt{k}} + \delta \|\x^\star\|_2,\; \frac{\|\x_\S-\x^\star_\S\|_2^2}{\|\x^\star\|_2} + \delta^2\|\x^\star\|_2\biggr\} \\
&\quad + \frac{c\sigma}{\|\x^\star\|_2^2}\sqrt{\frac{\log n}{m}}\Bigl(\delta \|\x^\star\|_2 + \min\bigl\{\|\x_\S\|_1, \sqrt{k}\|\x_\S-\x^\star_\S\|_2\bigr\}\Bigr)  
\end{align*}
holds for all $\x\in\mathcal{X}$. 

In particular, if $\frac{\|\x - \x^\star\|_2}{\|\x^\star\|_2} \ge \max\Bigl\{\delta,\; \frac{c\sigma}{\gamma\|\x^\star\|_2^2}\sqrt{\frac{k\log n}{m}}\Bigr\}$, $\|\x_\S\|_1 \ge \frac{1}{2}\|\x^\star\|_2$, $\delta \le \frac{1}{2\sqrt{k}}$ and $c_s(\gamma)\ge \frac{c^2}{\gamma^2}$, then we have the simplified bound
\begin{align*}
\frac{\big|\langle \nabla F(\x), \x-\x^\star\rangle - \langle \nabla f(\x), \x-\x^\star\rangle\big|}{\|\x^\star\|_2^3}
\le 4\gamma\min \biggl\{\frac{\|\x_S\|_1}{\sqrt{k}}, \; \frac{\|\x-\x^\star\|_2^2}{\|\x^\star\|_2}\biggr\} \quad \text{ for all } \x\in\mathcal{X}.
\end{align*}
\end{lemma}

The following lemma characterizes the region to which the trajectory of early-stopped mirror descent (both in continuous and discrete time) is confined.
\begin{lemma} \label{lemma:support3}
Let Assumption \ref{assumption} hold. There exist universal constants $c_p,c_1,c_2,c_3>0$ such that the following holds. Let $\delta\le c_1/n$ and let $\X(t)$ be defined by the continuous-time mirror descent algorithm given by (\ref{eq:mirror_descent_continuous}) with mirror map (\ref{eq:mirror_map}) and initialization (\ref{eq:initialization}) with $\beta \le \delta$. Assume that there is a $T>0$ such that
\begin{align*}
\|\X_{\S^c}(t)\|_1 \le \delta \|\x^\star\|_2 \qquad \text{and} \qquad \frac{\|\X(t) - \x^\star\|_2}{\|\x^\star\|_2} \ge c_2\max\biggl\{\sqrt{k}\delta,\; \frac{\sigma}{\|\x^\star\|_2^2}\sqrt{\frac{k\log n}{m}}\biggr\}
\end{align*}
are satisfied for all $t\le T$.
Then, there is a $\xi \in \{-1,+1\}$, such that, with probability at least $1 - c_pn^{-10}$, the following holds for all $t\le T$:
\begin{align}
\xi X_i(t)x^\star_i &\ge 0 \quad \text{for all } i\in [n], \label{eq:claim1}\\
\frac{1}{3} - (3+9\sigma)\sqrt{\frac{\log n}{m}} &\le \frac{\|\X(t)\|_2^2}{\|\x^\star\|_2^2} \le 2, \label{eq:claim2}\\
\sqrt{3} \cdot 2\big|\X(t)^\top\x^\star\big| &\ge 3\|\X(t)\|_2^2 - \|\x^\star\|_2^2. \label{eq:claim3}
\end{align}
The same result holds for $\X^t$ defined by the discrete-time mirror descent algorithm given by (\ref{eq:mirror_descent_discrete}) with $\eta \le c_3/\|\x^\star\|_2^3$.
\end{lemma}

The following inequalities will be useful throughout the following proofs.
Under the assumptions of Lemma \ref{lemma:support3}, we can use (\ref{eq:claim2}) to bound
\begin{equation}\label{eq:bound_size}
\|\X_\S(t)\|_1 = \|\X(t)\|_1 - \|\X_{\S^c}(t)\|_1 \ge \biggl(\frac{1}{3} - (3+9\sigma)\sqrt{\frac{\log n}{m}}\biggr)^{-\frac{1}{2}} - \delta \ge \frac{1}{2},
\end{equation}
and, using (\ref{eq:claim1}) of Lemma \ref{lemma:support3}, we can bound the inner product
\begin{equation}\label{eq:bound_innerproduct}
\X(t)^\top\x^\star \ge \|\X_\S(t)\|_1x^\star_{min} \ge \frac{c_\star}{2\sqrt{k}}.
\end{equation}

\begin{proof}[Proof of Theorem \ref{theorem}]
We first consider the population dynamics to illustrate the main ideas of the proof of Theorem \ref{theorem}. That is, we first assume that we had access to the population gradient $\nabla f$, or in other words that $m=\infty$. 
The proof of Theorem \ref{theorem} relies on the identity
\begin{equation}\label{eq:bregman_derivative}
\frac{d}{dt}D_{\Phi}(\x^\star,\X(t)) = - \big\langle\nabla F(\X(t)), \X(t) - \x^\star \big\rangle,
\end{equation}
and we bound the inner product on the right hand side to show that the Bregman divergence $D_{\Phi}(\x^\star, \X(t))$ decreases as claimed.

\textbf{Initial Bregman divergence}\\
Writing $\hat{\theta} = (\sum_{j=1}^mY_j/m)^{1/2}$ for the estimate of the signal size $\|\x^\star\|_2$, we can use standard concentration bounds for sub-exponential random variables (see e.g.\ Prop.\ 5.16 of \cite{V12}) to bound, with probability $1-4n^{10}$, 
\begin{equation*}
\hat{\theta}^2 = \frac{1}{m}\sum_{j=1}^mY_j = \frac{1}{m}\sum_{j=1}^m(\A_j^\top\x^\star)^2 + \frac{1}{m}\sum_{j=1}^m\varepsilon_j < \biggl(1 + (9 + 25\sigma)\sqrt{\frac{\log n}{m}}\biggr) \|\x^\star\|_2^2,
\end{equation*}
Using the definition of the initialization (\ref{eq:initialization}), we can bound the initial Bregman divergence
\begin{align}
\label{eq:bregman_divergence_initial}
D_{\Phi}(\x^\star,\X(0)) &= \sum_{i\neq I_0} \beta - \sqrt{(x_i^\star)^2 + \beta^2} - x^\star_i \log\frac{\beta}{x_i^\star + \sqrt{(x_i^\star)^2 + \beta^2}} \nonumber\\
&\quad + \sqrt{\hat{\theta}^2/3 + \beta^2} - \sqrt{(x_{I_0}^\star)^2 + \beta^2} - x^\star_{I_0} \log\frac{\hat{\theta}/\sqrt{3} + \sqrt{\hat{\theta}^2/3 + \beta^2}}{x_{I_0}^\star + \sqrt{(x_{I_0}^\star)^2 + \beta^2}} \nonumber\\
&\le \sum_{i\neq I_0}|x^\star_i|\log \frac{1}{\beta} + \underbrace{\beta - \sqrt{(x^\star_i)^2+\beta^2} + |x^\star_i|\log \Big(|x^\star_i| + \sqrt{(x^\star_i)^2 + \beta^2}\Big)}_{\le 0} + 1 \nonumber\\
&\le \|\x^\star\|_1\log \frac{1}{\beta} + 1,
\end{align}
where we used the assumption $\|\x^\star\|_2 = 1$ and the fact that the function $x\log\frac{\beta}{x+\sqrt{x^2+\beta^2}}$ is symmetric, that is $x\log\frac{\beta}{x+\sqrt{x^2+\beta^2}} = -x\log\frac{\beta}{-x+\sqrt{x^2+\beta^2}}$.

\textbf{Bounding $\langle\nabla f(\X(t)), \X(t) - \x^\star \rangle$}\\
We can compute
\begin{align*}
\big\langle\nabla f(\X(t)), \X(t) - \x^\star \big\rangle &= -\big[3\big(\|\X(t)\|_2^2 - 1\big) + 2\big(\X(t)^\top\x^\star\big)\big]\big(\X(t)^\top\x^\star\big) \\
&\quad + \big(3\|\X(t)\|_2^2 - 1\big)\|\X(t)\|_2^2.
\end{align*}
To bound this quantity, we distinguish two cases.
\begin{itemize}
\item \textbf{Case 1: $\|\X(t)\|_2^2 \le \frac{2}{5}$}\\
In this case, we use the fact that $3\|\X(t)\|_2^2-1 \ge -(9+27\sigma)\sqrt{\frac{\log n}{m}}$ by Lemma \ref{lemma:support3} to bound
\begin{align}\label{eq:population_dynamics1}
\big\langle \nabla f(\X(t)), \X(t) - \x^\star \big\rangle &\ge \biggl(\frac{9}{5} - 2\sqrt{\frac{2}{5}}\biggr)\big(\X(t)^\top\x^\star\big) - \frac{18 + 54\sigma}{5}\sqrt{\frac{\log n}{m}} \nonumber\\
& \ge \frac{1}{2}\big(\X(t)^\top\x^\star\big), 
\end{align}
where for the first inequality we used that $\X(t)^\top\x^\star\le \|\X(t)\|_2\|\x^\star\|_2$ by the Cauchy-Schwarz inequality, and the second inequality holds by (\ref{eq:bound_innerproduct}) since $m\ge c_sk^2\log^2 n$.

\item \textbf{Case 2: $\|\X(t)\|_2^2 > \frac{2}{5}$}\\
In this case, we can write $2(\X(t)^\top\x^\star) = \|\X(t)\|_2^2 + \|\x^\star\|_2^2 - \|\X(t)-\x\|_2^2$ to compute
\begin{align*}
\big\langle \nabla f(\X(t)), \X(t) - \x^\star \big\rangle = \big(2 - 4\|\X(t)\|_2^2 + \Delta \big)\big(\X(t)^\top\x^\star\big) + \big(3\|\X(t)\|_2^2 - 1\big) \|\X(t)\|_2^2,
\end{align*}
where we denote $\Delta = \|\X(t)-\x^\star\|_2^2$.
If $2 - 4\|\X(t)\|_2^2 + \Delta\ge 0$, this is lower bounded by
\begin{equation}
\label{eq:population_dynamics20}
\big\langle \nabla f(\X(t)), \X(t) - \x^\star \big\rangle \ge \frac{2}{25}.
\end{equation}
Otherwise, again using $\X(t)^\top\x^\star \le \|\X(t)\|_2\|\x^\star\|_2$, we have
\begin{align}\label{eq:population_dynamics2}
\big\langle \nabla f(\X(t)), \x^\star-\X(t) \big\rangle &\ge \underbrace{3\|\X(t)\|_2^4 - 4\|\X(t)\|_2^3 - \|\X(t)\|_2^2 + 2\|\X(t)\|_2}_{\ge 0} + \Delta \|\X(t)\|_2 \nonumber\\
&\ge \sqrt{\frac{2}{5}} \Delta .
\end{align}
\end{itemize}
The inequalities (\ref{eq:population_dynamics1})--(\ref{eq:population_dynamics2}) can be used to show that the Bregman divergence $D_{\Phi}(\x^\star,\X(t))$ decreases as claimed. In particular, we can use Lemma \ref{lemma2} to bound $\|\X(t)-\x^\star\|_2^2$ in terms of $D_{\Phi}(\x^\star, \X(t))$, and inequality (\ref{eq:population_dynamics2}) yields a bound leading to linear convergence as long as $\|\X_{\S^c}(t)\|_1$ is sufficiently small compared to $\|\X_\S(t)-\x^\star_\S\|_2^2$. In order to make the proof of Theorem \ref{theorem} rigorous, we need to replace the population gradient $\nabla f$ by the empirical gradient $\nabla F$ in the outline we provided above. To this end, we divide the analysis of the convergence of mirror descent into two stages bounded by
\begin{align*}
T_1 &= \inf\biggl\{t>0 : \min_{i\in \S} \frac{|X_i(t)|}{|x^\star_i|} > \frac{1}{2} \biggr\}, \text{ and}\\
T_2 &= \inf\biggl\{t>0: \frac{\|\X(t)-\x^\star\|_2^2}{\|\x^\star\|_2^2} \le c^2\max\Bigl\{c_\star\sqrt{k}\delta, \; \frac{\sigma^2}{\|\x^\star\|_2^4}\frac{k\log n}{m}\Bigr\} \biggr\},
\end{align*}
respectively, where $\delta = \sqrt{n\beta}$.
We consider the stages (i) $t\le T_1$ and (ii) $T_1< t \le T_2$. Note that we have $T_2>T_1$, because if there is an index $i\in \S$ with $|X_i(t)|<\frac{1}{2}|x^\star_i|$, then we also have $\|\X(t)-\x^\star\|_2 > \frac{c_\star}{2\sqrt{k}}$.
In both stages, we will (a) bound the length of the stage by using (\ref{eq:bregman_derivative}) to bound $T_i$, and (b) show that off-support coordinates satisfy $\|\X_{\S^c}(t)\|_1\le \delta$. Throughout the proof we will assume that the inequalities in Lemmas \ref{lemma:support1}--\ref{lemma:support3} are satisfied, which happens with probability at least $1-c_pn^{-10}$. Note that as $\|\X_\S(t)\|_1 \ge \|\X(t)\|_2 - \delta \ge \frac{1}{2}$ by the lower bound in (\ref{eq:claim2}) of Lemma \ref{lemma:support3} and since $\delta\lesssim n^{-1}$, the conditions for the simplified bounds in Lemmas \ref{lemma:support1} and \ref{lemma:support2} are satisfied for $t\le T_2$.

\textbf{Stage (i), part (a): $t\le T_1$, bound $T_1$}\\
Assume for now that we have already shown $\|\X_{\S^c}(t)\|_1< \delta_1 = n\beta^{3/4}$ for all $t\le T_1$. Note that we have $\delta_1\le \delta$ since $\beta \le n^{-2}$. We have already computed the rate (\ref{eq:bregman_derivative}) at which the Bregman divergence $D_{\Phi}(\x^\star,\X(t))$ decreases if we had access to the population gradient $\nabla f$ in (\ref{eq:population_dynamics1}) and (\ref{eq:population_dynamics2}). Using these bounds, we can bound the inner product in (\ref{eq:bregman_derivative}) with the empirical gradient $\nabla F$.
\begin{itemize}
\item \textbf{Case 1: $\|\X(t)\|_2^2\le \frac{2}{5}$}\\
By Lemma \ref{lemma:support2} and the definitions of $\delta_1$ and $T_2$, we have
\begin{equation*}
\big|\big\langle \nabla F(\X(t)), \X(t) - \x^\star \big\rangle - \big\langle \nabla f(\X(t)), \X(t) - \x^\star \big\rangle\big| \le  \frac{c_\star}{4\sqrt{k}}\|\X_\S(t)\|_1 \le \frac{1}{4}\big(\X(t)^\top\x^\star\big).
\end{equation*}
Together with (\ref{eq:population_dynamics1}), this bound leads to
\begin{align} \label{eq:bound1}
\frac{d}{dt}D_{\Phi}(\x^\star, \X(t)) &\le -\big\langle \nabla f(\X(t)), \X(t) - \x^\star \big\rangle \nonumber\\*
&\quad + \big|\big\langle \nabla F(\X(t)), \X(t) - \x^\star \big\rangle - \big\langle \nabla f(\X(t)), \X(t) - \x^\star \big\rangle\big| \nonumber\\
&\le - \frac{1}{4}\big(\X(t)^\top\x^\star\big) 
\le - \frac{c_\star}{8\sqrt{k}},
\end{align}
where for the last line we used (\ref{eq:bound_innerproduct}). 

\item \textbf{Case 2: $\|\X(t)\|_2^2\ge \frac{2}{5}$}\\
As in the previous case, we can bound the difference $\langle \nabla F(\X(t)) - \nabla f(\X(t)), \X(t) - \x^\star \rangle$ using Lemma \ref{lemma:support2}. Recalling (\ref{eq:population_dynamics20}) and (\ref{eq:population_dynamics2}), we obtain
\begin{equation}\label{eq:bound21}
\frac{d}{dt}D_{\Phi}(\x^\star, \X(t)) = -\big\langle \nabla F(\X(t)), \X(t) - \x^\star \big\rangle \le - \frac{1}{25}
\end{equation}
if $2 - 4\|\X(t)\|_2^2 + \|\X(t)-\x^\star\|_2^2\ge 0$, and
\begin{equation}\label{eq:bound2}
\frac{d}{dt}D_{\Phi}(\x^\star, \X(t))= -\big\langle \nabla F(\X(t)), \X(t) - \x^\star \big\rangle \le -\frac{1}{2}\|\X(t) - \x^\star\|_2^2
\end{equation}
if $2 - 4\|\X(t)\|_2^2 + \|\X(t)-\x^\star\|_2^2 < 0$, where we used that 
\begin{align*}
\left|\big\langle \nabla F(\X(t)) - \nabla f(\X(t), \; \X(t) - \x^\star \big\rangle \right| \le \biggl(\sqrt{\frac{2}{5}}-\frac{1}{2}\biggr)\|\X(t) - \x^\star\|_2^2,
\end{align*}
for all $t\le T_2$ by the simplified bound in Lemma \ref{lemma:support2}.
\end{itemize}
We can now bound $T_1$.
Define
\begin{equation*}
T_0 = \inf\biggl\{t>0: \|\X(t)\|_2^2>\frac{2}{5}\biggr\},
\end{equation*}
as the time until which Case 1 holds.
Then, the bound (\ref{eq:bound1}) from Case 1 implies
\begin{equation*}
\min\{T_1,T_0\} \le \frac{8\sqrt{k}D_{\Phi}(\x^\star,\X(0))}{c_\star} \le \frac{8}{c_\star} k\log \frac{1}{\beta} + \frac{8\sqrt{k}}{c_\star} \le \frac{c_2}{2}k\log\frac{1}{\beta},
\end{equation*}
provided that $c_2\ge \frac{16}{c_\star} + \frac{16}{c_\star\sqrt{k}\log \frac{1}{\beta}}$, where we used (\ref{eq:bregman_divergence_initial}) to bound the initial Bregman divergence $D_{\Phi}(\x^\star,\X(0))$.
If $T_1<T_0$, then we have bounded $T_1$ as desired.
Otherwise, we can use the bounds (\ref{eq:bound21}) and (\ref{eq:bound2}) from Case 2 to control $T_1 - T_0$. The first bound (\ref{eq:bound21}) can apply at most for $t\le 25D_{\Phi}(\x^\star, \X(0))\le 25\sqrt{k}\log \frac{1}{\beta} + 25$, where we again used the bound (\ref{eq:bregman_divergence_initial}).
As the bound (\ref{eq:bound2}) depends on the $\ell_2$-distance $\|\X(t)-\x^\star\|_2^2$, the idea is to show that this quantity is sufficiently large as long as the Bregman divergence $D_{\Phi}(\x^\star,\X(t))$ is large.
To this end, define
\begin{equation*}
S(t) = \biggl\{i\in \S: \frac{|X_i(t)|}{|x^\star_i|} < \frac{1}{2}\biggr\}.
\end{equation*}
With this, we can bound
\begin{align}
\|\X(t)-\x^\star\|_2^2 &\ge \sum_{i\in S(t)} \Bigl(\frac{1}{2}x^\star_i\Bigr)^2 + \sum_{i\notin S(t)}\bigl(X_i(t) - x^\star_i\bigr)^2 \nonumber \\
&\ge \frac{c_\star^2}{4k}|S(t)| + \sum_{i\notin S(t)}\bigl(X_i(t) - x^\star_i\bigr)^2. \label{eq:boundbreg1}
\end{align}
As in the computation for (\ref{eq:bregman_divergence_initial}) and the proof of Lemma \ref{lemma2}, we can bound
\begin{align}\label{eq:boundbreg2}
D_{\Phi}(\x^\star,\X(t)) \le \sum_{i\in S(t)} |x^\star_i| \log \frac{1}{\beta} + \frac{\sqrt{k}}{c_\star} \sum_{i\notin S(t)}\bigl(X_i(t) - x^\star_i\bigr)^2 + \|\X_{\S^c}(t)\|_1.
\end{align}
Because $\|\x^\star\|_\infty\le \|\x^\star\|_2=1$, we have $\sum_{i\in S(t)}|x^\star_i|\le |S(t)|$. Further, since $t< T_1$, we also have $\|\X_{\S^c}(t)\|_1\le \delta_1\le \frac{1}{4}(x^\star_{min})^2\le \sum_{i\notin S(t)}(X_i(t)-x^\star_i)^2$. With this, we can combine (\ref{eq:bound2}), (\ref{eq:boundbreg1}) and (\ref{eq:boundbreg2}) to bound
\begin{equation*}
\frac{d}{dt}D_{\Phi}(\x^\star,\X(t)) \le -\frac{1}{2} \|\X(t)-\x^\star\|_2^2 \le - \frac{c_\star^2}{8k\log \frac{1}{\beta}} D_{\Phi}(\x^\star,\X(t)),
\end{equation*}
which shows that $D_{\Phi}(\x^\star,\X(t))$ decreases linearly at the rate $c_\star^2/(8k\log \frac{1}{\beta})$ for $T_0 < t < T_1$. We have
\begin{equation*}
D_{\Phi}(\x^\star,\X(T_0)) \le D_{\Phi}(\x^\star,\X(0)) \le \sqrt{k}\log \frac{1}{\beta}+1,
\end{equation*}
and, for $t<T_1$,
\begin{equation*}
\frac{c_\star^2}{4k}\le \|\X(t)-\x^\star\|_2^2 \le 2\sqrt{2+\beta^2}\; D_{\Phi}(\x^\star,\X(t))\le 3D_{\Phi}(\x^\star,\X(t)),
\end{equation*}
where for the second inequality we used that $\|\X(t)\|_\infty\le \sqrt{2}$ by Lemma \ref{lemma:support3} and the bound (\ref{eq:lemma2lb}) of Lemma \ref{lemma2}.
This implies
\begin{equation*}
T_1-T_0 \le \frac{8k\log\frac{1}{\beta}}{c_\star^2} \log \biggl(\frac{12}{c_\star^2} k^{\frac{3}{2}}\log \frac{1}{\beta} + \frac{12}{c_\star^2}k\biggr) \le \frac{c_2}{2}k\log\biggl(\frac{1}{\beta}\biggr) \log\biggl(k\log \frac{1}{\beta}\biggr),
\end{equation*}
provided that $c_2\ge \frac{24}{c_\star^2} + \frac{16}{c_\star^2}\log\frac{24}{c_\star^2}$.

\textbf{Stage (i), part (b): $t\le T_1$, bound $\|\X_{\S^c}(t)\|_1$}\\
The idea to controlling $\|\X_{\S^c}(t)\|_1$ is as follows: we will show that for coordinates $j\notin \S$ and $i\in \S$, $X_j(t)$ can only grow at a comparatively slower rate than $X_i(t)$. We will show that the growth of both coordinates is bounded by exponentials, and use the fact that, for any fixed $\epsilon>0$, the gap between $\beta(1+2\epsilon)^t$ and $\beta(1+\epsilon)^t$ can be made arbitrarily large by choosing $t$ large and $\beta$ small enough. Recall that
\begin{equation*}
\frac{d}{dt}X_i(t) = -\sqrt{X_i(t)^2 + \beta^2} \; \nabla F(\X(t))_i.
\end{equation*}
By Lemma \ref{lemma:support3}, we have $\sqrt{3}\cdot 2(\X(t)^\top\x^\star) \ge 3\|\X(t)\|_2^2 - 1 $. For any $i\in \S$ with $x^\star_i>0$ and $X_i(t)\le \frac{1}{2}x^\star_i$, we have
\begin{equation*}
\nabla f(\X(t))_i \le \biggl(\frac{\sqrt{3}}{2}-1\biggr)\cdot 2\big(\X(t)^\top\x^\star\big)x^\star_i \le - \frac{c_\star^2\|\X_\S(t)\|_1}{4k}.
\end{equation*}
As before, Lemma \ref{lemma:support1} gives
\begin{equation*}
\big|\nabla F(\X(t))_i - \nabla f(\X(t))_i\big| \le \frac{c_\star^2\|\X_\S(t)\|_1}{8k},
\end{equation*}
for $t\le T_2$. This allows us to bound
\begin{equation*}
\nabla F(\X(t))_{i} \le - \frac{c_\star^2\|\X_\S(t)\|_1}{8k},
\end{equation*}
from which it follows that $\frac{d}{dt}X_i(t)>0$. The analogous result holds for coordinates $i\in \S$ with $x^\star_i<0$. In other words, once we have $|X_i(t_0)|\ge \frac{1}{2} |x^\star_i|$ for some $t_0>0$, then $|X_i(t)|\ge \frac{1}{2}|x^\star_i|$ continues to hold for $t\ge t_0$.
With this, we can define $I_1\in \S$ to be the last coordinate which crosses this threshold, that is for which $\frac{|X_i(t)|}{|x^\star_i|}\ge\frac{1}{2}$, and assume without loss of generality that $x^\star_{I_1}>0$. By definition, we have $|X_{I_1}(t)| \le \frac{1}{2} |x^\star_{I_1}|$ for all $t\le T_1$.
For any $j\notin \S$ with  $X_j(t)\ge 0$, we have  
\begin{equation*}
\nabla f(\X(t))_j \ge -(9 + 27\sigma)\sqrt{\frac{\log n}{m}} X_j(t).
\end{equation*}
As before, we can use Lemma \ref{lemma:support1} to obtain
\begin{equation*}
\nabla F(\X(t))_j \ge -\frac{1}{4\sqrt{2}}\frac{c_\star^2\|\X_\S(t)\|_1}{8k}.
\end{equation*}
The analogous result holds for $X_j(t)<0$, which shows that, for any $j\notin \S$,
\begin{equation}
\label{eq:stage1b}
|\nabla F(\X(t))_j| \le \frac{1}{4\sqrt{2}} |\nabla F(\X(t))_{I_1}|
\end{equation}
holds for all $t\le T_1$.
The idea is that $X_{I_1}(t)$ and $X_j(t)$ both grow (approximately) exponentially, but at different (time-varying) rates. By the definition of $I_1$, we have $X_{I_1}(T_1) = \frac{1}{2}x^\star_{I_1}$, and $X_j(T_1)$ can be made arbitrarily small by choosing a sufficiently small parameter $\beta$.

To make this rigorous, let $T_\beta = \inf\{t>0: \X_{I_1}(t) \ge \beta\}$ be the time when $\X_{I_1}(t)$ first reaches $\beta$. Then, since $\sqrt{X_{I_1}(t)^2 + \beta^2} \ge X_{I_1}(t)$, we can bound
\begin{equation*}
X_{I_1}(t) \ge \beta \exp\biggl(-\int_{T_\beta}^{t}\nabla F(\X(s))ds\biggr) \quad \Rightarrow \quad \exp\biggl(-\int_{T_\beta}^{t}\nabla F(\X(s))ds\biggr) \le \frac{X_{I_1}(t)}{\beta}.
\end{equation*}
Recalling the bound (\ref{eq:stage1b}), we have for $j\notin \S$,
\begin{align*}
|X_j(T_\beta)| \le \beta, \qquad \biggl|\frac{d}{dt}X_j(t)\biggr| \le \frac{1}{4\sqrt{2}}\sqrt{X_j(t)^2 + \beta^2} \; |\nabla F(\X(t))_{I_1}|.
\end{align*}
Then, since $\sqrt{X_j(t)^2 + \beta^2} \le \sqrt{2} |X_j(t)|$ for $X_j(t)\ge \beta$, we can bound
\begin{equation*}
X_j(t) \le \beta \exp\biggl(-\frac{1}{4}\int_{T_\beta}^{t}\nabla F(\X(s))ds\biggr) \le \beta \biggl(\frac{X_{I_1}(t)}{\beta}\biggr)^{1/4} \le \beta^{3/4},
\end{equation*}
for all $t\le T_1$, where we used the fact that $X_{I_1}(t)\le \frac{1}{2}x^\star_{I_1}\le 1$ for $t\le T_1$.
As this holds for every $j\notin \S$, we have, for all $t\le T_1$,
\begin{equation*}
\|\X_{\S^c}(t)\|_1 \le n\beta^{3/4} = \delta_1
\end{equation*}

\textbf{Stage (ii), part (a): $T_1<t\le T_2$, bound $D_{\Phi}(\x^\star,\X(t))$}\\
As in Stage (i), we first assume that we have already shown $\|\X_{\S^c}(t)\|_1\le \delta$ for all $t\le T_2$. As before, we can use Lemma \ref{lemma:support2} to derive the bound (\ref{eq:bound2}). By the definition of $T_1$, we have $|X_i(t)|\ge \frac{1}{2}|x^\star_i|$ for all $i\in [n]$. By Lemma \ref{lemma:support3} we also have $X_i(t)x^\star_i\ge 0$ (since we assumed $x^\star_{I_0}>0$), so the assumptions for inequality (\ref{eq:lemma2ub}) of Lemma \ref{lemma2} are satisfied, and we can bound
\begin{equation*}
D_{\Phi}(\x^\star,\X(t)) \le \frac{\sqrt{k}}{c_\star}\|\X_\S(t)-\x^\star_\S\|_2^2 + \|\X_{\S^c}(t)\|_1 \le \frac{2\sqrt{k}}{c_\star}\|\X(t)-\x^\star\|_2^2,
\end{equation*}
where we used that $\|\X_{\S^c}(t)\|_1 \le \delta \le \frac{\sqrt{k}}{c_\star}\|\X(t)-\x^\star\|_2^2$. With this, inequality (\ref{eq:bound2}) reads
\begin{equation*}
\frac{d}{dt}D_{\Phi}(\x^\star,\X(t)) \le -\frac{c_\star}{4\sqrt{k}}D_{\Phi}(\x^\star,\X(t)).
\end{equation*}
Using the fact that $\|\X(t)-\x^\star\|_2^2\le 3$ by Lemma \ref{lemma:support3}, we can bound $D_{\Phi}(\x^\star,\X(T_1)) \le \frac{6\sqrt{k}}{c_\star}$. Hence, we have for $T_1\le t\le T_2$,
\begin{align*}
D_{\Phi}(\x^\star,\X(t)) &\le D_{\Phi}(\x^\star,\X(T_1)) \exp\biggl(-\frac{c_\star}{4\sqrt{k}}(t-T_1)\biggr) 
\le \frac{6\sqrt{k}}{c_\star} \exp\biggl(-\frac{c_\star}{4\sqrt{k}}(t-T_1)\biggr).
\end{align*}
Together with the fact that $\|\X(t)\|_\infty^2\le 2$ by Lemma \ref{lemma:support3}, bound (\ref{eq:lemma2lb}) of Lemma \ref{lemma2} implies
\begin{equation*}
D_{\Phi}(\x^\star, \X(t)) > \frac{1}{3}\|\X(t)-\x^\star\|_2^2 \ge \frac{c^2c_\star\sqrt{k}\delta}{3}
\end{equation*}
for $t\le T_2$, so we have
$
T_2-T_1 \le \frac{4\sqrt{k}}{c_\star}\log \frac{18}{c^2c_\star^2\delta}.
$

\textbf{Stage (ii), part (b): $T_1<t\le T_2$, bound $\|\X_{\S^c}(t)\|_1$}\\
Recall that $|X_j(t)|\le \beta^{3/4}$ for all $j\notin \S$ and $t\le T_1$. Further, as in the previous stage we can use Lemma \ref{lemma:support1} to show
\begin{equation*}
|\nabla F(\X(t))_j| \le \frac{c_\star\|\X_\S(t)\|_1}{32k} \le \frac{c_\star}{16\sqrt{2k}},
\end{equation*}
where we used that $\|\X_\S(t)\|_1\le \sqrt{k}\|\X_\S(t)\|_2\le \sqrt{2k}$. Noting $\sqrt{x^2+\beta^2}\le \sqrt{2}x$ for $x\ge \beta$, we can bound, for $t\le T_2$,
\begin{align*}
|X_j(t)| &\le \beta^{3/4}\exp\biggl(\frac{c_\star}{16\sqrt{k}}(T_2-T_1)\biggr) 
\le \beta^{3/4}\exp\biggl(\frac{c_\star}{16\sqrt{k}} \frac{4\sqrt{k}}{c_\star}\log \frac{18}{c^2c_\star^2\delta}\biggr) \\
&\le \beta^{3/4}\biggl(\frac{18}{c^2c_\star^2\delta}\biggr)^{1/4}
 \le \frac{\delta}{n},
\end{align*}
provided that $\beta \le (c^2c_\star^2/18)^2/n^3$, where we used the definition $\delta = \sqrt{n\beta}$. This completes the proof that $\|\X_{\S^c}(t)\|_1\le \delta$ for all $t\le T_2$.
\end{proof}

%% file: conclusion.tex
\section{Conclusion}
\label{section:conclusion}
In this paper, we establish a general theory of mirror descent with early stopping for the problem of noisy sparse phase retrieval. We provide a full convergence analysis in both continuous and discrete time and we show that, when equipped with the hyperbolic entropy mirror map and with proper initialization, early-stopped mirror descent achieves a nearly minimax-optimal rate of convergence, provided the number of measurements $m$ is sufficiently large and the minimum (in modulus) non-zero signal component is on the order of $\|\x^\star\|_2/\sqrt{k}$. These conditions have been previously considered to establish theoretical results for sparse phase retrieval \cite{NJS15, WZGAC18, ZWGC18}. 

Previous procedures achieving a nearly minimax-optimal rate of convergence include empirical risk minimization with sparsity constraint \cite{LM15}, which does not lead to a tractable algorithm, and thresholded Wirtinger flow \cite{CLM16}, which relies on thresholding steps to promote sparsity. In contrast, unlike most existing algorithms designed to solve sparse phase retrieval, mirror descent does not rely on added regularization terms or thresholding steps to promote sparsity, and no a-priori knowledge of the sparsity $k$ or the noise level $\sigma$ is needed to \emph{run} the algorithm. Our numerical simulations attest that a standard data-dependent stopping rule that does not require knowledge of $k$ or $\sigma$ yields results that validate our theoretical findings.

Most of the literature on early stopping for iterative gradient-based methodologies has focused on the convex setting of ridge regression and kernel methods via Euclidean gradient descent and boosting. Our results establish connections between early stopping and sparsity in the non-convex setting of sparse phase retrieval via the general framework of mirror descent. The potential-based analysis that we present unveils and exploits a quantitative version of variational coherence that is satisfied along the path traced by the iterates of mirror descent, and it might inspire similar approaches to investigate implicit regularization via early stopping in other non-convex settings in statistical inference.

%% file: proof_theorem4.tex
%%%%%%%%%%%%%%%%%%%%%%%%%%%%%%%%%%%%%%%%%%%%%%
%%%% Proof of Theorem 4:
\section{Proof of Theorem \ref{theorem_discrete}}
\label{supp:proof_theorem4}
As in the proof of Theorem \ref{theorem}, we will assume $\|\x^\star\|_2 = 1$ and $(\X^0)^\top\x^\star \ge 0$ for notational simplicity. 
The proof of Theorem \ref{theorem_discrete} largely follows the same steps as the proof of Theorem \ref{theorem}, and we can reuse many of the bounds shown before. Before proving the theorem, we derive some generic bounds on how much $X_i^t$ can increase or decrease in one iteration. The update (\ref{eq:update_eg}) can be written as (see e.g.\ Theorem 24 \cite{GHS20}) 
\begin{align}
X_i^{t+1} - X_i^t &= \frac{\sqrt{(X_i^t)^2 + \beta^2} + X_i^t}{2}\bigl(\exp(-\eta\nabla F(\X^t)_i) - 1\bigr) \nonumber\\
&\quad - \frac{\sqrt{(X_i^t)^2 + \beta^2} - X_i^t}{2}\bigl(\exp(\eta\nabla F(\X^t)_i) - 1\bigr). \label{eq:update}
\end{align}
Assuming that the conditions for the simplified bound in Lemma \ref{lemma:support1} are satisfied, we can bound
\begin{equation}\label{eq:bound_gradient_difference}
\big|\nabla F(\X^t)_i - \nabla f(\X^t)_i\big|\le \frac{0.1c_\star\big((\X^t)^\top\x^\star\big)}{\sqrt{k}}\le 0.1\big((\X^t)^\top\x^\star\big)|x^\star_i|,
\end{equation}
with probability $1 - c_pn^{-10}$, where we used (\ref{eq:bound_innerproduct}) to bound $\|\X_\S^t\|_1\le \frac{\sqrt{k}}{c_\star}((\X^t)^\top\x^\star)$ in the first inequality, and for the second inequality we used the assumption $x^\star_{min}\le \frac{c_\star}{\sqrt{k}}$. Recalling the definition of the population gradient 
\begin{equation*}
\nabla f(\x) = \big(3\|\x\|_2^2 - 1\big)\x - 2\big(\x^\top\x^\star\big)\x^\star,
\end{equation*}
an application of the triangle inequality yields
\begin{equation}\label{eq:bound_gradient_abs}
|\nabla F(\X^t)_i| \le \big|\big(3\|\X^t\|_2^2 - 1\big)X_i^t\big| + \big|2.1\big((\X^t)^\top\x^\star\big)x^\star_i\big| \le 7.1\sqrt{2},
\end{equation}
provided that $\|\X^t\|_2\le \sqrt{2}$, where we used that $(\X^t)^\top\x^\star\le \|\X^t\|_2\|\x^\star\|_2$ by the Cauchy-Schwarz inequality. Hence, we have $|\eta \nabla F(\X^t)_i|\le \frac{1}{2}$, provided that $\eta \le \frac{1}{14.2\sqrt{2}}$. 

To bound the exponentials in (\ref{eq:update}), we use that $e^x \le 1+\frac{3}{2}x$ for $0\le x \le \frac{1}{2}$, and $1+x\le e^x$ for $-\frac{1}{2}\le x \le 0$. Without loss of generality, we assume $X_i^t\ge 0$; analogous bounds can be derived for $X_i^t<0$.
If $\nabla F(\X^t)_i<0$ (i.e.\ $X_i^{t+1}>X_i^t$), then we can bound
\begin{align*}
X_i^{t+1} - X_i^t &\le \frac{\sqrt{(X_i^t)^2 + \beta^2} + X_i^t}{2}\Bigl(-\frac{3}{2}\eta\nabla F(\X^t)_i\Bigr) - \frac{\sqrt{(X_i^t)^2 + \beta^2} - X_i^t}{2}\eta\nabla F(\X^t)_i \\
&\le -\frac{3}{2}\eta\nabla F(\X^t)_i\sqrt{(X_i^t)^2 + \beta^2}.
\end{align*}
If $\nabla F(\X^t)_i> 0$, we obtain the analogous lower bound
\begin{align*}
X_i^{t+1} - X_i^t &\ge -\frac{5}{4}\eta\nabla F(\X^t)_i\sqrt{(X_i^t)^2 + \beta^2}.
\end{align*}
Combining both bounds yields an upper bound on how much $X_i^t$ can increase or decrease in one iteration:
\begin{equation}\label{eq:flag}
|X_i^{t+1}-X_i^t| \le \frac{3}{2}\eta |\nabla F(\X^t)_i|\sqrt{(X_i^t)^2 + \beta^2}.
\end{equation}
Similarly, we can use the bounds $1+x\le e^x$ for $0\le x\le\frac{1}{2}$ and $e^x\le 1 + \frac{x}{2}$ for $-\frac{1}{2}\le x\le 0$ to obtain the lower bound
\begin{equation}\label{eq:flag_lb}
|X_i^{t+1}-X_i^t| \ge \frac{1}{2}\eta |\nabla F(\X^t)_i| |X_i^t|.
\end{equation}

\begin{proof}[Proof of Theorem \ref{theorem_discrete}]
The proof of Theorem \ref{theorem_discrete} relies on the following identity, which follows from the definition of the Bregman divergence and the mirror descent update (\ref{eq:mirror_descent_discrete}), to show that $D_{\Phi}(\x^\star, \X^t)$ decreases as claimed:
\begin{equation}\label{eq:bregman_difference}
D_{\Phi}(\x^\star,\X^{t+1}) - D_{\Phi}(\x^\star, \X^t) = -\eta\big\langle\nabla F(\X^t), \X^t - \x^\star \big\rangle + D_{\Phi}(\X^t,\X^{t+1}).
\end{equation}
The first term in (\ref{eq:bregman_difference}) has been bounded in the continuous-time case, and we can reuse those bounds. The second term measures the distance between the iterates $\X^t$ and $\X^{t+1}$ in terms of the Bregman divergence due to the discretization, and we will derive bounds for the second term similar to those shown in the continuous-time case for the first term. The function $\Phi$ is $1/2$-strongly convex with respect to the $\ell_2$-norm on the $\ell_2$-ball $\{\x\in\R^n:\|\x\|_2^2\le 2\}$ provided $\beta\le \sqrt{2}$, since the Hessian $\nabla^2\Phi(\x)$ is a diagonal matrix with entries 
\begin{equation*}
\nabla^2\Phi(\x)_{ii} = \frac{1}{\sqrt{x_i^2 + \beta^2}} \ge \frac{1}{\sqrt{2 + \beta^2}} \ge \frac{1}{2}.
\end{equation*}
Hence, we can bound
\begin{align}
D_{\Phi}(\X^t,\X^{t+1}) &= \Phi(\X^t) - \Phi(\X^{t+1}) - \langle \nabla\Phi(\X^{t+1}), \X^t - \X^{t+1}\rangle \nonumber\\ 
&\le \langle \nabla\Phi(\X^t) - \nabla\Phi(\X^{t-1}), \X^t - \X^{t+1}\rangle - \frac{1}{4}\|\X^t-\X^{t+1}\|_2^2 \nonumber\\
&\le \langle\eta\nabla F(\X^t), \X^t - \X^{t+1}\rangle - \frac{1}{4}\|\X^t-\X^{t+1}\|_2^2 \nonumber\\
&\le \eta\|\nabla F(\X^t)\|_2 \|\X^t - \X^{t+1}\|_2 - \frac{1}{4}\|\X^t - \X^{t+1}\|_2^2 \nonumber\\
&\le \eta^2\|\nabla F(\X^t)\|_2^2, \label{eq:discretization}
\end{align}
where the first inequality follows from strong convexity of $\Phi$, the second inequality from the definition of the mirror descent update, the third inequality from the Cauchy-Schwarz inequality, and for the last inequality we optimized the quadratic function in $\|\X^t - \X^{t+1}\|_2$.

As in continuous time, we divide the proof of Theorem \ref{theorem_discrete} into two stages bounded by 
\begin{align*}
T_1 &= \inf\biggl\{t>0 : \min_{i\in \S} \frac{|X_i^t|}{|x^\star_i|} > \frac{1}{2} \biggr\}, \text{ and}\\
T_2 &= \inf\biggl\{t>0: \frac{\|\X^t-\x^\star\|_2^2}{\|\x^\star\|_2^2} \le c^2\max\Bigl\{c_\star\sqrt{k}\delta, \; \frac{\sigma^2}{\|\x^\star\|_2^4} \frac{k\log n}{m}\Bigr\} \biggr\},
\end{align*}
respectively, where $\delta = \sqrt{n\beta}$, and (a) bound the length of the stage by using (\ref{eq:bregman_difference}) to bound $T_i$, and (b) show that off-support coordinates $\|\X_{\S^c}^t\|_1$ stay sufficiently small. Throughout the proof we will assume that the inequalities in Lemmas \ref{lemma:support1}--\ref{lemma:support3} are satisfied, which happens with probability at least $1-c_pn^{-10}$. As in the continuous-time case, the conditions for the simplified bounds in Lemmas \ref{lemma:support1} and \ref{lemma:support2} are satisfied for $t\le T_2$.

\textbf{Stage (i), part (a): $t\le T_1$, bound $T_1$}\\
Assume for now that we have already shown $\|\X_{\S^c}^t\|_1< \delta_1 = 2n\beta^{3/4}$ for all $t\le T_1$. As in continuous time, we consider the following two cases:
\begin{itemize}
\item \textbf{Case 1: $\|\X^t\|_2^2\le \frac{2}{5}$}\\
In this case, we have shown in (\ref{eq:bound1}) that 
\begin{equation*}
-\eta\big\langle \nabla F(\X^t), \X^t - \x^\star \big\rangle \le -\frac{\eta}{4}\big((\X^t)^\top\x^\star\big).
\end{equation*}
It remains to bound the term $D_\Phi(\X^t, \X^{t+1})$. 
Using the bounds (\ref{eq:bound_gradient_abs}), (\ref{eq:discretization}) and the inequality $(a+b)^2\le 2a^2+2b^2$, we can bound
\begin{align*}
D_{\Phi}(\X^t, \X^{t+1}) &\le 2\eta^2 \sum_{i=1}^n \bigl(3\|\X^t\|_2^2 - 1\bigr)^2(X_i^t)^2 + 2.1^2\bigl((\X^t)^\top\x^\star\bigr)^2(x^\star_i)^2 \\
&\le 2\eta^2 \bigl((\X^t)^\top\x^\star\bigr)^2\sum_{i=1}^n\bigl(12(X_i^t)^2 + 2.1^2(x^\star_i)^2\bigr) \\
&\le \frac{\eta}{8}\bigl((\X^t)^\top\x^\star\bigr)
\end{align*}
for $\eta \le \frac{1}{94}$, where we used (\ref{eq:claim3}) of Lemma \ref{lemma:support3} to bound $3\|\X^t\|_2^2-1\le \sqrt{3}\cdot 2((\X^t)^\top\x^\star)$ for the second inequality, and for the last inequality we used the assumptions $\|\X^t\|_2^2\le 2/5$ and $\|\x^\star\|_2=1$, and the Cauchy-Schwarz inequality to bound $(\X^t)^\top\x^\star\le \sqrt{2/5}$. 
Plugging this bound into (\ref{eq:bregman_difference}), we have
\begin{equation}\label{eq:bound1_discrete}
D_{\Phi}(\x^\star, \X^{t+1}) - D_{\Phi}(\x^\star, \X^t) \le -\frac{\eta}{8}\big((\X^t)^\top\x^\star\big) \le -\frac{c_\star\eta}{16\sqrt{k}},
\end{equation}
where for the second inequality we used (\ref{eq:bound_innerproduct}).

\item \textbf{Case 2: $\|\X^t\|_2^2\ge \frac{2}{5}$}\\
In this case, we have already shown that 
\begin{equation}\label{eq:bound21_discrete}
-\eta\big\langle \nabla F(\X^t), \X^t - \x^\star \big\rangle \le  -\frac{\eta}{25}
\end{equation}
if $4\|\X^t\|_2^2 -2 - \|\X^t-\x^\star\|_2^2\le 0$, and
\begin{equation}\label{eq:bound2_discrete}
-\eta\big\langle \nabla F(\X^t), \X^t - \x^\star\big\rangle \le -\frac{\eta}{2}\|\X^t - \x^\star\|_2^2
\end{equation}
if $4\|\X^t\|_2^2 -2 - \|\X^t-\x^\star\|_2^2> 0$. It remains to bound $D_\Phi(\X^t, \X^{t+1})$ in terms of the expressions in (\ref{eq:bound21_discrete}) and (\ref{eq:bound2_discrete}). As the bound corresponding to (\ref{eq:bound21_discrete}) can be obtained the same (but easier) way as for (\ref{eq:bound2_discrete}), we only show the computation for the latter case.

Substituting $\mathbf{Z} = \X^t - \x^\star$, we can write
\begin{align*}
|\nabla f(\X(t))_i| &= \big|\big(3\|\X^t\|_2^2 -1\big)Z_i + \big(3\|\mathbf{Z} + \x^\star\|_2^2 - 1 - 2\big((\mathbf{Z} + \x^\star)^\top\x^\star\big) \big)x^\star_i\big| \\
&= \big|\big(3\|\X^t\|_2^2 -1\big)Z_i + \big(\mathbf{Z}^\top(3\X^t+\x^\star)\big)x^\star_i\big| \\
&\le 5 |Z_i| + (1 + 3\sqrt{2})\|\mathbf{Z}\|_2 |x^\star_i|,
\end{align*}
where we used that $\|\X^t\|_2\le \sqrt{2}$ by (\ref{eq:claim2}) of Lemma \ref{lemma:support3}. By Lemma \ref{lemma:support1}, we can bound
\begin{equation*}
|\nabla F(\X^t)_i - \nabla f(\X^t)_i| \le \frac{0.1c_\star\|\X^t-\x^\star\|_2}{\sqrt{k}}\le 0.1\|\X^t-\x^\star\|_2|x^\star_i|,
\end{equation*}
for $t\le T_2$, so that we can bound 
\begin{equation*}
|\nabla F(\X^t)_i| \le 5.4\|\X^t - \x^\star\|_2 |x^\star_i| + 5 |X_i^t - x^\star_i|.
\end{equation*}
As in the previous case, we can use this bound together with (\ref{eq:bound_gradient_abs}) and (\ref{eq:discretization}) to obtain
\begin{align*}
D_{\Phi}(\X^t, \X^{t+1}) &\le 2\eta^2 \sum_{i=1}^n \bigl(5.4^2\|\X^t-\x^\star\|_2^2(x^\star_i)^2 + 5^2(\X^t_i - x^\star_i)^2\bigr) \\
&\le 2\eta^2 \|\X^t-\x^\star\|_2^2\biggl(5^2 + \sum_{i=1}^n5.4^2(x^\star_i)^2 \biggr) \\
&\le \frac{\eta}{4}\|\X^t - \x^\star\|_2^2
\end{align*}
for $\eta \le \frac{1}{434}$, where we used that $\|\x^\star\|_2 = 1$. Plugging this bound into (\ref{eq:bregman_difference}), we have
\begin{equation}\label{eq:bound2_difference}
D_{\Phi}(\x^\star, \X^{t+1}) - D_{\Phi}(\x^\star, \X^t) \le -\frac{\eta}{4}\|\X^t-\x^\star\|_2^2.
\end{equation}
\end{itemize}
We can now bound $T_1$.
Define
\begin{equation*}
T_0 = \inf\biggl\{t>0: \|\X^t\|_2^2>\frac{2}{5}\biggr\},
\end{equation*}
as the time until which Case 1 holds.
Then, the bound (\ref{eq:bound1_discrete}) from Case 1 shows that
\begin{equation*}
\min\{T_1,T_0\} \le \frac{16\sqrt{k}D_{\Phi}(\x^\star,\X^0)}{c_\star\eta} \le \frac{16}{c_\star\eta} k\log \frac{1}{\beta} + \frac{16\sqrt{k}}{c_\star\eta} \le \frac{c_2}{2\eta}k\log\frac{1}{\beta},
\end{equation*}
provided that $c_2\ge \frac{32}{c_\star} + \frac{32}{c_\star\sqrt{k}\log \frac{1}{\beta}}$, where we used (\ref{eq:bregman_divergence_initial}) to bound the initial Bregman divergence $D_{\Phi}(\x^\star, \X^0)$.
If $T_1<T_0$, then we have bounded $T_1$ as desired.
Otherwise, we can use (\ref{eq:bound2_difference}) and the bound corresponding to (\ref{eq:bound21_discrete}) to control $T_1 - T_0$. As in the continuous-time case, the bound corresponding to (\ref{eq:bound21_discrete}) can apply for at most $t\le 50\sqrt{k}\log \frac{1}{\beta} + 50$ iterations.
Following the same steps as in the continuous-time case, we can derive the bounds (\ref{eq:boundbreg1}) and (\ref{eq:boundbreg2}), which allow us to bound the $\ell_2$-distance $\|\X^t-\x^\star\|_2^2$ in terms of the Bregman divergence $D_{\Phi}(\x^\star,\X^t)$. Combined with (\ref{eq:bound2_difference}), this gives
\begin{equation*}
D_{\Phi}(\x^\star, \X^{t+1}) - D_{\Phi}(\x^\star, \X^t) \le -\frac{\eta}{4} \|\X^t-\x^\star\|_2^2 \le -\frac{c_\star^2\eta}{16k\log \frac{1}{\beta}} D_{\Phi}(\x^\star,\X^t),
\end{equation*}
which shows that $D_{\Phi}(\x^\star,\X^t)$ decreases at least linearly at the rate $(1 - c_\star^2\eta/(16k\log \frac{1}{\beta}))$ for $T_0 \le t < T_1$. We have
\begin{equation*}
D_{\Phi}(\x^\star,\X^{T_0}) \le D_{\Phi}(\x^\star,\X^0) \le \sqrt{k}\log \frac{1}{\beta}+1,
\end{equation*}
and, for $t\le T_1$,
\begin{equation*}
\frac{c_\star^2}{4k}\le \|\X^t-\x^\star\|_2^2 \le 2\sqrt{2+\beta^2}\; D_{\Phi}(\x^\star,\X^t)\le 3D_{\Phi}(\x^\star,\X^t),
\end{equation*}
where for the second inequality we used the bound (\ref{eq:lemma2lb}) of Lemma \ref{lemma2} and $\|\X^t\|_\infty\le \sqrt{2}$ by Lemma \ref{lemma:support3}.
This implies
\begin{equation*}
T_1-T_0 \le \frac{16k\log\frac{1}{\beta}}{c_\star^2\eta} \log \biggl(\frac{12}{c_\star^2} k^{\frac{3}{2}}\log \frac{1}{\beta} + \frac{12}{c_\star^2}k\biggr) \le \frac{c_2}{2\eta}k\log\biggl(\frac{1}{\beta}\biggr) \log\biggl(k\log \frac{1}{\beta}\biggr)
\end{equation*}
provided that $c_2\ge \frac{48}{c_\star^2} + \frac{32}{c_\star^2}\log\frac{24}{c_\star^2}$.

\textbf{Stage (i), part (b): $t\le T_1$, bound $\|\X_{\S^c}^t\|_1$}\\
As in the continuous-time case, let $I_1\in \S$ be the last coordinate which satisfies $|X_i^t|\ge \frac{1}{2}|x^\star_i|$ for some $t>0$, and let $j\notin \S$. To simplify notation, assume $X_{I_1}^t, X_j^t > 0$, in which case we can show $\nabla F(\X^t)_{I_1}<0$ for all $t\le T_1$ as in the continuous-time case.

As shown in (\ref{eq:stage1b}), we have (using different constants in Lemma \ref{lemma:support1})
\begin{equation}
\label{eq:stage1b_discrete}
|\nabla F(\X^t)_j| \le \frac{\sqrt{2}}{30} |\nabla F(\X^t)_{I_1}|
\end{equation}
for all $t\le T_1$. The idea of the proof is the same as in continuous time: the growth of both $X_{I_1}^t$ and $X_j^t$ is bounded by exponentials with different (time-varying) rates. Using (\ref{eq:flag_lb}), we have
\begin{align}\label{eq:stage1b_discrete_i}
X_{I_1}^{t+1} \ge X_{I_1}^t - \frac{1}{2}\eta \nabla F(\X^t)_{I_1} X_{I_1}^t \ge \Bigl(1 - \frac{1}{2}\eta \nabla F(\X^t)_{I_1}\Bigr)X_{I_1}^t,
\end{align}
and similarly, using (\ref{eq:flag}), we have, for $X_j^t\ge \beta$, 
\begin{align}\label{eq:stage1b_discrete_j}
X_j^{t+1} &\le \Bigl(1 + \frac{3}{2}\sqrt{2}\eta |\nabla F(\X^t)_j|\Bigr)X_j^t.
\end{align}
Denote $G_t=-\nabla F(\X^t)_{I_1}$, and let $T_\beta = \min\{t: |X_{I_1}^t|\ge \beta\}$. Then, we use (\ref{eq:stage1b_discrete_i}) to bound
\begin{equation*}
X_{I_1}^{T_1} \ge X_{I_1}^{T_\beta}\prod_{t=T_\beta}^{T_1}\Bigl(1 + \frac{1}{2}\eta G_t\Bigr),
\end{equation*}
and, since $|X_j^{T_\beta}|\le 2\beta$ and $\sqrt{(X_j^{T_\beta})^2 + \beta^2} \le \sqrt{2}|X_j^{T_\beta}|$ for $|X_j^{T_\beta}| \ge \beta$, we can use (\ref{eq:stage1b_discrete}) and (\ref{eq:stage1b_discrete_j}) to bound
\begin{align*}
X_j^{T_1} &\le 2\beta \prod_{t=T_\beta}^{T_1} \Bigl(1 + \frac{1}{10}\eta G_t\Bigr) \le 2\beta \Biggl(\prod_{t=T_\beta}^{T_1}1 + \frac{1}{2}\eta G_t\Biggr)^{\frac{1}{4}} \le 2\beta\Biggl(\frac{X_{i_1}^{T_1}}{X_{i_1}^{T_\beta}}\Biggr)^{\frac{1}{4}} \le 2\beta \biggl(\frac{1}{2\beta}\biggr)^{\frac{1}{4}} \le 2\beta^{\frac{3}{4}},
\end{align*}
where for the second inequality we used the inequality $\frac{x}{x+1}\le \log (1+x)\le x$ for $x\ge 0$, with which we have, for $\eta G_t \le \frac{1}{2}$,
\begin{equation*}
\log\Bigl(1 + \frac{1}{10}\eta G_t\Bigr) \le \frac{1}{10}\eta G_t \le \frac{1}{4} \frac{\eta G_t/2}{1 + \eta G_t/2} \le \frac{1}{4} \log \Bigl(1 + \frac{1}{2}\eta G_t\Bigr).
\end{equation*}

\textbf{Stage (ii), part (a): $T_1<t\le T_2$, bound $D_{\Phi}(\x^\star,\X^t)$}\\
As in the continuous-time case, we can use Lemma \ref{lemma2} to bound
\begin{equation*}
D_{\Phi}(\x^\star,\X^t) \le \frac{2\sqrt{k}}{c_\star}\|\X^t-\x^\star\|_2^2,
\end{equation*}
and Lemma \ref{lemma:support3} to bound $D_{\Phi}(\x^\star,\X^{T_1}) \le \frac{6\sqrt{k}}{c_\star}$. With this, inequality (\ref{eq:bound2_difference}) reads
\begin{equation*}
D_{\Phi}(\x^\star,\X^{t+1}) - D_{\Phi}(\x^\star, \X^t) \le -\frac{c_\star\eta}{8\sqrt{k}}D_{\Phi}(\x^\star,\X^t).
\end{equation*}
Hence, we have for $T_1\le t\le T_2$,
\begin{align*}
D_{\Phi}(\x^\star,\X^t) &\le D_{\Phi}(\x^\star,\X^{T_1}) \biggl(1-\frac{c_\star\eta}{8\sqrt{k}}\biggr)^{t-T_1} 
\le \frac{6\sqrt{k}}{c_\star} \biggl(1-\frac{c_\star\eta}{8\sqrt{k}}\biggr)^{t-T_1}.
\end{align*}
Together with the fact that $\|\X^t\|_\infty^2\le 2$ by Lemma \ref{lemma:support3}, bound (\ref{eq:lemma2lb}) of Lemma \ref{lemma2} implies
\begin{equation*}
D_{\Phi}(\x^\star, \X^t) \ge \frac{1}{3}\|\X^t - \x^\star\|_2^2 \ge \frac{c^2c_\star\sqrt{k}\delta}{3}
\end{equation*}
for $t\le T_2$, so we can bound
\begin{equation*}
T_2-T_1 \le \frac{8\sqrt{k}}{c_\star\eta}\log \frac{18}{c^2c_\star^2\delta}.
\end{equation*}

\textbf{Stage (ii), part (b): $T_1<t\le T_2$, bound $\|\X_{\S^c}^t\|_1$}\\
Recall that we have shown $|X_j^t|\le 2\beta^{\frac{3}{4}}$ for all $j\notin \S$ and $t\le T_1$. As in the continuous-time case, we can use Lemma \ref{lemma:support1} to show
\begin{equation*}
|\nabla F(\X^t)_j| \le \frac{c_\star\|\X_\S^t\|_1}{64k} \le \frac{c_\star}{32\sqrt{2k}},
\end{equation*}
where we used that $\|\X_\S^t\|_1\le \sqrt{k}\|\X_\S^t\|_2 \le \sqrt{2k}$.
As in Stage (i), we can bound, for $t\le T_2$ (as $\sqrt{x^2+\beta^2}\le \sqrt{2}x$ for $x\ge \beta$),
\begin{align*}
|X_j(t)| &\le 2\beta^{\frac{3}{4}}\biggl(1 + \frac{3}{\sqrt{2}}\eta \frac{c_\star}{32\sqrt{k}}\biggr)^{T_2-T_1} \\
&\le 2\beta^{\frac{3}{4}}\exp\biggl(\frac{\eta c_\star}{32\sqrt{k}} \frac{8\sqrt{k}}{c_\star\eta}\log \frac{18}{c^2c_\star^2\delta}\biggr) \\
&\le 2\beta^{\frac{3}{4}}\biggl(\frac{18}{c^2c_\star^2\delta}\biggr)^{\frac{1}{4}} \\
 &\le \frac{\delta}{n},
\end{align*}
provided that $\beta \le (c^2c_\star^2 / 288)^2/ n^3$, where we used the definition $\delta = \sqrt{n\beta}$. This completes the proof that $\|\X_{\S^c}^t\|_1\le \delta$ for all $t\le T_2$.
\end{proof}

%% file: proof_supporting_lemmas.tex
%%%%%%%%%%%%%%%%%%%%%%%%%%%%%%%%%%%%%%%%%%%%%%
%%%% Proof of supporting lemmas:
\section{Proof of supporting lemmas}
\label{supp:proof_supporting_lemmas}
In this section, we provide the proofs of the supporting lemmas stated in Section \ref{proof:theorem3}.
Throughout this section, we will assume $\|\x^\star\|_2 = 1$ for notational simplicity's sake. The general case $\|\x^\star\|_2\neq 1$ follows by writing $\x^\star = \frac{\x^\star}{\|\x^\star\|_2}\|\x^\star\|_2$, $\x = \frac{\x}{\|\x^\star\|_2}\|\x^\star\|_2$ and $\varepsilon_j = \frac{\varepsilon_j}{\|\x^\star\|_2}\|\x^\star\|_2$ in what follows.

\begin{proof}[Proof of Lemma \ref{lemma:support1}]
To prove Lemma \ref{lemma:support1}, we will we will bound the difference $|\nabla F(\x)_i - \nabla f(\x)_i|$ for any $i\in [n]$. The lemma then follows by taking a union bound. First, we write $\mathbf{w}\in \R^n$ for the vector $\x_\S$ padded with zeroes, that is $w_i = x_i$ for $i\in \S$ and $w_i = 0$ otherwise. For convenience, we denote by 
\begin{equation*}
\nabla \widetilde{F}(\x) = \frac{1}{m}\sum_{j=1}^m\bigl((\A_j^\top\x)^2 - (\A_j^\top\x^\star)^2\bigr)\bigl(\A_j^\top\x\bigr)\mathbf{A}_j
\end{equation*}
the gradient if the measurements were noiseless. Then, as $\nabla f(\x) = \E[\nabla F(\x)] = \E[\nabla \widetilde{F}(\x)]$ since the random variables $\{\varepsilon_j\}_{j=1}^m$ are centered and independent of $\{\mathbf{A}_j\}_{j=1}^m$, we can decompose and bound the difference by
\begin{align}\label{eq:split}
\big|\nabla F(\x)_i - \nabla f(\x)_i\big| &\le \big|\nabla F(\x)_i - \nabla\widetilde{F}(\x)_i\big| + \big|\nabla \widetilde{F}(\x)_i - \nabla \widetilde{F}(\mathbf{w})_i\big| \nonumber\\
&\quad + \big|\nabla \widetilde{F}(\mathbf{w})_i - \E[\nabla \widetilde{F}(\mathbf{w})_i]\big| + \big|\E[\nabla \widetilde{F}(\mathbf{w})_i] - \E[\nabla \widetilde{F}(\x)_i]\big|,
\end{align}
and we will bound the four terms separately.

\textbf{Step 1: Bound the term $|\nabla F(\x)_i - \nabla\widetilde{F}(\x)_i|$}\\
We begin by bounding the first term of (\ref{eq:split}) by 
\begin{equation*}
c\sigma\sqrt{\frac{\log n}{m}}\min\Bigl\{\|\x\|_1,\; 1 + \delta + \sqrt{k}\|\x_\S - \x_\S^\star\|_2\Bigr\}.
\end{equation*}
To obtain the second bound in the minimum, we write $\mathbf{z} = \x-\x^\star$ and bound
\begin{align*}
\big|\nabla F(\x)_i - \nabla\widetilde{F}(\x)_i\big| &= \Bigg|\frac{1}{m}\sum_{j=1}^m\varepsilon_j(\A_j^\top\x)A_{ji}\Bigg| \\
&\le \Bigg|\frac{1}{m}\sum_{j=1}^m\varepsilon_j(\A_{j,\S}^\top\x_\S)A_{ji}\Bigg| + \Bigg|\frac{1}{m}\sum_{j=1}^m\varepsilon_j(\A_{j,\S^c}^\top\x_{\S^c})A_{ji}\Bigg| \\
&\le \Bigg|\frac{1}{m}\sum_{j=1}^m\varepsilon_j(\A_{j,\S}^\top\mathbf{z}_\S)A_{ji}\Bigg| + \Bigg|\frac{1}{m}\sum_{j=1}^m\varepsilon_j(\A_{j,\S}^\top\x^\star_\S)A_{ji}\Bigg| + \Bigg|\sum_{l\notin \S}x_l\frac{1}{m}\sum_{j=1}^m\varepsilon_jA_{jl}A_{ji}\Bigg| \\
&\le \big(\|\mathbf{z}_{\S}\|_1 + \|\x_{\S^c}\|_1\big)\max_{l\in [n]}\Bigg|\frac{1}{m}\sum_{j=1}^m\varepsilon_jA_{jl}A_{ji}\Bigg|  + \Bigg|\frac{1}{m}\sum_{j=1}^m\varepsilon_j(\A_{j,\S}^\top\x^\star_\S)A_{ji}\Bigg| \\
&\le c\sigma \sqrt{\frac{\log n}{m}}\big(\|\mathbf{z}_\S\|_1 + 1 + \delta\big),
\end{align*}
with probability $1-\frac{c_p}{3}n^{-11}$ for any constant $c>0$ that is at least the universal constant from Lemma \ref{lemma:tech3}, where we used H\"{o}lder's inequality in the penultimate line. The second bound in the minimum follows as $\|\mathbf{z}_\S\|_1\le \sqrt{k}\|\mathbf{z}_\S\|_2$. The first bound in the minimum follows by directly bounding $\frac{1}{m}\sum_{j=1}^m\varepsilon_j(\A_j^\top\x)A_{ji}$ using H\"{o}lder's inequality and Lemma \ref{lemma:tech3} as above, without substituting $\mathbf{z} = \x - \x^\star$.

\textbf{Step 2: Bound the term $|\nabla \widetilde{F}(\x)_i - \nabla \widetilde{F}(\mathbf{w})_i|$}\\
Next, we bound the second term of (\ref{eq:split}) by $\frac{c}{2}\delta$. We can write
\begin{equation*}
\nabla \widetilde{F}(\x)_i = \frac{1}{m}\sum_{j=1}^m\bigl((\A_{j,\S}^\top\x_\S + \A_{j,\S^c}\x_{\S^c})^3 - (\A_{j,\S}^\top\x_\S + \A_{j,\S^c}^\top\x_{\S^c})(\A_j^\top\x^\star)^2\bigr)A_{ji}.
\end{equation*}
Then, we have
\begin{align}\label{eq:split2}
\nabla \widetilde{F}(\x)_i - \nabla \widetilde{F}(\mathbf{w})_i &= \frac{3}{m}\sum_{j=1}^m A_{ji}(\A_{j,\S}^\top\x_\S)^2(\A_{j,\S^c}^\top\x_{\S^c}) + \frac{3}{m}\sum_{j=1}^m A_{ji}(\A_{j,\S}^\top\x_\S)(\A_{j,\S^c}^\top\x_{\S^c})^2 \nonumber\\
&\quad + \frac{1}{m}\sum_{j=1}^m A_{ji}(\A_{j,\S^c}^\top\x_{\S^c})^3 - \frac{1}{m}\sum_{j=1}^m A_{ji}(\A_{j,\S^c}^\top\x_{\S^c})(\A_j^\top\x^\star)^2.
\end{align}
These four terms can be bounded as follows: for the first term, we have
\begin{align*}
\Bigg|\frac{1}{m}\sum_{j=1}^m A_{ji}(\A_{j,\S}^\top\x_\S)^2(\A_{j,\S^c}^\top\x_{\S^c})\Bigg| &= \Bigg|\sum_{l\notin \S}x_l \frac{1}{m}\sum_{j=1}^mA_{ji}A_{jl}(\A_{j,\S}^\top\x_{\S})^2\Bigg|\\
&\le \|\x_{\S^c}\|_1 \max_{l\notin \S}\Bigg|\frac{1}{m}\sum_{j=1}^mA_{ji}A_{jl}(\A_{j,\S}^\top\x_{\S})^2\Bigg| \\
&\le \|\x_{\S^c}\|_1 \max_{l\notin \S}\sqrt{\frac{1}{m}\sum_{j=1}^mA_{ji}^2A_{jl}^2}\sqrt{\frac{1}{m}\sum_{j=1}^m(\A_{j,\S}^\top\x_\S)^4}, 
\end{align*}
where we used H\"{o}lder's inequality to obtain both inequalities. The first sum is bounded by Lemma \ref{lemma:tech3}: recalling $m\ge c_s(\gamma)k^2\log^2 n$, we have with probability at least $1-c_2n^{-13}$, where $c_2$ is the universal constant from Lemma \ref{lemma:tech3},
\begin{equation*}
\max_{l\notin \S}\Bigg|\frac{1}{m}\sum_{j=1}^mA_{ji}^2A_{jl}^2\Bigg| \le 1 + \frac{1}{k} \le 2.
\end{equation*}
By Lemma \ref{lemma:tech1} with $t=5\sqrt{\log n}$, we can bound, with probability $1-4n^{-12.5}$,
\begin{equation*}
\sqrt{\frac{1}{m}\sum_{j=1}^m (\A_{j,\S}^\top\x_\S)^4} \le \frac{1}{\sqrt{m}}\Bigl((3m)^{\frac{1}{4}} + \sqrt{k} + 5\sqrt{\log n}\Bigr)^2 \le 11
\end{equation*}
for all $\x\in\mathcal{X}$, where we used $5\sqrt{\log n} \le m^{\frac{1}{4}}$, which holds if $c_s(\gamma) \ge 5^4\min\{k^{-2},\; \log^{-3}n\}$. 
Put together, this gives 
\begin{equation*}
\frac{1}{m}\sum_{j=1}^m A_{ji}(\A_{j,\S}^\top\x_\S)^2(\A_{j,\S^c}^\top\x_{\S^c}) \le c'\delta,
\end{equation*}
where $c' = 11\sqrt{2}$.
The other terms in (\ref{eq:split2}) can be bounded the same way. For instance, we can write
\begin{equation*}
\Bigg|\frac{1}{m}\sum_{j=1}^m A_{ji}(\A_{j,\S}^\top\x_S)(\A_{j,\S^c}^\top\x_{\S^c})^2\Bigg| = \Bigg|\sum_{l\notin \S}x_l\sum_{r\notin \S}x_r \frac{1}{m}\sum_{j=1}^mA_{ji}A_{jl}A_{js}(\A_{j,\S}^\top\x_\S)\Bigg|
\end{equation*}
and, following the same steps as above, we obtain the bounds
\begin{align*}
\Bigg|\frac{1}{m}\sum_{j=1}^m A_{ji}(\A_{j,\S^c}^\top\x_{\S^c})(\A_j^\top\x^\star)^2\Bigg| &\le c' \delta,\\
\Bigg|\frac{1}{m}\sum_{j=1}^m A_{ji}(\A_{j,\S}^\top\x_S)(\A_{j,\S^c}^\top\x_{\S^c})^2\Bigg| &\le  c''\delta^2,\\
\Bigg|\frac{1}{m}\sum_{j=1}^m A_{ji}(\A_{j,\S^c}^\top\x_{\S^c})^3\Bigg| &\le c'''\delta^3,
\end{align*}
for all $\x\in\mathcal{X}$ with probability $1-3(c_2+4)n^{-12.5}$. Recalling that $\delta\le \frac{1}{4}$, we have for any constant $c>0$ satisfying $4c'\delta + 3c''\delta^2 + c'''\delta^3 \le \frac{c}{2}\delta$,
\begin{equation*}
\big|\nabla \widetilde{F}(\x)_i - \nabla \widetilde{F}(\mathbf{w})_i\big| \le \frac{c}{2}\delta \quad \text{for all } \x\in\mathcal{X}
\end{equation*}
with probability at least $1-\frac{c_p}{3}n^{-11}$.

\textbf{Step 3: Bound the term $|\E[\nabla \widetilde{F}(\x)_i] - \E[\nabla \widetilde{F}(\mathbf{w})_i]|$}\\
We use the Cauchy-Schwarz inequality to bound each of the four terms in (\ref{eq:split2}) in expectation. The first term in (\ref{eq:split2}) can be bounded by
\begin{align*}
\E\Bigg[\frac{1}{m}\sum_{j=1}^m A_{ji}(\A_{j,\S}^\top\x_\S)^2(\A_{j,\S^c}^\top\x_{\S^c})\Bigg] &\le \E\big[A_{1i}^2\big]^{\frac{1}{2}} \E\big[(\A_{1,\S}^\top\x_\S)^4(\A_{1,\S^c}^\top\x_{\S^c})^2\big]^{\frac{1}{2}} \\
&= \Bigl(\E\big[(\A_{1,\S}^\top\x_\S)^4\big]\E\big[(\A_{1,\S^c}^\top\x_{\S^c})^2\big]\Bigr)^{\frac{1}{2}} \\
&\le \sqrt{3\|\x_{\S^c}\|_2^2} \\
&\le \sqrt{3}\delta
\end{align*}
for all $\x\in\mathcal{X}$, where we used that $\A_{1,\S}^\top\x_\S \sim \gauss(0,\|\x_\S\|_2^2)$ and $\A_{1,\S^c}^\top\x_{\S^c} \sim \gauss(0,\|\x_{\S^c}\|_2^2)$ are independent, and that $\|\x_\S\|_2\le 1$ and $\|\x_{\S^c}\|_2\le \delta$ by the definition of $\mathcal{X}$. Similarly, we can bound the other terms:
\begin{align*}
\E\Bigg[\frac{1}{m}\sum_{j=1}^m A_{ji}(\A_{j,\S^c}^\top\x_{\S^c})(\A_j^\top\x^\star)^2\Bigg] &\le \sqrt{3}\delta,\\
\E\Bigg[\frac{1}{m}\sum_{j=1}^m A_{ji}(\A_{j,\S}^\top\x_\S)(\A_{j,\S^c}^\top\x_{\S^c})^2\Bigg] &\le  \sqrt{3}\delta^2,\\
\E\Bigg[\frac{1}{m}\sum_{j=1}^m A_{ji}(\A_{j,\S^c}^\top\x_{\S^c})^3\Bigg] &\le \sqrt{15}\delta^3.
\end{align*}
This completes the proof that 
\begin{equation*}
\big|\E[\nabla \widetilde{F}(\x)_i] - \E\big[\nabla \widetilde{F}(\mathbf{w})_i]\big| \le \frac{c}{2}\delta.
\end{equation*}

\textbf{Step 4, part (a): Bound the term $|\nabla \widetilde{F}(\mathbf{w})_i - \E[\nabla \widetilde{F}(\mathbf{w})_i]|$ by $\gamma\frac{\|\x_\S\|_1}{k}$}\\
Finally, we need to bound the term $|\nabla \widetilde{F}(\mathbf{w})_i - \E[\nabla \widetilde{F}(\mathbf{w})_i]|$ in (\ref{eq:split}) for all $\x\in\mathcal{X}$ and $i\in [n]$ with probability $1-\frac{c_p}{3}n^{-10}$, which then completes the proof of Lemma \ref{lemma:support1}. We begin by showing the bound $\gamma\frac{\|\x_\S\|_1}{k}$.

We decompose $\nabla \widetilde{F}(\mathbf{w})_i$ in a straightforward, albeit somewhat lengthy manner. We have
\begin{align*}
\nabla \widetilde{F}(\mathbf{w})_i &= \frac{1}{m}\sum_{j=1}^m\bigl((\A_j^\top\mathbf{w})^3 - (\A_j^\top\mathbf{w})(\A_j^\top\x^\star)^2\bigr)A_{ji} \\
&= \big(w_i^3 - w_i(x^\star_i)^2\big)\frac{1}{m}\sum_{j=1}^mA_{ji}^4 + \big(3w_i^2 - (x^\star_i)^2\big)\frac{1}{m}\sum_{j=1}^mA_{ji}^3 (\A_{j,-i}^\top\mathbf{w}_{-i}) \\
&\quad - 2w_ix^\star_i\frac{1}{m}\sum_{j=1}^mA_{ji}^3(\A_{j,-i}^\top\x^\star_{-i}) + 3w_i\frac{1}{m}\sum_{j=1}^mA_{ji}^2(\A_{j,-i}^\top\mathbf{w}_{-i})^2 \\
&\quad - 2x^\star_i\frac{1}{m}\sum_{j=1}^mA_{ji}^2(\A_{j,-i}^\top\mathbf{w}_{-i})(\A_{j,-i}^\top\x^\star_{-i}) - w_i\frac{1}{m}\sum_{j=1}^mA_{ji}^2(\A_{j,-i}^\top\x^\star_{-i})^2 \\
&\quad + \frac{1}{m}\sum_{j=1}^mA_{ji} (\A_{j,-i}^\top\mathbf{w}_{-i})^3 - \frac{1}{m}\sum_{j=1}^mA_{ji}(\A_{j,-i}^\top\mathbf{w}_{-i})(\A_{j,-i}^\top\x^\star_{-i})^2 \\
&=: B_1 + B_2 + B_3 + B_4 + B_5 + B_6 + B_7 + B_8,
\end{align*}
and we will show that $|B_l-\E[B_l]|$ is small for all $l=1,...,8$. By Lemmas \ref{lemma:tech3} and \ref{lemma:tech7}, all the following statements hold with probability $1-\frac{c_p}{3}n^{-10}$ for all $\mathbf{w}\in \R^n$ with $\|\mathbf{w}\|_2\le 1$ and $\mathbf{w}_{\S^c}=\mathbf{0}$, and for all $i=1,...,n$. We write the bounds in terms of $k$ instead of $m$ using the assumption $m\ge c_s(\gamma)k^2\log^2 n$.
\begin{itemize}
\item For the first term, we have $\E[B_1] = 3(w_i^3 - w_i(x^\star_i)^2)$, and 
\begin{equation*}
\big|B_1 - \E[B_1]\big| = \Bigg|\big(w_i^3 - w_i(x^\star_i)^2\big)\Biggl(\frac{1}{m}\sum_{j=1}^mA_{ji}^4 - 3\Biggr)\Bigg| \le \frac{\gamma\|\mathbf{w}\|_1}{8k}
\end{equation*}
by (\ref{eq:tech3_1}) of Lemma \ref{lemma:tech3}, where we used $|w_i^3 - w_i(x^\star_i)^2| \le \|\mathbf{w}\|_1$.

\item For the second term, we have $\E[B_2] = 0$, and
\begin{equation*}
|B_2| \le \big|3w_i^2 - (x^\star_i)^2\big| \|\mathbf{w}\|_1 \max_{l\neq i} \Bigg|\frac{1}{m}\sum_{j=1}^mA_{ji}^3A_{jl}\Bigg| \le \frac{\gamma\|\mathbf{w}\|_1}{8k}
\end{equation*}
by (\ref{eq:tech3_1}) of Lemma \ref{lemma:tech3}.

\item For the third term, we have $\E[B_3] = 0$, and
\begin{equation*}
|B_3| \le |2w_ix^\star_i|\frac{\gamma}{16k} \le \frac{\gamma\|\mathbf{w}\|_1}{8k}
\end{equation*}
by (\ref{eq:tech3_2}) of Lemma \ref{lemma:tech3}.

\item For the fourth term, we have $\E[B_4] = 3w_i\|\mathbf{w}_{-i}\|_2^2$, and
\begin{align*}
\big|B_4 - \E[B_4]\big| &\le 3|w_i|\Bigg|\sum_{l\neq i}w_l^2\Biggl(\frac{1}{m}\sum_{j=1}^mA_{ji}^2A_{jl}^2 - 1\Biggr) + \sum_{l\neq i}w_l\sum_{r\neq l,i} w_r \frac{1}{m}\sum_{j=1}^mA_{ji}^2A_{jl}A_{jr}\Bigg| \\
&\le \frac{\gamma\|\mathbf{w}\|_1}{16k} + \frac{\gamma\|\mathbf{w}\|_1}{16k} \\
&= \frac{\gamma\|\mathbf{w}\|_1}{8k},
\end{align*}
where we used Lemma (\ref{eq:tech3_1}) of \ref{lemma:tech3} to bound the first and (\ref{eq:tech7_2}) of Lemma \ref{lemma:tech7} to bound the second term.

\item For the fifth term, we have $\E[B_5] = -2x^\star_i\sum_{l\neq i}w_lx^\star_l $, and
\begin{align*}
\big|B_5 - \E[B_5]\big| &\le |2x^\star_i|\Bigg|\sum_{l\neq i} w_l \Biggl(\frac{1}{m}\sum_{j=1}^mA_{ji}^2A_{jl}(\A_{j,-i}^\top\x^\star_{-i}) - x^\star_l\Biggr)\Bigg| \le \frac{\gamma\|\mathbf{w}\|_1}{8k}
\end{align*}
where we used H\"{o}lder's inequality and (\ref{eq:tech3_2}) of Lemma \ref{lemma:tech3}.

\item For the sixth term, we have $\E[B_6] = -w_i\|\x^\star_{-i}\|_2^2$, and
\begin{align*}
\big|B_6 - \E[B_6]\big| &\le |w_i|\frac{\gamma}{8k} \le \frac{\gamma\|\mathbf{w}\|_1}{8k}
\end{align*} 
by (\ref{eq:tech3_3}) of Lemma \ref{lemma:tech3}.

\item For the seventh term, we have $\E[B_7] = 0$, and
\begin{align*}
|B_7| &\le \frac{\gamma\|\mathbf{w}\|_1}{8k}
\end{align*} 
by (\ref{eq:tech7_1}) of Lemma \ref{lemma:tech7}.

\item Finally, for the eighth term, we have $\E[B_8] = 0$, and
\begin{align*}
|B_8|&\le \Bigg|\sum_{l\neq i} w_l \frac{1}{m}\sum_{j=1}^mA_{ji}A_{jl}(\A_{j,-i}^\top\x^\star_{-i})^2\Bigg| \le \frac{\gamma\|\mathbf{w}\|_1}{8k}
\end{align*}
where we used H\"{o}lder's inequality and (\ref{eq:tech3_3}) of Lemma \ref{lemma:tech3}.
\end{itemize}

\textbf{Step 4, part (b): Bound the term $|\nabla \widetilde{F}(\mathbf{w})_i - \E[\nabla \widetilde{F}(\mathbf{w})_i]|$ by $\gamma\frac{\|\x_\S-\x^\star_\S\|_2}{\sqrt{k}}$}\\
To show this bound, we parametrize $\mathbf{z} = \mathbf{w} - \x^\star$ and write
\begin{align*}
\nabla \widetilde{F}(\mathbf{w})_i &= \frac{1}{m}\sum_{j=1}^m\bigl((\A_j^\top(\mathbf{z} + \x^\star))^3 - (\A_j^\top(\mathbf{z} + \x^\star))(\A_j^\top\x^\star)^2\bigr)A_{ji} \\
&= \frac{1}{m}\sum_{j=1}^m\bigl((\A_j^\top\mathbf{z})^3 + 3(\A_j^\top\mathbf{z})^2(\A_j^\top\x^\star) +2(\A_j^\top\mathbf{z})(\A_j^\top\x^\star)^2\bigr)A_{ji}. 
\end{align*}
We can decompose this expression as before and observe that each term depends at least linearly on $\|\mathbf{z}\|_2$. Further, because $\mathbf{z}_{\S^c} = \mathbf{w}_{\S^c} = \mathbf{0}$, we have $\|\mathbf{z}\|_1 \le \sqrt{k}\|\mathbf{z}\|_2$. Hence, we obtain the bound $\gamma\frac{\|\mathbf{z}\|_1}{k}\le \gamma\frac{\|\mathbf{z}\|_2}{\sqrt{k}}$ the same way as above using the bounds in Lemmas \ref{lemma:tech3} and \ref{lemma:tech7}. We omit the details to avoid repetition. 

All in all, putting everything together we have, with probability $1-\frac{c_p}{3}n^{-10}$,
\begin{equation*}
\big|\nabla \widetilde{F}(\mathbf{w})_i - \E[\nabla \widetilde{F}(\mathbf{w})_i]\big| \le \gamma \min\biggl\{\frac{\|\mathbf{w}\|_1}{k},\; \frac{\|\mathbf{z}\|_2}{\sqrt{k}}\biggr\} = \gamma \min\biggl\{\frac{\|\x_\S\|_1}{k},\; \frac{\|\x_\S - \x^\star_\S\|_2}{\sqrt{k}}\biggr\},
\end{equation*}
for all $\mathbf{w}\in \R^n$ with $\|\mathbf{w}\|_2\le 1$ and $\mathbf{w}_{\S^c}=\mathbf{0}$, and all $i=1,...,n$, which completes the proof of Lemma \ref{lemma:support1}. Finally, the simplified bound can be obtained directly by plugging in the additional assumptions.
\end{proof}

\begin{proof}[Proof of Lemma \ref{lemma:support2}]
As in the proof of Lemma \ref{lemma:support1}, we assume $\|\x^\star\|_2=1$ for notational simplicity, and denote the gradient corresponding to noiseless measurements by $\nabla \widetilde{F}$. Then,
\begin{equation*}
\big|\langle \nabla F(\x) - \nabla f(\x), \x-\x^\star\rangle\big| \le \big|\langle \nabla F(\x) - \nabla \widetilde{F}(\x), \x-\x^\star\rangle\big| + \big|\langle \nabla \widetilde{F}(\x) - \nabla f(\x), \x-\x^\star\rangle\big|,
\end{equation*}
and we begin by bounding the first of the two terms. 

\textbf{Step 1: Bound the term $\big|\langle \nabla F(\x) - \nabla \widetilde{F}(\x), \x-\x^\star\rangle\big|$}\\
We begin by bounding the first term $\big|\langle \nabla F(\x) - \nabla \widetilde{F}(\x), \x-\x^\star\rangle\big|$ by
\begin{equation*}
c\sigma\sqrt{\frac{\log n}{m}}\Bigl(\delta + \min\{\|\x_\S\|_1,\; \sqrt{k}\|\x_\S-\x^\star_\S\|_2\}\Bigr).
\end{equation*}
To obtain the second bound in the minimum, we can substitute $\mathbf{z} = \x - \x^\star$ and write
\begin{align*}
\big|\langle \nabla F(\x) - \nabla \widetilde{F}(\x), \x-\x^\star\rangle\big| &= \Bigg|\frac{1}{m}\sum_{j=1}^m\varepsilon_j (\A_j^\top\x)(\A_j^\top(\x-\x^\star))\Bigg| \\
&\le \Bigg|\frac{1}{m}\sum_{j=1}^m\varepsilon_j (\A_j^\top\mathbf{z})^2\Bigg| + \Bigg|\frac{1}{m}\sum_{j=1}^m\varepsilon_j (\A_j^\top\mathbf{z})(\A_j^\top\x^\star)\Bigg|.
\end{align*}
Using Lemmas \ref{lemma:tech3} and \ref{lemma:tech7} we can bound the first term, with probability $1 - \frac{c_p}{3}n^{10}$, by
\begin{align*}
\Bigg|\frac{1}{m}\sum_{j=1}^m\varepsilon_j (\A_j^\top\mathbf{z})^2\Bigg| &\le \Bigg|\frac{1}{m}\sum_{j=1}^m\varepsilon_j (\A_{j,\S}^\top\mathbf{z}_\S)^2\Bigg| + 2 \Bigg|\frac{1}{m}\sum_{j=1}^m\varepsilon_j (\A_{j,\S}^\top\mathbf{z}_\S)(\A_{j,\S^c}^\top\mathbf{z}_{\S^c})\Bigg| \\
&\quad + \Bigg|\frac{1}{m}\sum_{j=1}^m\varepsilon_j (\A_{j,\S^c}^\top\mathbf{z}_{\S^c})^2\Bigg| \\
&\le \frac{c}{4}\sigma \|\mathbf{z}_\S\|_1\sqrt{\frac{\log n}{m}} + 2 \|\mathbf{z}_{\S^c}\|_1 \max_{l\notin \S}\Bigg|\sum_{s\in \S}z_s \frac{1}{m}\sum_{j=1}^m\varepsilon_jA_{js}A_{jl}\Bigg| \\
&\quad + \|\mathbf{z}_{\S^c}\|_1^2\max_{l,s\notin \S} \Bigg|\frac{1}{m}\sum_{j=1}^m\varepsilon_jA_{jl}A_{js}\Bigg| \\
&\le \frac{c}{4}\sigma \|\mathbf{z}_\S\|_1\sqrt{\frac{\log n}{m}} + 2\frac{c}{4}\sigma \|\mathbf{z}_\S\|_1\sqrt{\frac{\log n}{m}} \delta + \frac{c}{4}\sigma \sqrt{\frac{\log n}{m}}\delta^2,
\end{align*}
provided that $\frac{c}{4}$ is at least as large as the universal constants from Lemmas \ref{lemma:tech3} and \ref{lemma:tech7}, where we used 
\begin{equation*}
\frac{1}{m}\sum_{j=1}^m\varepsilon_j (\A_{j,\S}^\top\mathbf{z}_\S)^2 = \sum_{l\in \S}z_l^2\frac{1}{m}\sum_{j=1}^m\varepsilon_j A_{jl}^2 + \sum_{\substack{l,s\in \S \\ l\neq s}}z_lz_s\frac{1}{m}\sum_{j=1}^m\varepsilon_j A_{jl}A_{js}
\end{equation*}
and Lemmas \ref{lemma:tech3} and \ref{lemma:tech7} for the second inequality, and Lemma \ref{lemma:tech3} to bound the other terms. Similarly, we can bound
\begin{align*}
\Bigg|\frac{1}{m}\sum_{j=1}^m\varepsilon_j (\A_j^\top\mathbf{z})(\A_j^\top\x^\star)\Bigg| &\le \Bigg|\sum_{l=1}^nz_l\frac{1}{m}\sum_{j=1}^m\varepsilon_j A_{jl}(\A_j^\top\x^\star)\Bigg| \le \frac{c}{2}\sigma\big(\|\mathbf{z}_\S\|_1 + \delta\big)\sqrt{\frac{\log n}{m}},
\end{align*}
where we used Lemma \ref{lemma:tech3}. All in all, this yields with probability $1-\frac{c_p}{3}n^{-10}$, recalling that $\|\mathbf{z}_\S\|_1 \le \sqrt{k}\|\x_\S-\x^\star_\S\|_2$,
\begin{align*}
\big|\langle \nabla F(\x) - \nabla \widetilde{F}(\x), \x-\x^\star\rangle\big| \le c\sigma\sqrt{\frac{\log n}{m}} \Bigl(\delta + \sqrt{k}\|\x_\S-\x^\star_\S\|_2\Bigr).
\end{align*}
The first bound in the minimum can be obtained following exactly the same steps as above without substituting $\mathbf{z} = \x- \x^\star$.

\textbf{Step 2: Bound the term $\big|\langle \nabla \widetilde{F}(\x) - \nabla f(\x), \x-\x^\star\rangle\big|$ by $\gamma(\frac{\|\x_\S\|_1}{\sqrt{k}} + \delta )$}\\
Writing $\mathbf{w}\in\R^n$ for the vector $\x_\S$ padded with zeroes, i.e.\ $w_i = x_i$ for $i\in \S$ and $w_i =0$ otherwise, we can bound
\begin{align}
\label{eq:decomposition}
\big|\langle \nabla \widetilde{F}(\x) - \nabla f(\x), \x-\x^\star\rangle\big| &\le \big|\langle \nabla \widetilde{F}(\x), \x-\x^\star\rangle - \langle \nabla \widetilde{F}(\mathbf{w}), \mathbf{w}-\x^\star\rangle \big| \nonumber\\
&\quad + \big|\langle \nabla \widetilde{F}(\mathbf{w}), \mathbf{w}-\x^\star\rangle - \E[\langle \nabla \widetilde{F}(\mathbf{w}), \mathbf{w}-\x^\star\rangle] \big| \nonumber\\
&\quad +\big|\E[\langle \nabla \widetilde{F}(\mathbf{w}), \mathbf{w}-\x^\star\rangle] - \E[\langle \nabla \widetilde{F}(\x), \x-\x^\star\rangle] \big|,
\end{align} 
where we used that $\nabla f(\x) = \E[\nabla F(\x)] = \E[\nabla \widetilde{F}(\x)]$. To bound these three terms, we write
\begin{align*}
\langle \nabla \widetilde{F}(\x), \x-\x^\star\rangle  &= \frac{1}{m}\sum_{j=1}^m(\A_j^\top\x)^4 - \frac{1}{m}\sum_{j=1}^m (\A_j^\top\x)^3(\A_j^\top\x^\star) \nonumber\\
&\quad - \frac{1}{m}\sum_{j=1}^m(\A_j^\top\x)^2(\A_j^\top\x^\star)^2 + \frac{1}{m}\sum_{j=1}^m(\A_j^\top\x)(\A_j^\top\x^\star)^3 \nonumber\\
&= g_{1,\x}(\A) - g_{2,\x}(\A) - g_{3,\x}(\A) + g_{4,\x}(\A),
\end{align*}
and note that we have analogous expressions for $\langle\widetilde{F}(\mathbf{w}),\mathbf{w}-\x^\star\rangle$, $\E[\langle\widetilde{F}(\mathbf{w}),\mathbf{w}-\x^\star\rangle]$ and $\E[\langle\widetilde{F}(\x),\x-\x^\star\rangle]$.
We will show that, with probability $1-\frac{c_p}{3}n^{-10}$, the four terms $g_{i,\x}(\A)$ deviate by at most $\frac{\gamma}{4}(\frac{\|\x\|_1}{\sqrt{k}} + \delta )$ from their means by bounding each difference as in (\ref{eq:decomposition}). 

\textbf{Step 2, part (a): Split $g_{i,\x}(\A)$ into a term depending on $\x_\S$ and a residual term}\\
We first show that the first and last terms in (\ref{eq:decomposition}) are both bounded by $\frac{\gamma}{2}\delta$. To this end, we split each of the four terms $g_{i,\x}(\A)$ into a term that only depends on $\x_\S$ (which corresponds to $\langle \nabla \widetilde{F}(\mathbf{w}), \mathbf{w}-\x^\star\rangle$) and a residual term that also depends on $\x_{\S^c}$ (which corresponds to the difference $\langle \nabla \widetilde{F}(\mathbf{x}), \mathbf{x}-\x^\star\rangle - \langle \nabla \widetilde{F}(\mathbf{w}), \mathbf{w}-\x^\star\rangle$).
We only present the computation for $g_{1,\x}(\A)$, since the other three terms can be bounded following the same steps. We have
\begin{align*}
g_{1,\x}(\A) &= \frac{1}{m} \sum_{j=1}^m(\A_{j,\S}^\top\x_\S + \A_{j,\S^c}^\top\x_{\S^c})^4 \\*
&= \frac{1}{m}\sum_{j=1}^m(\A_{j,\S}^\top\x_\S)^4 + \frac{4}{m}\sum_{j=1}^m(\A_{j,\S}^\top\x_\S)^3(\A_{j,\S^c}^\top\x_{\S^c}) + \frac{6}{m}\sum_{j=1}^m(\A_{j,\S}^\top\x_\S)^2(\A_{j,\S^c}^\top\x_{\S^c})^2 \\
&\quad + \frac{4}{m}\sum_{j=1}^m(\A_{j,\S}^\top\x_\S)(\A_{j,\S^c}^\top\x_{\S^c})^3 + \frac{1}{m}\sum_{j=1}^m(\A_{j,\S^c}^\top\x_{\S^c})^4.
\end{align*} 
The first term only depends on $\x_\S = \mathbf{w}_\S$, and we denote it by $h_{1,\mathbf{w}}(\A) = \frac{1}{m}\sum_{j=1}^m(\A_{j,\S}^\top\mathbf{w}_\S)^4$.
The other terms can be bounded as follows:
\begin{align*}
\frac{4}{m}\sum_{j=1}^m(\A_{j,\S}^\top\x_\S)^3(\A_{j,\S^c}^\top\x_{\S^c}) &= 4\sum_{l\notin \S}x_l \frac{1}{m}\sum_{j=1}^mA_{jl}(\A_{j,\S}^\top\x_\S)^3 \\
&\le 4\|\x_{\S^c}\|_1 \max_{l\notin \S}\Bigg|\frac{1}{m}\sum_{j=1}^mA_{jl}(\A_{j,\S}^\top\x_\S)^3\Bigg| \\
&\le \frac{\gamma}{11}\delta,
\end{align*}
provided that $m \ge (\frac{44\tilde{c}}{\gamma})^2k\log^2 n$, where we used that $\|\x_\S\|_1\le \sqrt{k}$ and $\tilde{c}$ is the universal constant from Lemma \ref{lemma:tech7}. For the first inequality we used H\"{o}lder's inequality, and the second inequality holds by (\ref{eq:tech7_1}) of Lemma \ref{lemma:tech7} with probability $1-\frac{c_p}{9}n^{-10}$.
Similarly, we can bound
\begin{align*}
\frac{6}{m}\sum_{j=1}^m(\A_{j,\S}^\top\x_\S)^2(\A_{j,\S^c}^\top\x_{\S^c})^2 &= 6\sum_{l\notin \S}x_l\sum_{r\notin \S}x_r \frac{1}{m} \sum_{j=1}^mA_{jl}A_{jr}(\A_{j,\S}^\top\x_\S)^2\\
&\le 6\|\x_{\S^c}\|_1^2 \max_{l,r\notin \S}\sqrt{\frac{1}{m}\sum_{j=1}^mA_{jl}^2A_{jr}^2} \sqrt{\frac{1}{m}\sum_{j=1}^m(\A_{j,\S}^\top\x_\S)^4}\\
&\le \frac{c}{11}\delta^2,
\end{align*}
provided that $c>0$ is large enough, where we used H\"{o}lder's inequality in the second line and the last inequality holds by Lemmas \ref{lemma:tech1} and \ref{lemma:tech3} with probability $1 - \frac{2c_p}{9}n^{-10}$. The same computation yields
\begin{align*}
\frac{4}{m}\sum_{j=1}^m(\A_{j,\S}^\top\x_\S)(\A_{j,\S^c}^\top\x_{\S^c})^3 &\le \frac{c}{11} \delta^3,  \\
\frac{1}{m}\sum_{j=1}^m(\A_{j,\S^c}^\top\x_{\S^c})^4 &\le \frac{c}{11}\delta^4. 
\end{align*}
Putting everything together, we can bound the residual term by $\frac{c}{8}\delta$ since $\delta \le \frac{1}{4}\min\{\frac{\gamma}{c}, 1\}$.
Bounding $g_{2,\x}(\A)$, $g_{3,\x}(\A)$ and $g_{4,\x}(\A)$ the same way, we have, with probability $1-\frac{c_p}{3}n^{-10}$, 
\begin{equation*}
\big| \langle \nabla \widetilde{F}(\mathbf{x}), \mathbf{x}-\x^\star\rangle - \langle \nabla \widetilde{F}(\mathbf{w}), \mathbf{w}-\x^\star\rangle \big| \le \frac{\gamma}{2}\delta.
\end{equation*}
For the difference in expectation, we can bound, using $\|\x\|_2\le 1$ and $\|\x_{\S^c}\|_2\le \|\x_{\S^c}\|_1\le\delta$,
\begin{align*}
\E\Bigg[\frac{1}{m}\sum_{j=1}^m(\A_{j,\S}^\top\x_\S)^3(\A_{j,\S^c}^\top\x_{\S^c})\Bigg] &= \E\big[(\A_{1,\S}^\top\x_\S)^3\big]\E\big[\A_{j,\S^c}^\top\x_{\S^c}\big] = 0, \\
\E\Bigg[\frac{1}{m}\sum_{j=1}^m(\A_{j,\S}^\top\x_\S)^2(\A_{j,\S^c}^\top\x_{\S^c})^2\Bigg] &= \E\big[(\A_{1,\S}^\top\x_\S)^2\big]\E\big[(\A_{j,\S^c}^\top\x_{\S^c})^2\big] \le \delta^2, \\
\E\Bigg[\frac{1}{m}\sum_{j=1}^m(\A_{j,\S}^\top\x_\S)(\A_{j,\S^c}^\top\x_{\S^c})^3\Bigg] &= \E\big[\A_{1,\S}^\top\x_\S\big]\E\big[(\A_{j,\S^c}^\top\x_{\S^c})^3\big] = 0, \\
\E\Bigg[\frac{1}{m}\sum_{j=1}^m(\A_{j,\S^c}^\top\x_{\S^c})^4\Bigg] &= \E\big[(\A_{j,\S^c}^\top\x_{\S^c})^4\big] \le 3\delta^4.
\end{align*}
Repeating this for $g_{2,\x}(\A)$, $g_{3,\x}(\A)$ and $g_{4,\x}(\A)$ shows 
\begin{equation*}
\big| \E\big[\langle \nabla F(\mathbf{x}), \mathbf{x}-\x^*\rangle\big] - \E\big[\langle \nabla F(\mathbf{w}), \mathbf{w}-\x^*\rangle\big] \big| \le \frac{\gamma}{2}\delta.
\end{equation*}

\textbf{Step 2, part (b): Bound the term $|\langle \nabla F(\mathbf{w},\mathbf{w}-\x^\star\rangle - \E[\langle \nabla F(\mathbf{w}), \mathbf{w}-\x^\star\rangle] |$}\\
Finally, we need to bound the second term in (\ref{eq:decomposition}), i.e.\ we need to show that the terms $h_{i,\mathbf{w}}(\A)$ concentrate around their respective expectations.
Let $N_{\epsilon}$ be a smallest $\epsilon$-net of the set $\mathcal{X} = \{\mathbf{w}\in\R^n: \mathbf{w}_{\S^c} = \mathbf{0}, \|\mathbf{w}\|_2 = 1\}$ (which corresponds to the unit sphere in $\R^k$), where $\epsilon = c_1\gamma n^{-3}$ for some constant $c_1 > 0$. We will first show that we can bound the difference $|h_{i,\mathbf{w}}(\A)- \E[h_{i,\mathbf{w}}(\A)]|\le \frac{\gamma}{8}\frac{\|\mathbf{w}\|_1}{\sqrt{k}}$ for every $\mathbf{w}\in N_{\epsilon}$ via concentration of Lipschitz functions for Gaussian random variables. Then, we will extend this bound to every $\mathbf{w}\in \mathcal{X}$.

As the functions $h_{i,\mathbf{w}}$ are not globally Lipschitz continuous, we cannot directly apply Theorem \ref{thm:ref1}. We will first bound the Lipschitz constant of $h_{i,\mathbf{w}}$ restricted to the set $\mathcal{A}$ defined as the intersection of the sets defined in Lemma \ref{lemma:tech2} and Lemma \ref{lemma:tech4}, and then extend this restricted function to a function $\tilde{h}_{i,\mathbf{w}}$ on the entire space such that $\tilde{h}_{i,\mathbf{w}}$ is globally Lipschitz continuous. We can then apply Theorem \ref{thm:ref1} to $\tilde{h}_{i,\mathbf{w}}$, which also provides a high probability bound for $h_{i,\mathbf{w}}$, since by construction $h_{i,\mathbf{w}}(\A)=\tilde{h}_{i,\mathbf{w}}(\A)$ with high probability.

\textbf{Step 2, part (b.1): Bound the Lipschitz constant of $h_{i,\mathbf{w}}$ restricted to $\mathcal{A}$}\\
Let the subset $\mathcal{A}\subset\R^{m\times n}$ be defined as the intersection of the sets defined in Lemma \ref{lemma:tech2} with $t = 5\sqrt{\log n}$ and Lemma \ref{lemma:tech4}. More precisely, we require the projection of $\mathcal{A}$ onto $\R^{m\times\S}$ to be as in Lemma \ref{lemma:tech2}.
By the aforementioned lemmas, we have $\P[\mathcal{A}^c]\le c_2n^{-12}$ for a universal constant $c_2>0$. Since $\mathcal{A}$ is convex, it follows from the mean-value theorem that the Lipschitz constant of $h_{i,\mathbf{w}}$ restricted to $\mathcal{A}$ is bounded by the norm of its gradient $\|\nabla h_{i,\mathbf{w}}\|_2$. For any $\{\mathbf{a}_j\}_{j=1}^m\in\mathcal{A}$, we can compute the following.
\begin{itemize}
\item $\frac{\partial}{\partial a_{jl}} h_{1,\mathbf{w}}(\mathbf{a}) = \frac{4}{m}w_l(\mathbf{a}_j^\top\mathbf{w})^3$. We can bound the Lipschitz constant by the squareroot of 
\begin{align*}
\|\nabla h_{1,\mathbf{w}}(\mathbf{a})\|_2^2 &= \frac{16}{m^2}\sum_{l=1}^nw_l^2\sum_{j=1}^m(\mathbf{a}_j^\top\mathbf{w})^6 \\
&\le \frac{16}{m^2}\bigl((15m)^{\frac{1}{6}} + \sqrt{8\log (2k)}\|\mathbf{w}\|_1 + 5\sqrt{\log n}\bigr)^6 \nonumber\\
&\le c_3 \frac{\|\mathbf{w}\|_1^2\log k}{m},
\end{align*}
where we used $\|\mathbf{w}\|_2= 1$, $\|\mathbf{w}\|_1\le \sqrt{k}$, $m \ge c_s(\gamma) \max\{k^2 \log^2 n, \log^5 n\}$ and Lemma \ref{lemma:tech2}.

\item $\frac{\partial}{\partial a_{jl}} h_{2,\mathbf{w}}(\mathbf{a}) = \frac{1}{m}(x_l^\star(\mathbf{a}_j^\top\mathbf{w})^3 + 3w_l (\mathbf{a}_j^\top\mathbf{w})^2(\mathbf{a}_j^\top\mathbf{w}^\star))$. Hence, we can bound the Lipschitz constant by the squareroot of
\begin{align*}
\|\nabla h_{2,\mathbf{w}}(\mathbf{a})\|_2^2 &= \frac{2}{m^2}\sum_{l=1}^n\sum_{j=1}^m(x^\star_l)^2(\mathbf{a}_j^\top\mathbf{w})^6 + 9 w_l^2(\mathbf{a}_j^\top\mathbf{w})^4(\mathbf{a}_j^\top\x^\star)^2 \\
&\le \frac{2}{m^2}\bigl((15m)^{\frac{1}{6}} + 2\sqrt{2\log (2k)}\|\mathbf{w}\|_1 + 5\sqrt{\log n} \bigr)^6 \\
&\quad + \frac{18}{m} \sqrt{\frac{1}{m}\sum_{j=1}^m(\mathbf{a}_j^\top\mathbf{w})^8} \sqrt{\frac{1}{m}\sum_{j=1}^m(\mathbf{a}_j^\top\x^\star)^4}\\
&\le c_4 \frac{\|\mathbf{w}\|_1^2\log k}{m},
\end{align*}
where we use Lemma \ref{lemma:tech2} to bound the sum $\frac{1}{m}\sum_{j=1}^m(\mathbf{a}_j^\top\mathbf{w})^8$ and Lemma \ref{lemma:tech4} to bound $\frac{1}{m}\sum_{j=1}^m(\mathbf{a}_j^\top\x^*)^4$.

\item $\frac{\partial}{\partial a_{ji}} h_{3,\mathbf{w}}(\mathbf{a}) = \frac{2}{m}(x_l^\star(\mathbf{a}_j^\top\mathbf{w})^2(\mathbf{a}_j^\top\x^\star) + w_l (\mathbf{a}_j^\top\mathbf{w})(\mathbf{a}_j^\top\x^\star)^2)$. Hence, we can bound the Lipschitz constant by the squareroot of
\begin{align*}
\|\nabla h_{3,\mathbf{w}}(\mathbf{a})\|_2^2 &= \frac{4}{m^2}\sum_{l=1}^n\sum_{j=1}^m2(x^\star_l)^2(\mathbf{a}_j^\top\mathbf{w})^4(\mathbf{a}_j^\top\x^\star)^2 + 2 w_l^2(\mathbf{a}_j^\top\mathbf{w})^2(\mathbf{a}_j^\top\x^\star)^4 \\
&\le \frac{8}{m}\sqrt{\frac{1}{m}\sum_{j=1}^m(\mathbf{a}_j^\top\mathbf{w})^8} \sqrt{\frac{1}{m}\sum_{j=1}^m(\mathbf{a}_j^\top\x^\star)^4} \\
&\quad + \frac{8}{m} \sqrt{\frac{1}{m}\sum_{j=1}^m(\mathbf{a}_j^\top\mathbf{w})^4} \sqrt{\frac{1}{m}\sum_{j=1}^m(\mathbf{a}_j^\top\x^\star)^8}\\
&\le c_5 \frac{\|\mathbf{w}\|_1^2\log k}{m},
\end{align*}
again by Lemmas \ref{lemma:tech2} and \ref{lemma:tech4}.

\item $\frac{\partial}{\partial a_{ji}} h_{4,\mathbf{w}}(\mathbf{a}) = \frac{1}{m}(3x_l^\star(\mathbf{a}_j^\top\mathbf{w})(\mathbf{a}_j^\top\x^\star)^2 + w_l(\mathbf{a}_j^\top\x^\star)^3)$. Hence, we can bound the Lipschitz constant by the squareroot of
\begin{align*}
\|\nabla h_{4,\mathbf{w}}(\mathbf{a})\|_2^2 &= \frac{2}{m^2}\sum_{l=1}^n\sum_{j=1}^m9(x^\star_l)^2(\mathbf{a}_j^\top\mathbf{w})^2(\mathbf{a}_j^\top\x^\star)^4 + w_l^2(\mathbf{a}_j^\top\x^\star)^6 \\
&\le \frac{18}{m}\sqrt{\frac{1}{m}\sum_{j=1}^m(\mathbf{a}_j^\top\mathbf{w})^4} \sqrt{\frac{1}{m}\sum_{j=1}^m(\mathbf{a}_j^\top\x^\star)^8} + \frac{2}{m} \frac{1}{m}\sum_{j=1}^m(\mathbf{a}_j^\top\x^\star)^6\\
&\le c_6 \frac{\|\mathbf{w}\|_1^2\log k}{m},
\end{align*}
again by Lemmas \ref{lemma:tech2} and \ref{lemma:tech4}.
\end{itemize}

\textbf{Step 2, part (b.2): Construct a globally Lipschitz continuous extension of $h_{i,\mathbf{w}}$}\\
We only present the following steps for the first term $h_{1,\mathbf{w}}$, as the proofs for the other three terms follow the exact same steps.
Consider the following Lipschitz extension of $h_{1,\mathbf{w}}$:
\begin{equation*}
\tilde{h}_{1,\mathbf{w}}(\mathbf{a}) = \inf_{\mathbf{a'}\in \mathcal{A}} \bigl(h_{1,\mathbf{w}}(\mathbf{a}') + \operatorname{Lip}(h_{1,\mathbf{w}}) \|\mathbf{a} - \mathbf{a}'\|_2\bigr),
\end{equation*} 
where we write $\operatorname{Lip}(h_{1,\mathbf{w}}) =  \sqrt{c_3 \frac{\|\mathbf{w}\|_1^2\log k}{m}}$. By definition, we have $\tilde{h}=h$ on $\mathcal{A}$, and it follows from an application of the triangle inequality that $\tilde{h}$ is globally Lipschitz continuous with Lipschitz constant $\operatorname{Lip}(h)$ (see e.g.\ Theorem 7.2 of \cite{M95}). 

We will show that $\tilde{h}_{1,\mathbf{w}}$ concentrates around its mean, which can potentially differ from the mean of $h_{1,\mathbf{w}}$.
Since $h_{1,\mathbf{w}}$ and $\tilde{h}_{1,\mathbf{w}}$ differ only on $\mathcal{A}^c$ (which has probability less than $c_2n^{-12}$), we can bound, using the Cauchy-Schwarz inequality,
\begin{equation*}
\E\big[|h_{1,\mathbf{w}}(\A)|\Eins_{\mathcal{A}^c}(\A)\big] \le \sqrt{\E\big[h_{1,\mathbf{w}}(\A)^2\big]} \sqrt{\E\big[\Eins_{\mathcal{A}^c}(\A)\big]} \le \frac{c'}{n^6},
\end{equation*}
where we used $\E[h_{1,\mathbf{w}}(\A)^2] \le 3 + 105/m$. Similarly, since $\tilde{h}_{1,\mathbf{w}}(\mathbf{a}) \le \operatorname{Lip}(h_{1,\mathbf{w}}) \|\mathbf{a}\|_2$, 
\begin{equation*}
\E\big[|\tilde{h}_{1,\mathbf{w}}(\A)|\Eins_{\mathcal{A}^c}(\A)\big] \le \operatorname{Lip}(h_{1,\mathbf{w}})\sqrt{\E\big[\|\A\|_2^2\big]} \sqrt{\E\big[\Eins_{\mathcal{A}^c}(\A)\big]} \le \frac{c''}{n^5}.
\end{equation*}
All in all, this shows that
\begin{equation*}
\big|\E[h_{1,\mathbf{w}}(\A)] - \E[\tilde{h}_{1,\mathbf{w}}(\A)]\big| \le \frac{c_{7}}{n^5}
\end{equation*}
for a constant $c_7 > 0$. Finally, using the triangle inequality and Theorem \ref{thm:ref1}, we have
\begin{equation*}
\P\biggl[\big|\tilde{h}_{1,\mathbf{w}}(\A) - \E[h_{1,\mathbf{w}}(\A)]\big| > \frac{\gamma}{8}\frac{\|\mathbf{w}\|_1}{\sqrt{k}} \biggr] \le 2\exp\Biggl(- \frac{(\frac{\gamma}{8}\frac{\|\mathbf{w}\|_1}{\sqrt{k}} - \frac{c_{7}}{n^5})^2}{2c_3\frac{\|\mathbf{w}\|_1^2\log k}{m}}\Biggr) \le 2\exp(-c_{8} k \log n)
\end{equation*}
for a constant $c_{8} \le \frac{\gamma^2c_s(\gamma)}{128c_3} - \frac{\gamma c_s(\gamma)c_{7}}{8c_3n^5}$.

\textbf{Step 2, part (b.3): Union bound over $\mathbf{w}\in N_{\epsilon}$}\\
Taking the union bound over all $\mathbf{w}\in N_{\epsilon}$, which has cardinality bounded by $(3/\epsilon)^k$, we have
\begin{equation*}
\P\biggl[\big|\tilde{h}_{1,\mathbf{w}}(\A) - \E[h_{1,\mathbf{w}}(\A)]\big| > \frac{\gamma}{8}\frac{\|\mathbf{w}\|_1}{\sqrt{k}} \text{ for some } \mathbf{w}\in N_{\epsilon}\biggr] \le 2\exp\biggl(-c_{8} k \log n + k\log \frac{3}{\epsilon}\biggr).
\end{equation*}
Since $h_{1,\mathbf{w}}(\mathbf{a}) = \tilde{h}_{1,\mathbf{w}}(\mathbf{a})$ for all $\mathbf{w}$ provided that $\mathbf{a}\in\mathcal{A}$, this implies
\begin{equation*}
\P\biggl[\big|h_{1,\mathbf{w}}(\A) - \E[h_{1,\mathbf{w}}(\A)]\big| > \frac{\gamma}{8}\frac{\|\mathbf{w}\|_1}{\sqrt{k}} \text{ for some } \mathbf{w}\in N_{\epsilon}\biggr] \le 2\exp(-c_9 k \log n) + c_2n^{-12},
\end{equation*} 
for a constant $c_{9} \le c_{8}-3 - \frac{1}{\log n}\log\frac{3}{c_1\gamma}$, as we have $\epsilon = c_1\gamma n^{-3}$.

\textbf{Step 2, part (b.4): From $\epsilon$-net to the full sphere}\\
Next, we show that $|h_{1,\mathbf{w}}(\A) - \E[h_{1,\mathbf{w}}(\A)]|\le \frac{\gamma}{4}\frac{\|\mathbf{w}\|_1}{\sqrt{k}}$ for any $\mathbf{w}\in \mathcal{X}$. The case $\|\mathbf{w}\|_2<1$ follows by considering the vector $\mathbf{w}/\|\mathbf{w}\|_2$ and rescaling.
For any $\mathbf{w}\in \mathcal{X}$, let $\mathbf{w}'\in N_{\epsilon}$ with $\|\mathbf{w}-\mathbf{w}'\|_2 \le \epsilon$. Then,
\begin{align*}
\big|h_{1,\mathbf{w}}(\A) - \E[h_{1,\mathbf{w}}(\A)]\big| &\le \big|h_{1,\mathbf{w}}(\A) - h_{1,\mathbf{w}'}(\A)\big| +\big|h_{1,\mathbf{w}'}(\A) - \E[h_{1,\mathbf{w}'}(\A)]\big| \\
&\quad +\big|\E[h_{1,\mathbf{w}'}(\A)] - \E[h_{1,\mathbf{w}}(\A)]\big|.
\end{align*}
The first and third term can be bounded using the indentity $a^4-b^4 = (a^2 + b^2)(a+b)(a-b)$:
\begin{align*}
\big|h_{1,\mathbf{w}}(\A) - h_{1,\mathbf{w}'}(\A)\big| &= \Bigg|\frac{1}{m}\sum_{j=1}^m \bigl((\A_j^\top\mathbf{w})^2 + (\A_j^\top\mathbf{w}')^2\bigr) \bigl(\A_j^\top(\mathbf{w}+\mathbf{w}')\bigr) \bigl(\A_j^\top(\mathbf{w}-\mathbf{w}')\bigr) \Bigg| \\
&\le \max_j 2\|\A_j\|_2^2 \cdot 2\|\A_j\|_2 \|\A_j\|_2 \|\mathbf{w}-\mathbf{w}'\|_2 \\
&\le 4 \big(\sqrt{k} + 5\sqrt{\log n}\big)^4 \|\mathbf{w}-\mathbf{w}'\|_2 \\
&\le \frac{\gamma}{16} \frac{\|\mathbf{w}\|_1}{\sqrt{k}},
\end{align*}
with probability $1-mn^{-12.5}$, provided $c_1$ is sufficiently large, where we used the fact that the norm $\|\cdot \|_2$ is 1-Lipschitz continuous and applied Theorem \ref{thm:ref1} to bound the term $\|\A_j\|_2$, and used $\|\mathbf{w}-\mathbf{w}'\|_2 \le \epsilon$ for the last inequality.
For the expectation, the same argument yields
\begin{align*}
\big|\E[h_{1,\mathbf{w}}(\A)] - \E[h_{1,\mathbf{w}'}(\A)]\big| &\le \E\bigl[2\|\A_j\|_2^2 \cdot 2\|\A_j\|_2 \|\A_j\|_2 \|\mathbf{w}-\mathbf{w}'\|_2\bigr] \\
&= 4 \big(3k + k(k-1)\big) \|\mathbf{w}-\mathbf{w}'\|_2 \\
&\le \frac{\gamma}{16} \frac{\|\mathbf{w}\|_1}{\sqrt{k}}.
\end{align*}
This completes the proof that 
\begin{equation*}
\P\biggl[\big|h_{1,\mathbf{w}}(\A) - \E[h_{1,\mathbf{w}}(\A)]\big| < \frac{\gamma}{4}\frac{\|\mathbf{w}\|_1}{\sqrt{k}} \text{ for all } \mathbf{w}\in \mathcal{X} \biggr] \ge 1 - \frac{c_p}{24}n^{-10},
\end{equation*}
provided that $c_p \ge 24 + 24c_2n^{-2}$, and if $m\le n^{2.5}$. The other case $m>n^{2.5}$ is simpler and can be shown following the same steps, writing the probabilities in terms of $m$ instead of $n$. We omit the details to avoid repetition.
Repeating the same steps for the terms $h_{2,\mathbf{w}}$, $h_{3,\mathbf{w}}$ and $h_{4,\mathbf{w}}$ shows that
\begin{align*}
&\P\biggl[\big|\langle \nabla \widetilde{F}(\mathbf{w}), \mathbf{w}-\x^\star\rangle - \E\big[\langle \nabla \widetilde{F}(\mathbf{w}), \mathbf{w}-\x^\star\rangle\big]\big|\le \gamma\frac{\|\mathbf{w}\|_1}{\sqrt{k}} \text{ for all } \mathbf{w}\in\mathcal{X} \biggr] \ge& 1 - \frac{c_p}{6}n^{-10}.
\end{align*}
Finally, the bound also holds for any vector $\mathbf{w}\in\R^n$ with $\mathbf{w}_{\S^c} = \mathbf{0}$ and $\|\mathbf{w}\|_2<1$ by considering $\mathbf{w}/\|\mathbf{w}\|_2$, and noting that each of the four terms which make up $\langle \nabla \widetilde{F}(\mathbf{w}), \mathbf{w}-\x^\star\rangle$ scale at least linearly in $\|\mathbf{w}\|_2$.

\textbf{Step 3: Bound the term $\big|\langle \nabla \widetilde{F}(\x) - \nabla f(\x), \x-\x^\star\rangle\big|$ by $\gamma(\|\x_\S-\x_\S^\star\|_2^2 +\delta^2)$}\\
Substituting $\mathbf{z} = \x - \x^\star$, we have
\begin{equation}\label{eq:decomp_bound2}
\langle \nabla \widetilde{F}(\x), \x - \x^\star\rangle = \frac{1}{m}\sum_{j=1}^m(\A_j^\top\mathbf{z})^4 + \frac{3}{m}\sum_{j=1}^m(\A_j^\top\mathbf{z})^3(\A_j^\top\x^\star) + \frac{2}{m}\sum_{j=1}^m (\A_j^\top\mathbf{z})^2(\A_j^\top\x^\star)^2.
\end{equation}
The proof follows the same steps as the proof of the bound $\gamma (\frac{\|\x_\S\|_1}{\sqrt{k}} + \delta)$, writing $\mathbf{z}_\S = \|\mathbf{z}_\S\|_2\frac{\mathbf{z}_\S}{\|\mathbf{z}_\S\|_2}$ and using the fact that $\|\mathbf{z}_{\S^c}\|_1 = \|\x_{\S^c}\|_1 \le \delta$ and $\frac{\|\mathbf{z}_\S\|_1}{\sqrt{k}}\le \|\mathbf{z}_\S\|_2$. We then obtain the desired bound since each of the three terms depends at least quadratically on $\mathbf{z}$.

We illustrate this argument for the last term in (\ref{eq:decomp_bound2}). The other two terms can be controlled similarly (they are easier to control because of the higher order dependence on $\mathbf{z}$). We have
\begin{align*}
\frac{1}{m}\sum_{j=1}^m (\A_j^\top\mathbf{z})^2(\A_j^\top\x^\star)^2 &= \frac{1}{m}\sum_{j=1}^m (\A_{j,\S}^\top\mathbf{z}_\S)^2(\A_j^\top\x^\star)^2 + \frac{2}{m}\sum_{j=1}^m (\A_{j,\S}^\top\mathbf{z}_\S)(\A_{j,\S^c}^\top\mathbf{z}_{\S^c})(\A_j^\top\x^\star)^2 \\
&\quad + \frac{1}{m}\sum_{j=1}^m (\A_{j,\S^c}^\top\mathbf{z}_{\S^c})^2(\A_j^\top\x^\star)^2\\
&=: B_1 + B_2 + B_3.
\end{align*}
The same computation as in Step 2, part (b) shows that, with probability $1-\frac{c_p}{27}n^{-10}$,
\begin{equation*}
\big|B_1 - \E[B_1]\big| \le \frac{\gamma}{12} \|\mathbf{z}_\S\|_2^2.
\end{equation*}
As in Step 2, part (a), we can bound with probability $1-\frac{c_p}{27}n^{-10}$,
\begin{align*}
\big|B_2 - \E[B_2]\big| &= 2\Bigg| \sum_{i\notin \S}z_i \sum_{l\in \S}z_l \frac{1}{m}\sum_{j=1}^mA_{ji}A_{jl}(\A_j^\top\x^\star)^2 \Bigg| \\*
&\le 2\|\mathbf{z}_{\S^c}\|_1 \|\mathbf{z}_\S\|_1 \max_{i\notin \S, l\in \S}\Bigg|\frac{1}{m}\sum_{j=1}^mA_{ji}A_{jl}(\A_j^\top\x^\star)^2 \Bigg|  \\
&\le \frac{\gamma}{6} \|\mathbf{z}_\S\|_2 \delta,
\end{align*}
where we used H\"{o}lder's inequality, (\ref{eq:tech3_3}) of Lemma \ref{lemma:tech3} together with $m\ge c_s(\gamma)k^2\log^2 n$, and $\|\mathbf{z}_\S\|_1\le \sqrt{k}\|\mathbf{z}_\S\|_2$. The same argument gives
\begin{equation*}
|B_3-\E[B_3]| \le \frac{\gamma}{12}\delta^2
\end{equation*}
with probability $1-\frac{c_p}{27}n^{-10}$. Combining these bounds, we have, with probability $1-\frac{c_p}{9}n^{-10}$,
\begin{equation*}
\Bigg|\frac{2}{m}\sum_{j=1}^m (\A_j^\top\mathbf{z})^2(\A_j^\top\x^\star)^2 - \E\Biggl[\frac{2}{m}\sum_{j=1}^m (\A_j^\top\mathbf{z})^2(\A_j^\top\x^\star)^2\Biggr]\Bigg| \le \frac{\gamma}{3} \big(\|\mathbf{z}_\S\|_2^2 + \delta^2\big),
\end{equation*}
where we used the inequality $2ab\le a^2 + b^2$.
Repeating these steps for the other two terms in (\ref{eq:decomp_bound2}) completes the proof that, with probability $1-\frac{c_p}{3}n^{-10}$, we have
\begin{equation*}
\big|\langle \nabla \widetilde{F}(\x), \x - \x^\star\rangle - \E[\langle \nabla \widetilde{F}(\x), \x - \x^\star\rangle]\big| \le \gamma\big(\|\x_\S-\x_\S^\star\|_2^2 + \delta^2\big) \quad\text{ for all } \x\in \mathcal{X}.
\end{equation*}
Finally, the simplified bound can be obtained directly by plugging in the additional assumptions.
\end{proof}

As the proof of Lemma \ref{lemma:support3} involves additional technical challenges in discrete time, we present the proof of Lemma \ref{lemma:support3} separately for the continuous-time and discrete-time cases. 

\begin{proof}[Proof of Lemma \ref{lemma:support3} in the continuous-time case] 
We will show that the three inequalities (\ref{eq:claim1}), (\ref{eq:claim2}) and (\ref{eq:claim3}) are satisfied by showing that, as long as all inequalities are satisfied, neither can be violated first. 

Let $I_0$ be the index for which we have the non-zero initialization $X_{I_0}(0) >0$. As we can only recover the signal $\x^\star$ up to a global sign from phaseless measurements, we can assume without loss of generality that $x^\star_{I_0} >0$, since we can otherwise replace $\x^\star$ by $-\x^\star$ in the proof below. That is, we need to show (\ref{eq:claim1}) with $\xi = +1$.

At $t=0$, (\ref{eq:claim1}) is satisfied by the definition of the initialization. Using standard concentration bounds for sub-exponential random variables (see e.g.\ Prop. 5.16 of \cite{V12}), we can bound with probability $1-4n^{-10}$,
\begin{equation*}
\frac{1}{m}\sum_{j=1}^mY_j = \frac{1}{m}\sum_{j=1}^m(\A_j^\top\x^\star)^2 + \frac{1}{m}\sum_{j=1}^m\varepsilon_j > 1-(9 + 25\sigma)\sqrt{\frac{\log n}{m}}.
\end{equation*}
Hence, the initialization (\ref{eq:initialization}) satisfies $\|\X(0)\|_2^2\ge \frac{1}{3}-(3+9\sigma)\sqrt{\frac{\log n}{m}}$. Similarly, we can show $\|\X(0)\|_2^2\le 2$ and therefore (\ref{eq:claim2}). Finally, inequality (\ref{eq:claim3}) is satisfied at $t=0$ since we have 
\begin{equation*}
\X(0)^\top\x^\star\ge \|\X_\S(0)\|_1x^\star_{min}\ge \frac{c_\star}{2\sqrt{k}} \ge (9 + 25\sigma )\sqrt{\frac{\log n}{m}} \ge 3\|\X(0)\|_2^2 - 1.
\end{equation*}

\textbf{Step 1: (\ref{eq:claim1}) continues to hold as long as (\ref{eq:claim2}) holds}\\
We prove this inequality by contradiction. Define $T_1 = \inf\{t\ge 0: X_i(t)x^\star_i <0 \text{ for some } i\}$ as the first time inequality (\ref{eq:claim1}) is violated. Assume that $T_1<T$ and that (\ref{eq:claim2}) holds for all $t\le T_1$. Let $i$ be the index for which $X_i(t)x^\star_i<0$ first occurs. This is only possible for a coordinate $i\in \S$, and by continuity we must have $X_i(T_1) = 0$.
Without loss of generality, assume that $x^\star_i>0$. We will show that
\begin{equation*}
\frac{d}{dt}X_i(T_1) = -\sqrt{X_i(T_1)^2 + \beta^2} \; \nabla F(\X(T_1))_i > 0,
\end{equation*}
that is $X_i(t)$ must become positive for $t$ close enough to $T_1$, which is a contradiction to the definition of $T_1$, and we hence must have $T_1 \ge T$.
We can bound, since both (\ref{eq:claim1}) and (\ref{eq:claim2}) hold at $T_1$, 
\begin{equation}\label{eq:innerproduct}
\X(T_1)^\top\x^\star \ge \|\X_\S(T_1)\|_1 x^\star_{min} \ge \|\X_\S(t_1)\|_1\frac{c_\star}{\sqrt{k}},
\end{equation}
and hence
\begin{equation*}
\nabla f(\X(T_1))_i = -2\big(\X(T_1)^\top\x^\star\big)x^\star_i \le - 2\|\X_\S(T_1)\|_1\frac{c_\star^2}{k},
\end{equation*}
where we used $x^\star_i \ge x^\star_{min} \ge \frac{c_\star}{\sqrt{k}}$.
As we assume $T_1 < T$, we can use the simplified bound in Lemma \ref{lemma:support1} to bound
\begin{equation*}
|\nabla F(\X(T_1))_i - \nabla f(\X(T_1))_i| \le 0.1\frac{c_\star^2\|\X_\S(T_1)\|_1}{k} \le \frac{1}{20}|\nabla f(\X(T_1))_i|
\end{equation*}
with probability $1-c_pn^{-10}$ if $c_s$ is sufficiently large, since we have, by assumption, $\|\X_{\S^c}(T_1)\|_1\le \delta\le c_1/n$.
Hence, we have $\nabla F(\X(T_1))_i < 0$, which implies $\frac{d}{dt}X_i(T_1)>0$ and contradicts the definition of $T_1$. Therefore, we must have $T_1 \ge T$.

\textbf{Step 2: (\ref{eq:claim2}) continues to hold as long as (\ref{eq:claim1}) and (\ref{eq:claim3}) hold}\\
We begin by showing the lower bound in (\ref{eq:claim2}). 

Define $T_2 = \inf \{t\ge 0: \|\X(t)\|_2^2<\frac{1}{3} - (3+9\sigma)\sqrt{\frac{\log n}{m}}\}$, and assume that $T_2<T$ and that inequalities (\ref{eq:claim1}) and (\ref{eq:claim3}) are satisfied for all $t\le T_2$. By continuity, we must have $\|\X(T_2)\|_2^2 = \frac{1}{3} - (3+9\sigma)\sqrt{\frac{\log n}{m}}$, and we will show that
\begin{align*}
\frac{d}{dt}\|\X(T_2)\|_2^2 &= -2\sum_{i=1}^nX_i(T_2)\frac{d}{dt}X_i(T_2) = -2\sum_{i=1}^n X_i(T_2) \sqrt{X_i(T_2)^2+\beta^2}\; \nabla F(\X(T_2))_i
\end{align*}
is positive, which implies $\|\X(t)\|_2^2>\frac{1}{3} - (3+9\sigma)\sqrt{\frac{\log n}{m}}$ for $t$ close enough to $T_2$ and contradicts the definition of $T_2$. Hence, we must have $T_2\ge T$.
Since $3\|\X(T_2)\|_2^2 - 1 < 0$, we can bound the population gradient using (\ref{eq:innerproduct}) and the assumption $\|\X_{\S^c}(t)\|_1\le \delta$:
\begin{align}\label{eq:popgradient}
\nabla f(\X(T_2))_i \begin{cases} \le -2\|\X_\S(T_2)\|_1\frac{c_\star^2}{k} \quad & x^\star_i>0 \\ \in (-\delta, \delta) & x^\star_i=0  \\ \ge 2\|\X_\S(T_2)\|_1\frac{c_\star^2}{k} & x^\star_i<0 \end{cases}
\end{align}
Recall that in the previous step we have shown, with probability $1-c_pn^{-10}$ and for all $i\in [n]$,
\begin{equation}\label{eq:difgradient}
|\nabla F(\X(T_1))_i - \nabla f(\X(T_1))_i| \le 0.1\frac{c_\star^2\|\X_\S(T_1)\|_1}{k},
\end{equation}
which results in bounds analogous to (\ref{eq:popgradient}) for the gradient $\nabla F(\X(T_2))_i$. In order to bound $\frac{d}{dt}\|\X(T_2)\|_2^2$, we write
\begin{align*}
\Bigg|\sum_{i\notin \S} X_i(T_2) \sqrt{X_i(T_2)^2+\beta^2}\; \nabla F(\X(T_2))_i\Bigg| \le \|\X_{\S^c}(T_2)\|_1 \le \delta,
\end{align*}
where we used $|\sqrt{X_i(T_2)^2+\beta^2}\; \nabla F(\X(T_2))_i|\le 1$. Using $\|\X_\S(T_2)\|_1 \ge \frac{1}{2}$, we have
\begin{align*}
-\sum_{i\in \S} X_i(T_2) \sqrt{X_i(T_2)^2+\beta^2}\; \nabla F(\X(T_2))_i \ge\sum_{i\in \S} X_i(T_2)^2  \frac{c_\star^2}{2k} \ge 0.15 \frac{c_\star^2}{k},
\end{align*}
where we used $\|\X_\S(T_2)\|_2^2 = \|\X(T_2)\|_2^2 - \|\X_{\S^c}(T_2)\|_2^2 \ge 0.3$. This shows that we must have $\frac{d}{dt}\|\X(T_2)\|_2^2 >0$, which contradicts the definition of $T_2$. Therefore, we must have $T_2 \ge T$.

The upper bound in (\ref{eq:claim2}) is an immediate consequence of (\ref{eq:claim3}): by the Cauchy-Schwarz inequality we have $\X(t)^\top\x^\star\le \|\X(t)\|_2$, so, for $3\|\X(t)\|_2^2 - 1 > 0$, we can bound
\begin{equation*}
\frac{2\|\X(t)\|_2}{3\|\X(t)\|_2^2-1} \ge \frac{2(\X(t)^\top\x^\star)}{3\|\X(t)\|_2^2-1} \ge \frac{1}{\sqrt{3}} \quad \Rightarrow \quad \|\X(t)\|_2 \le \frac{1+\sqrt{2}}{\sqrt{3}} < \sqrt{2}
\end{equation*}
by solving the quadratic form.

\textbf{Step 3: (\ref{eq:claim3}) continues to hold as long as (\ref{eq:claim1}) and (\ref{eq:claim2}) hold}\\
The proof of (\ref{eq:claim3}) follows the same recipe as the two previous steps, although the calculations are more complicated. When $3\|\X(t)\|_2^2-1 \le 0$, there is nothing to show. Otherwise, we can consider the ratio 
\begin{equation*}
R(t) = \frac{2(\X(t)^\top\x^\star)}{3\|\X(t)\|_2^2 - 1}
\end{equation*}
and show that it is bounded from below by $\frac{1}{\sqrt{3}}$ for all $t\le T$. 

Let $T_3 = \inf\{t\ge 0: R(t) < \frac{1}{\sqrt{3}}\}$ and assume $T_3<T$ as before. For notational simplicity, we will omit the argument $T_3$ in $X_i(T_3)$ in what follows. We can compute
\begin{align}\label{eq:diffratio}
\frac{d}{dt}R(T_3) = \sum_{i=1}^n-\sqrt{X_i^2 + \beta^2}\; \nabla F(\X)_i  \frac{2x^\star_i(3\|\X\|_2^2-1) - 2(\X^\top\x^\star)\cdot 6X_i}{(3\|\X\|_2^2-1)^2}.
\end{align}
By continuity, we have $R(T_3) = \frac{1}{\sqrt{3}}$. Then, for $X_i,x^\star_i>0$, we have
\begin{align}\label{eq:ratiosign}
2x^\star_i(3\|\X\|_2^2-1) - 2(\X^\top\x^\star)\cdot 6X_i &> 0 \nonumber\\
\Leftrightarrow \hspace{53mm}  X_i &< \frac{1}{\sqrt{3}}x^\star_i, 
\end{align}
and the analogous result for $X_i,x^*_i<0$. The idea to showing $\frac{d}{dt}R(t_3) >0$ is to show that coordinates with small magnitude $|X_i| < \frac{1}{\sqrt{3}}|x^\star_i|$ are increasing in magnitude, and conversely coordinates with large magnitude $|X_i|>\frac{1}{\sqrt{3}}|x^\star_i|$ are decreasing in magnitude. To this end, we split the coordinates $i\in [n]$ into five subsets:
\begin{align*}
\S^c &= \{i\in [n]: x^\star_i = 0\}, \\
\S_1 &= \biggl\{i\in \S: |X_i| < \biggl(\frac{1}{\sqrt{3}} - 0.1\biggr) |x^\star_i|\biggr\}, \\
\S_2 &= \biggl\{i\in \S: \biggl(\frac{1}{\sqrt{3}} - 0.1\biggr) |x^\star_i| \le |X_i| < \biggl(\frac{1}{\sqrt{3}} + 0.1\biggr) |x^\star_i|\biggr\}, \\
\S_3 &= \biggl\{i\in \S: \biggl(\frac{1}{\sqrt{3}} + 0.1\biggr) |x^\star_i| \le |X_i| < \frac{2}{\sqrt{3}} |x^\star_i|\biggr\}, \\
\S_4 &= \biggl\{i\in \S: |X_i| \ge \frac{2}{\sqrt{3}} |x^\star_i|\biggr\}.
\end{align*}
We bound the sum (\ref{eq:diffratio}) on each of these five sets.
\begin{itemize}
\item For $\S^c$, we have
\begin{equation*}
\Bigg|\sum_{i\in \S^c} \sqrt{X_i^2 + \beta^2}\; \nabla F(\X)_i \cdot 2\big(\X^\top\x^\star\big)\cdot 6X_i\Bigg| \le \frac{2c_\star^2\delta^2}{k}\|\X_\S\|_1 \big(\X^\top\x^\star\big),
\end{equation*}
where we used the fact that $\sqrt{X_i^2 + \beta^2}\le \sqrt{2}\delta$, $\|\X_{\S^c}\|_1^2\le \delta^2$ and that we can bound $|\nabla F(\X)_i|\le 0.1\|\X_\S\|_1\frac{c_\star^2}{k}$ using Lemma \ref{lemma:support1} as in (\ref{eq:difgradient}).

\item For an $i\in \S_1$ with $x^\star_i>0$, we have
\begin{equation*}
\nabla f(\X)_i = \big(3\|\X\|_2^2 - 1\big)X_i - 2\big(\X^\top\x^\star\big)x^\star_i \le -0.2\cdot \sqrt{3} \|\X_\S\|_1\frac{c_\star^2}{k},
\end{equation*}
where we used that $3\|\X\|_2^2-1 = \sqrt{3}\cdot 2(\X^\top\x^\star)$, $X_i < (\frac{1}{\sqrt{3}} - 0.1)x^\star_i$ and (\ref{eq:innerproduct}). Together with the bound (\ref{eq:difgradient}), this shows that $\nabla F(\X)_i <0$.
 Similarly, we can show $\nabla F(\X)_i >0$ for $i\in \S_1$ with $x^\star_i<0$. Hence, recalling (\ref{eq:ratiosign}),
\begin{equation*}
\sum_{i\in \S_1} -\sqrt{X_i^2 + \beta^2}\; \nabla F(\X)_i  \big(2x^\star_i\big(3\|\X\|_2^2-1\big) - 2\big(\X^\top\x^\star\big)\cdot 6X_i\big) \ge 0,
\end{equation*}
as each summand is non-neagtive.

\item For $\S_3$, we can use the same argument as for $\S_1$ to show that
\begin{equation*}
\sum_{i\in \S_3} -\sqrt{X_i^2 + \beta^2}\; \nabla F(\X)_i \big(2x^\star_i\big(3\|\X\|_2^2-1\big) - 2\big(\X^\top\x^\star\big)\cdot 6X_i\big) \ge 0.
\end{equation*}

\item For $\S_2$, we need to show that the sum  
\begin{align*}
\sum_{i\in \S_2}-\sqrt{X_i^2 + \beta^2}\; \nabla F(\X)_i \big(2x^\star_i\big(3\|\X\|_2^2-1\big) - 2\big(\X^\top\x^\star\big)\cdot 6X_i\big)
\end{align*}
is bounded from below. Let $i\in \S_2$ with $x^\star_i>0$. If $X_i < \frac{1}{\sqrt{3}}x^\star_i$ and $\nabla F(\X)_i<0$, or if $X_i > \frac{1}{\sqrt{3}}x^\star_i$ and $\nabla F(\X)_i>0$, then the summand is non-negative, i.e.\ bounded from below by zero. For $i\in \S_2$ with $X_i < \frac{1}{\sqrt{3}}x^\star_i$ and $\nabla F(\X)_i>0$, we can bound
\begin{align*}
\sqrt{X_i^2 + \beta^2} &\le (1+\beta) |X_i|, \\ 
\nabla f(\X)_i &<0 ,\\
\nabla F(\X)_i &\le \nabla f(\X)_i + \big|\nabla f(\X)_i - \nabla F(\X)_i\big|\le 0.1 \|\X_\S\|_1\frac{c_\star^2}{k}, 
\end{align*}
where we used the bound (\ref{eq:difgradient}). Recalling the definition of $\S_2$, we have
\begin{align*}
2x^\star_i\big(3\|\X\|_2^2-1\big) - 2\big(\X^\top\x^\star\big)\cdot 6X_i &= 12\big(\X^\top\x^\star\big) \biggl(\frac{1}{\sqrt{3}}x^\star_i - X_i\biggr) \le 3\big(\X^\top\x^\star\big)X_i.
\end{align*}
Putting this together, we can bound
\begin{align*}
&-\sqrt{X_i^2 + \beta^2}\; \nabla F(\x)_i  \big(2x^\star_i\big(3\|\X\|_2^2-1\big) - 2\big(\X^\top\x^\star\big)\cdot 6X_i\big) \\
 \ge &-\frac{0.3(1+\beta)c_\star^2}{k} \|\X_\S\|_1 \big(\X^\top\x^\star\big) X_i^2.
\end{align*}
Together with the analogous bounds for the cases $X_i > \frac{1}{\sqrt{3}}x^\star_i$ and $\nabla F(\X)_i<0$, and for $i\in\S_2$ with $x^\star_i<0$, this yields
\begin{align*}
&\sum_{i\in \S_2}-\sqrt{X_i^2 + \beta^2}\; \nabla F(\X)_i \big(2x^\star_i\big(3\|\X\|_2^2-1\big) - 2\big(\X^\top\x^\star\big)\cdot 6X_i\big) \nonumber\\
\ge &-\frac{0.3(1+\beta)c_\star^2}{k} \|\X_\S\|_1 \big(\X^\top\x^\star\big) \|\X_{\S_2}\|_2^2.
\end{align*}

\item Finally, for $i\in\S_4$ we have
\begin{align*}
\sqrt{X_i^2 + \beta^2} &\ge |X_i|, \\ 
|\nabla f(\X)_i| &= 2\big(\X^\top\x^\star\big) \big|\sqrt{3}X_i-x^\star_i\big| \ge 2\big(\X^\top\x^\star\big)|x_i^\star| \ge 2\|\X_\S\|_1\frac{c_\star^2}{k},\\
|\nabla F(\X)_i| &\ge |\nabla f(\X)_i| - |\nabla f(\X)_i - \nabla F(\X)_i|\ge 1.9 \|\X_\S\|_1\frac{c_\star^2}{k}, 
\end{align*}
where we used the bound (\ref{eq:difgradient}). Recalling the definition of $\S_4$, we can bound
\begin{align*}
\big|2x^\star_i\big(3\|\X\|_2^2-1\big) - 2\big(\X^\top\x^\star\big)\cdot 6X_i\big| = 12\big(\X^\top\x^\star\big) \biggl|\frac{1}{\sqrt{3}}x^\star_i - X_i\biggr| \ge 6\big(\X^\top\x^\star\big)X_i.
\end{align*}
Putting everything together, we can bound
\begin{align*}
&\sum_{i\in \S_2}-\sqrt{X_i^2 + \beta^2}\; \nabla F(\X)_i \big(2x^\star_i\big(3\|\X\|_2^2-1\big) - 2\big(\X^\top\x^\star\big)\cdot 6X_i\big) \nonumber\\
\ge &\frac{11.4c_\star^2}{k} \|\X_\S\|_1 \big(\X^\top\x^\star\big) \|\X_{\S_4}\|_2^2.
\end{align*}
\end{itemize}
Putting these five sums together, we have shown that $\frac{d}{dt}R(T_3)>0$ if we can show that
\begin{equation*}
\frac{11.4c_\star^2}{k} \|\X_\S\|_1 \big(\X^\top\x^\star\big) \|\X_{\S_4}\|_2^2\ge\biggl(\frac{0.3(1+\beta)c_\star^2}{k} \|\X_{\S_2}\|_2^2 + \frac{2c_\star^2\delta^2}{k}\biggr) \|\X_\S\|_1 \big(\X^\top\x^\star\big) 
\end{equation*}
Since $\delta \le c_1/n$ is sufficiently small compared to $\|\X_{\S_4}\|_2$, this reduces to showing
\begin{equation}\label{eq:s2s4}
11\|\X_{\S_4}\|_2^2 \ge 0.3(1+\beta)\|\X_{\S_2}\|_2^2.
\end{equation}
Now, we can rearrange the equality $R(t) = \frac{1}{\sqrt{3}}$ to obtain
\begin{align*}
&3\Bigl(\|\X_{\S_1}\|_2^2 + \|\X_{\S_2}\|_2^2 + \|\X_{\S_3}\|_2^2 + \|\X_{\S_4}\|_2^2 + \|\X_{\S^c}\|_2^2\Bigr) - 1 \\
=& \sqrt{3}\cdot 2\Bigl(\X_{\S_1}^\top\x^\star_{\S_1} + \X_{\S_2}^\top\x^\star_{\S_2} + \X_{\S_3}^\top\x^\star_{\S_3} + \X_{\S_4}^\top\x^\star_{\S_4}\Bigr).
\end{align*} 
By definition, we have $3X_i^2 < \sqrt{3}\cdot 2X_ix_i^\star$ for $i\in S_1\cup S_2\cup S_3 \cup S_4$, which gives
\begin{equation*}
3\|\X_{\S_j}\|_2^2 < \sqrt{3}\cdot 2\X_{\S_j}^\top\x^\star_{\S_j}, \quad \text{for } j=1,...,4.
\end{equation*} . Further, we have $\|\X_{\S^c}\|_2^2 \le \delta$, so
\begin{equation*}
3\|\X_{\S_4}\|_2^2 - 2\sqrt{3}\cdot \X_{\S_4}^\top\x^\star_{\S_4} > 1-3\delta + \Bigl(2\sqrt{3}\cdot \X_{\S_2}^\top\x^\star_{\S_2} - 3\|\X_{\S_2}\|_2^2\Bigr). 
\end{equation*}
By the definition of $\S_2$, we have for $i\in\S_2$,
\begin{equation*}
2\sqrt{3} X_ix^\star_i - 3X_i^2 \ge \biggl(\frac{2\sqrt{3}}{1/\sqrt{3} + 0.1} - 3\biggr)X_i^2 \ge 2.2X_i^2.
\end{equation*}
Since $\X_{\S_4}^\top\x_{\S_4}^\star \ge 0$ by (\ref{eq:claim1}), this gives
\begin{equation*}
3\|\X_{\S_4}\|_2^2 > 1- 3\delta + 2.2 \|\X_{\S_2}\|_2^2,
\end{equation*}
which shows that (\ref{eq:s2s4}) holds, thus completing the proof of (\ref{eq:claim3}).
\end{proof}

\begin{proof}[Proof of Lemma \ref{lemma:support3} in the discrete-time case]
We will prove the inequalities (\ref{eq:claim1}), (\ref{eq:claim2}) and (\ref{eq:claim3}) via induction. 
Let $I_0$ be the index for which we have the non-zero initialization $X_{I_0}^0 >0$. As in the continuous-time case, we can assume without loss of generality that $x^\star_{I_0} >0$ and show (\ref{eq:claim1}) with $\xi = +1$. The inequalities (\ref{eq:claim1}), (\ref{eq:claim2}) and (\ref{eq:claim3}) hold at $t=0$ as in the continuous-time case.

\textbf{Step 1: (\ref{eq:claim1}) continues to hold as long as (\ref{eq:claim2}) holds}\\
For $t\ge 0$, let $i\in \S$ with $x^\star_i,X_i^t \ge 0$. We will show that also $X_i^{t+1}\ge 0$. Negative coordinates can be treated the same way, and for $i\notin \S$ there is nothing to show. Rearranging (\ref{eq:update}) shows that $X_i^{t+1}\ge 0$ if and only if
\begin{equation}\label{eq:condition}
\frac{\beta^2}{\bigl(\sqrt{(X_i^t)^2+\beta^2} + X_i^t\bigr)^2} = \frac{\sqrt{(X_i^t)^2+\beta^2} - X_i^t}{\sqrt{(X_i^t)^2 + \beta^2} + X_i^t} \le \exp\bigl(-2\eta\nabla F(\X^t)_i \bigr).
\end{equation}
In order to show (\ref{eq:condition}), we need to bound the gradient $\nabla F(\X^t)_i$ from above. 
Since by the induction hypothesis we have $\|\X_\S^t\|_1 = \|\X^t\|_1 - \|\X_{\S^c}^t\|_1 \ge \frac{1}{2}$, we can apply the simplified bound in Lemma \ref{lemma:support1} to bound, for $t\le T$,
\begin{equation*}
|\nabla F(\X^t)_i - \nabla f(\X^t)_i| \le \frac{c_\star^2\|\X_\S^t\|_1}{k}
\end{equation*}
with probability $1-c_pn^{-10}$, provided that $c_s$ is sufficiently large, where we used the assumption that $\delta \le c_1/n$.
Further, we can bound, since (\ref{eq:claim1}) holds at time $t$ by induction hypothesis, 
\begin{equation*}
(\X^t)^\top\x^\star \ge \|\X_\S^t\|_1 x^\star_{min} \ge \|\X_\S^t\|_1\frac{c_\star}{\sqrt{k}},
\end{equation*}
which leads to
\begin{align}
\nabla F(\X^t)_i &\le \nabla f(\X^t)_i + |\nabla F(\X^t)_i - \nabla f(\X^t)_i| \nonumber\\
&\le \big(3\|\X^t\|_2^2 - 1\big)X_i^t - \big((\X^t)^\top\x^\star\big)x^\star_i. \label{eq:bound_gradient}
\end{align}
If $X_i^t \le \frac{c_\star}{2\sqrt{3k}}$, then $\nabla F(\X^t)_i \le 0$ by (\ref{eq:claim3}). In this case, (\ref{eq:condition}) holds as 
\begin{equation*}
\frac{\beta^2}{\big(\sqrt{(X_i^t)^2+\beta^2} + X_i^t\big)^2} \le 1 \le \exp\bigl(-2\eta\nabla F(\X^t)_i \bigr).
\end{equation*}
If on the other hand $X_i^t>\frac{c_\star}{2\sqrt{3k}}$, then we can bound $\nabla F(\X^t)_i \le 5X_i^t$ since $\|\X^t\|_2^2\le 2$, and
\begin{equation*}
\frac{\beta^2}{\bigl(\sqrt{(X_i^t)^2+\beta^2} + X_i^t\bigr)^2} \le \frac{\beta^2}{4(X_i^t)^2} \le \exp\bigl(-10\eta X_i^t\bigr)
\end{equation*}
holds for all $\eta \le \frac{1}{\sqrt{200}}$, where we used that $X_i^t\le \sqrt{2}$ by (\ref{eq:claim2}) and $\log \frac{c_\star^2}{3k\beta^2}\ge 1$ since by assumption $\beta \le c_1/n$. This completes the proof that (\ref{eq:condition}) holds. Repeating the same argument for $i\in\S$ with $x^\star_i,X^t_i \le 0$ shows that (\ref{eq:claim1}) is satisfied at $t+1$.

\textbf{Step 2: (\ref{eq:claim2}) continues to hold as long as (\ref{eq:claim1}) and (\ref{eq:claim3}) hold}\\
We begin by inductively showing the lower bound in (\ref{eq:claim2}). For $t\ge 0$, we write
\begin{equation*}
\|\X^{t+1}\|_2^2 - \|\X^t\|_2^2 = \sum_{i=1}^n \bigl(X_i^{t+1} + X_i^t\bigr)\bigl(X_i^{t+1}-X_i^t\bigr),
\end{equation*}
and consider the partial sums containing all $i\in \S$ and $i\notin \S$, respectively.
For $i\notin \S$, we have
\begin{equation*}
\sum_{i\notin \S} \bigl(X_i^{t+1} + X_i^t\bigr)\bigl(X_i^{t+1}-X_i^t\bigr) \ge -\sum_{i\notin \S} (X_i^t)^2 \ge -n\delta^2,
\end{equation*}
where we used $\|\X_{\S^c}^t\|_2\le \sqrt{n}\|\X_{\S^c}^t\|_1 \le \sqrt{n}\delta$.

To bound the partial sum containing $i\in\S$, we distinguish two cases: $\|\X^t\|_2^2 - \frac{1}{3}\ge 2n\delta^2$ and $\|\X^t\|_2^2 - \frac{1}{3} < 2n\delta^2$. In the first case, let $i\in \S$ and without loss of generality $x^\star_i > 0$. If $X_i^{t+1} \ge X_i^t$, then the summand $(X_i^{t+1} + X_i^t)(X_i^{t+1}-X_i^t)$ is non-negative. If on the other hand $X_i^{t+1} < X_i^t$, then $\nabla F(\X^t)_i>0$, and we can bound, recalling (\ref{eq:flag}), 
\begin{align*}
X_i^{t+1} - X_i^t \ge -\frac{3}{2}\eta\nabla F(\X^t)_i \sqrt{(X_i^t)^2 + \beta^2} \ge -\frac{3}{2}\eta\big(3\|\X^t\|_2^2 - 1\big) X_i^t \sqrt{(X_i^t)^2 + \beta^2},
\end{align*}
where we used (\ref{eq:bound_gradient}) for the last inequality. An analogous upper bound can be derived for negative coordinates with $x^\star_i<0$, so, writing $\S' = \{i\in \S: |X_i^{t+1}|<|X_i^t|\}$ for the coordinates which decrease in magnitude, we can bound
\begin{align*}
\sum_{i\in \S'}\bigl(X_i^{t+1} + X_i^t\bigr)\bigl(X_i^{t+1}-X_i^t\bigr) &\ge -3\eta\sum_{i\in \S'}X_i^t\big(3\|\X^t\|_2^2 - 1\big) X_i^t \sqrt{(X_i^t)^2 + \beta^2} \\
&\ge -\biggl(9\eta\sum_{i\in \S'} (X_i^t)^2\big(|X_i^t|+\beta \big)\biggr) \Bigl(\|\X^t\|_2^2 - \frac{1}{3}\Bigr),
\end{align*}
where for the last inequality we used that $\sqrt{x^2+\beta^2}\le |x| + \beta$ for $\beta \ge 0$. Hence, we have 
\begin{equation*}
\|\X^{t+1}\|_2^2 - \|\X^t\|_2^2 \ge \frac{1}{3}-\|\X^t\|_2^2
\end{equation*}
for $\eta \le \frac{1}{54}$, where we used that $\sum_{i\in \S'}(X_i^t)^2(|X_i^t|+\beta)\le 3$ since the upper bound in (\ref{eq:claim2}) holds at time $t$ by induction hypothesis.

Next, we consider the case $\|\X^t\|_2^2 - \frac{1}{3} < 2n\delta^2$. In this case, it follows from (\ref{eq:bound_gradient}) that every coordinate $i\in \S$ must be increasing in magnitude, i.e.\ $|X_i^{t+1}|>|X_i^t|$ for all $i\in\S$. Assuming without loss of generality that $x^\star_i>0$, we can use (\ref{eq:flag_lb}) to bound by how much $X_i^t$ must at least increase:
\begin{align*}
X_i^{t+1}- X_i^t &\ge -\frac{1}{2}\eta\nabla F(\X^t)_i \sqrt{(X_i^t)^2 + \beta^2} \\
&\ge -\frac{1}{2}\eta\bigl(\big(3\|\X^t\|_2^2 - 1\big) X_i^t - \big((\X^t)^\top\x^\star\big)x^\star_i\bigr) \sqrt{(X_i^t)^2 + \beta^2}\\
&\ge -\frac{1}{2}\eta \biggl(6n\delta^2 X_i^t - \frac{c_\star^2}{2k}\biggr) \sqrt{(X_i^t)^2 + \beta^2} \\
&\ge \frac{c_\star^2\eta}{5k} X_i^t,
\end{align*}
where we used (\ref{eq:bound_innerproduct}) together with the assumption $x^\star_{min}\ge \frac{c_\star}{\sqrt{k}}$, and $\delta \le \frac{c_1}{n}$. Hence, we have 
\begin{align*}
\sum_{i\in \S}\big(X_i^{t+1} + X_i^t\big)\big(X_i^{t+1}-X_i^t\big) &\ge \sum_{i\in \S}2X_i^t \frac{c_\star^2\eta}{5k}X_i^t\\
&\ge \eta \frac{2c_\star^2}{5k}\biggl(\frac{1}{3} - (3 + 9\sigma)\sqrt{\frac{\log n}{m}} - \delta^2\biggr),
\end{align*}
where we used the induction hypothesis that the lower bound in (\ref{eq:claim2}) is satisfied at time $t$ along with the bound $\|\X_{\S^c}^t\|_2^2 \le \|\X_{\S^c}^t\|_1^2\le \delta^2$. 
Thus, we have $\|\X^{t+1}\|_2^2 - \|\X^t\|_2^2 >0$ if $\delta\le c_1/n$ for $c_1>0$ sufficiently small.

The upper bound in (\ref{eq:claim2}) can be shown the same way as in the continuous-time case.

\textbf{Step 3: (\ref{eq:claim3}) continues to hold as long as (\ref{eq:claim1}) and \ref{eq:claim2} hold}\\
As in the continuous-time case, we can restrict our attention to the case $3\|\X^t\|_2^2 - 1 > 0$ (since otherwise there is nothing to show), and consider the ratio 
\begin{equation*}
R_t = \frac{2((\X^t)^\top\x^\star)}{3\|\X^t\|_2^2 - 1}.
\end{equation*} 
As in the continuous-time case, we can show $R_0 \ge 1/\sqrt{3}$.

To show $R_{t+1}\ge 1/\sqrt{3}$ for any $t\ge 0$, we consider the two cases $R_t\ge 1/\sqrt{3} + 0.05$ and $1/\sqrt{3} \le R_t < 1/\sqrt{3} + 0.05$. We assume $3\|\X^t\|_2^2 - 1 > 0$ for convenience's sake, as the other case $3\|\X^t\|_2^2 - 1 < 0$ can be treated the same way, but requires distinguishing cases at various points in the following argument, which we omit to avoid repetition.

\textbf{Step 3, Case 1: $R_t \ge 1/\sqrt{3} + 0.05$}\\
We consider four types of coordinates, defined by
\begin{align*}
\S^c &= \{i\in [n]: x^\star_i = 0\} && \text{(off-support)}\\
\S_1 &= \biggl\{i\in \S: |X_i^t|\le \frac{1}{3}|x^\star_i|\biggr\}  && \text{(small coordinates)}\\
\S_2 &= \biggl\{i\in \S: |X_i^t|> \frac{1}{3}|x^\star_i|,\; |X_i^{t+1}| > |X_i^t|\biggr\} && \text{(large increasing coordinates)} \\
\S_3 &= \biggl\{i\in \S: |X_i^t|> \frac{1}{3}|x^\star_i|,\; |X_i^{t+1}| \le |X_i^t|\biggr\} && \text{(large decreasing coordinates)}
\end{align*}
We will show $R_{t+1} \ge 1/\sqrt{3}$ by considering the following sequence of vectors defined by
\begin{equation*}
\begin{gathered}
\X^{(1)} = \X^t, \quad 
X^{(2)}_i = \begin{cases}X^{(1)}_i & i\notin \S^c \\ X_i^{t+1} & i\in \S^c\end{cases}, \quad 
X^{(3)}_i = \begin{cases}X^{(2)}_i & i\notin \S_2 \\ X_i^{t+1} & i\in \S_2\end{cases}, \\
X^{(4)}_i = \begin{cases}X^{(3)}_i & i\notin \S_3 \\ X_i^{t+1} & i\in \S_3\end{cases}, \quad
\X^{(5)} = \X^{t+1},
\end{gathered}
\end{equation*}
and bound the ratio $\frac{2((\X^{(i)})^\top\x^\star)}{3\|\X^{(i)}\|_2^2 - 1}$ for each $i=2,...,5$. 

\textbf{Exchanging off-support coordinates}\\
By assumption, we have $R_t \ge 1/\sqrt{3} + 0.05$, which can also be written as 
\begin{equation*}
\biggl(\frac{1}{\sqrt{3}} + 0.05\biggr)^{-1} \cdot 2\big((\X^{(1)})^\top\x^\star\big) \ge 3\|\X^{(1)}\|_2^2 - 1.
\end{equation*}
As shown in Step 2, $\|\X^{(2)}\|_2^2 \le \|\X^{(1)}\|_2^2 + n\delta^2$, so we have, since $((\X^{(1)})^\top\x^\star) \ge \frac{c_\star}{2\sqrt{k}}$,
\begin{equation}\label{eq:exchange2}
\biggl(\frac{1}{\sqrt{3}} + 0.04\biggr)^{-1}\cdot 2\big((\X^{(2)})^\top\x^\star\big) \ge 3\|\X^{(2)}\|_2^2 - 1,
\end{equation}
provided that $\delta \le c_1/n$ is sufficiently small.

\textbf{Exchanging large increasing coordinates}\\
Let $i\in \S_2$ with $x^\star_i>0$. Analogous bounds for the case $x^\star_i<0$ can be derived the same way. As $X_i^{t+1}>X_i^t$, we must have $\nabla F(\X^t)_i<0$. First, we bound by how much $X_i^t$ can increase in one iteration. We can write
\begin{align*}
\nabla f(\X^t)_i = \frac{2((\X^t)^\top\x^\star)}{R_t}X_i^t - 2\big((\X^t)^\top\x^\star\big)x^\star_i = -2\big((\X^t)^\top\x^\star\big) \biggl(x^\star_i - \frac{1}{R_t}X_i^t\biggr). 
\end{align*}
Using the simplified bound of Lemma \ref{lemma:support1}, we have with probability $1-c_pn^{-10}$,
\begin{equation*}
|\nabla F(\X^t)_i - \nabla f(\X^t)_i| \le 0.1 \|\X_\S^t\|_1\frac{c_\star^2}{k}\le 0.1\big((\X^t)^\top\x^\star\big)x^\star_i,
\end{equation*}
where we used (\ref{eq:innerproduct}). This gives the lower bound
\begin{equation}\label{eq:gradient_lowerbound}
\nabla F(\X^t)_i \ge -2\big((\X^t)^\top\x^\star\big) \biggl(1.05x^\star_i - \frac{1}{R_t}X_i^t\biggr) \ge -1.05\cdot 2\big((\X^t)^\top\x^\star\big) x^\star_i,
\end{equation}
since we assumed $x^\star_i>0$, and hence $X_i^t\ge 0$ by (\ref{eq:claim1}). With this, we can use (\ref{eq:flag}) to bound
\begin{align}
X_i^{t+1} - X_i^t &\le -\frac{3}{2}\eta\nabla F(\X^t)_i \sqrt{(X_i^t)^2 + \beta^2} \nonumber \\ &\le 1.7\eta \cdot 2\big((\X^t)^\top\x^\star\big) x^\star_i|X_i^t|, \label{eq:difference_upperbound}
\end{align}
where we used that $\sqrt{(X_i^t)^2+\beta^2} \le 1.05 |X_i^t|$ for $i\in \S_2$.
Further, we can use the upper bound in (\ref{eq:claim2}) to bound $X_i^{t+1} + X_i^t \le 3$. With this, we obtain
\begin{align*}
\frac{2((\X^{(3)})^\top\x^\star)}{3\|\X^{(3)}\|_2^2-1} &\ge \frac{2((\X^{(2)})^\top\x^\star)}{3\|\X^{(2)}\|_2^2 - 1 + \sum_{i\in \S_2}3(X_i^{t+1}+X_i^t)(X_i^{t+1} - X_i^t)} \\
&\ge \frac{2((\X^{(2)})^\top\x^\star)}{3\|\X^{(2)}\|_2^2 - 1 + 3 \cdot 3 \cdot 1.7 \cdot 2\eta ((\X^t)^\top\x^\star)^2}, 
\end{align*}
which is bounded from below by $1/\sqrt{3}$, provided that
\begin{equation*}
\eta \le \frac{1}{193} \le \frac{1}{15.3((\X^t)^\top\x^\star)}\biggl(\sqrt{3} - \frac{3\|\X^{(2)}\|_2^2 - 1}{2((\X^{(2)})^\top\x^\star)}\biggr),
\end{equation*}
where we used that $(\X^{(2)})^\top\x^\star = (\X^t)^\top\x^\star \le \sqrt{2}$ by (\ref{eq:claim2}) and inequality (\ref{eq:exchange2}).

\textbf{Exchanging large decreasing coordinates}\\
Let $i\in \S_3$ with $x^\star_i>0$. Analogous bounds for the case $x^\star_i < 0$ can be derived the same way. We first bound by how much $X_i^t$ can decrease in one iteration. Following the same steps as in the derivation of the lower bound (\ref{eq:gradient_lowerbound}), we obtain the analogous upper bound
\begin{equation}\label{eq:gradient_upperbound}
\nabla F(\X^t)_i \le 2\big((\X^t)^\top\x^\star\big) \biggl(\frac{1}{R_t}X_i^t - 0.95 x^\star_i\biggr).
\end{equation}
Since we are considering a positive decreasing coordinate, we must have $\nabla F(\X^t)_i>0$ and therefore $X_i^t \ge 0.95R_tx^\star_i$. As before, we can use (\ref{eq:flag}) to bound
\begin{align}
X_i^{t+1} - X_i^t &\ge -\frac{3}{2}\eta \nabla F(\X^t)_i \sqrt{(X_i^t)^2+\beta^2} \nonumber\\
&\ge -3.2\eta \big((\X^t)^\top\x^\star\big)\frac{1}{R_t}X_i^t\big(X_i^t - 0.95R_t x^\star_i\big), \label{eq:difference_lowerbound}
\end{align}
where we used that $\sqrt{(X_i^t)^2+\beta^2} \le 1.05 |X_i^t|$ for $i\in \S_3$. Further, we can write
\begin{equation*}
X_i^{t+1} \ge \biggl(1 - 3.2\eta \big((\X^t)^\top\x^\star\big)\frac{1}{R_t}X_i^t\biggr) X_i^t + 3.2\eta \big((\X^t)^\top\x^\star\big)\frac{1}{R_t}X_i^t \cdot 0.95R_tx^\star_i \ge 0.95R_tx^\star_i,
\end{equation*}
provided that $\eta \le \frac{1}{6.4\sqrt{3}} \le \frac{R_t}{3.2((\X^t)^\top\x^\star)X_i^t}$, where we used (\ref{eq:claim2}) to bound $((\X^t)^\top\x^\star)X_i^t\le 2$ and $X_i^t\ge 0.95R_tx^\star_i$. We can now bound
\begin{align*}
3\|\X^{(4)}\|_2^2 - 1 &= 3\|\X^{(3)}\|_2^2 - 1 + \sum_{i\in \S_3}3(X_i^{t+1}+X_i^t)(X_i^{t+1}-X_i^t) \\
&\le 3\|\X^{(3)}\|_2^2 - 1 + \sum_{i\in \S_3}3\cdot 2\cdot 0.95R_tx^\star_i(X_i^{t+1}-X_i^t) \\
&\le \sqrt{3}\biggl( 2((\X^{(3)})^\top\x^\star) + \sum_{i\in \S_3}2x^\star_i (X_i^{t+1}-X_i^t)\biggr)\\
&= \sqrt{3}\cdot 2\big((\X^{(4)})^\top\x^\star\big),
\end{align*}
where the second inequality holds because $3\cdot 0.95R_t \ge \sqrt{3}$ for $R_t\ge \frac{1}{\sqrt{3}} + 0.05$, and $x^\star_i(X_i^{t+1}-X_i^t)<0$.

\textbf{Exchanging small coordinates}\\
Let $i\in \S_1$ with $x^\star_i>0$. Analogous bounds for the case $x^\star_i<0$ can be derived the same way. Then, we have $\nabla F(\X^t)_i<0$ by (\ref{eq:gradient_upperbound}) as $X_i^t \le \frac{1}{3}x^\star_i$, that is all small coordinates must increase in magnitude. Further, we can use (\ref{eq:flag}) to bound
\begin{align*}
X_i^{t+1} + X_i^t 
\le |X_i^{t+1} - X_i^t| + 2X_i^t 
\le -\frac{3}{2}\eta\nabla F(\X^t)_i \sqrt{(X_i^t)^2 + \beta^2} + \frac{2}{3}x^\star_i
\le \frac{2}{\sqrt{3}}x^\star_i,
\end{align*}
if $\eta \le 0.3$, where we used that $(\X^t)^\top\x^\star\le \sqrt{2}$ by (\ref{eq:claim2}), $\sqrt{(X_i^t)^2+\beta^2}\le 0.35$ by the definition of $\S_1$, and the lower bound (\ref{eq:gradient_lowerbound}) for the gradient $\nabla F(\X^t)_i$. Hence, we can bound
\begin{align*}
3\|\X^{(5)}\|_2^2 - 1 &= 3\|\X^{(4)}\|_2^2 - 1 + \sum_{i\in \S_1}3(X_i^{t+1} + X_i^t)(X_i^{t+1} - X_i^t) \\*
& \le \sqrt{3}\biggl(2((\X^{(4)})^\top\x^\star) + \sum_{i\in \S_1}2x^\star_i(X_i^{t+1} - X_i^t)\biggr)\\
&= \sqrt{3} \cdot 2\big((\X^{(5)})^\top\x^\star\big).
\end{align*}
This completes the proof of $R_{t+1} \ge 1/\sqrt{3}$ if $R_t\ge 1/\sqrt{3} + 0.05$.

\textbf{Step 3, Case 2: $1/\sqrt{3} \le R_t < 1/\sqrt{3} + 0.05$}\\
The proof that $R_{t+1}\ge 1/\sqrt{3}$ largely follows the same steps as the previous case, although it requires a different sequence of vectors interpolating between $\X^t$ and $\X^{t+1}$. Define $\widetilde{\X}^{t+1}$ by
\begin{equation*}
\nabla \Phi(\widetilde{\X}^{t+1}) = \nabla \Phi(\X^t) - \eta \nabla f(\X^t),
\end{equation*}
that is the vector obtained by applying one step of mirror descent with the population gradient. 
Replacing $X_i^{t+1}$ by $\widetilde{X}_i^{t+1}$ in the definitions of the sets $\S_i$, we consider the sequence
\begin{equation*}
\begin{gathered}
\X^{(1)} = \X^t, \quad 
X^{(2)}_i = \begin{cases}X^{(1)}_i & i\notin \S_2 \\ \widetilde{X}_i^{t+1} & i\in \S_2\end{cases}, \quad 
X^{(3)}_i = \begin{cases}X^{(2)}_i & i\notin \S_3 \\ \widetilde{X}_i^{t+1} & i\in \S_3\end{cases}, \\
X^{(4)}_i = \begin{cases}X^{(3)}_i & i\notin \S_2\cup \S_3 \\ X_i^{t+1} & i\in \S_2\cup\S_3\end{cases}, \quad
X^{(5)}_i = \begin{cases}X^{(4)}_i & i\notin \S^c \\ X_i^{t+1} & i\in \S^c\end{cases}, \quad
\X^{(6)} = \X^{t+1}.
\end{gathered}
\end{equation*}

\textbf{Exchanging large increasing coordinates with population gradient}\\
Let $i\in \S_2$ with $x^\star_i>0$. Analogous bounds for the case $x^\star_i<0$ can be established the same way.
Since we are considering an increasing coordinate, the population gradient
\begin{equation*}
\nabla f(\X^t)_i = -\big(3\|\X^t\|_2^2-1\big)\big(R_tx^\star_i-X_i^t\big)
\end{equation*}
must be negative, so we must have $X_i^t \le R_tx^\star_i$. The bound (\ref{eq:difference_upperbound}) becomes
\begin{equation*}
\widetilde{X}_i^{t+1} - X_i^t \le 1.6\eta \big(3\|\X^t\|_2^2 -1\big)\big(R_tx^\star_i - X_i^t\big) X_i^t,
\end{equation*}
and we can bound
\begin{equation*}
\widetilde{X}_i^{t+1}\le \Bigl(1 + 1.6\eta\big(3\|\X^t\|_2^2 - 1\big)\big(R_tx^\star_i - X_i^t\big)\Bigr)X_i^t \le R_tx^\star_i
\end{equation*}
if $\eta\le \frac{1}{8\sqrt{2}} \le \frac{1}{1.6(3\|\X^t\|_2^2-1)X_i^t}$, as the expression is increasing in $\eta$. We can bound the ratio
\begin{align*}
\frac{2((\X^{(2)})^\top\x^\star)}{3\|\X^{(2)}\|_2^2 - 1} &= \frac{2((\X^t)^\top\x^\star) + \sum_{i\in\S_2}2x^\star_i(\widetilde{X}_i^{t+1} - X_i^t)}{3\|\X^t\|_2^2 - 1 + \sum_{i\in\S_2}3(\widetilde{X}_i^{t+1} + X_i^t)(\widetilde{X}_i^{t+1} - X_i^t)}\\
&\ge \frac{2((\X^t)^\top\x^\star) + B}{3\|\X^t\|_2^2-1 + 3R_tB},
\end{align*}
where we write $B=\sum_{i\in \S_2}2x^\star_i(\widetilde{X}_i^{t+1} - X_i^t)$. This ratio is bounded from below by $1/\sqrt{3}$ if 
\begin{equation*}
B\le \frac{2((\X^t)^\top\x^\star) - (3\|\X^t\|_2^2 - 1)/\sqrt{3}}{\sqrt{3}R_t-1} = \frac{3\|\X^t\|_2^2-1}{\sqrt{3}},
\end{equation*}
which is satisfied as, using $|X_i^t| \le R_t|x^\star_i|$, we can bound
\begin{align*}
B &= \sum_{i\in \S_2}2x^\star_i(\widetilde{X}_i^{t+1} - X_i^t) \\
&\le \sum_{i\in \S_2}2x^\star_i\cdot 1.6\eta\big(3\|\X^t\|_2^2 - 1\big) \big(R_tx^\star_i - X_i^t\big) |X_i^t| \\
&\le \big(3\|\X^t\|_2^2-1\big) \cdot 3.2\eta R_t^2 \sum_{i\in \S_2}|x^\star_i|^3 \\
&\le \frac{3\|\X^t\|_2^2-1}{\sqrt{3}}
\end{align*}
for $\eta \le 0.45 \le \frac{1}{3.2\sqrt{3}R_t^2}$, where we used that $\sum_{i\in\S_2}|x^\star_i|^3\le 1$ and $R_t \le 1/\sqrt{3} + 0.05$.

\textbf{Exchanging large decreasing coordinates with population gradient}\\
Let $i\in \S_3$ with $x^\star_i>0$. Analogous bounds for the case $x^\star_i < 0$ can be established the same way.
Following the same steps as before, we can bound by how much $X_i^t$ can decrease after one step of mirror descent using the population gradient, and the bound (\ref{eq:difference_lowerbound}) becomes
\begin{equation*}
X^{(3)}_i - X_i^t \ge -3.2\eta \big((\X^t)^\top\x^\star\big)\biggl(\frac{1}{R_t}X_i^t - x^\star_i \biggr)X_i^t.
\end{equation*}
As in the previous step, rearranging this inequality shows that $X^{(3)}_i \ge R_tx^\star_i$ provided that $\eta \le 0.39 \le \frac{R_t}{3.2((\X^t)^\top\x^\star)X_i^t}$. We will show
\begin{equation*}
3\|\X^{(3)}\|_2^2 - 1 \le \sqrt{3}\cdot 2\big((\X^{(3)})^\top\x^\star\big) - \Delta,
\end{equation*} 
where $\Delta > 0$ is a buffer term due to very large coordinates decreasing. To this end, define $\S'_3 = \{i\in \S_3: |X_i^t| \ge \frac{2}{\sqrt{3} + 0.15}|x^\star_i|\}$. Then, since $R_t \le 1/\sqrt{3} + 0.05$, we have 
\begin{align*}
\frac{1}{1/\sqrt{3} + 0.05} \cdot 2\big((\X^t)^\top\x^\star\big) &\le 3\|\X^t\|_2^2-1 \\
\Rightarrow 3\sum_{i\in \S'_3} (X_i^t)^2 - \frac{2}{1/\sqrt{3} + 0.05}X_i^tx^\star_i &\ge 1 + \sum_{i\notin \S'_3} \frac{2}{1/\sqrt{3} + 0.05}X_i^t x^\star_i - 3(X_i^t)^2 \ge 1-3\delta^2,
\end{align*}
since $\frac{2}{1/\sqrt{3} + 0.05}X_i^t x^\star_i - 3(X_i^t)^2 \ge 0$ if $i\notin \S'_3\cup \S^c$, and $\|\X_{\S^c}^t\|_2^2\le \delta^2$. This implies
\begin{equation*}
\sum_{i\in\S'_3}(X_i^t)^2\ge \frac{1}{3}. 
\end{equation*}
Next, we bound by how much coordinates in $\S'_3$ must at least decrease. Let $i\in \S'_3$ with $x^\star_i>0$. We have
\begin{align*}
\nabla f(\X^t)_i &= 2\big((\X^t)^\top\x^\star\big)\biggl(\frac{1}{R_t}X_i^t - x^\star_i\biggr) \\
&\ge 2\big((\X^t)^\top\x^\star\big)\Biggl(\frac{2}{(\sqrt{3} + 0.15)(\frac{1}{\sqrt{3}} + 0.05)} - 1\Biggr)x^\star_i \\
&\ge 2\big((\X^t)^\top\x^\star\big) \frac{2c_\star}{3\sqrt{k}}.
\end{align*}
With this, we can bound, using (\ref{eq:flag_lb}),
\begin{align*}
\widetilde{X}_i^{t+1} - X_i^t \le -\frac{1}{2}\eta \nabla f(\X^t)_i X_i^t \le -\eta \cdot 2\big((\X^t)^\top\x^\star\big) \frac{c_\star}{3\sqrt{k}}X_i^t.
\end{align*}
Following the same steps as in Case 1, we can show that $\widetilde{X}_i^{t+1},X_i^t \ge R_tx^\star_i \ge x^\star_i/\sqrt{3}$ provided that $\eta \le \frac{1}{6.4\sqrt{3}}$, where we used that $R_t \ge 1/\sqrt{3}$. With this, we can bound 
\begin{align*}
3\|\X^{(3)}\|_2^2 - 1 &= 3\|\X^{(2)}\|_2^2 - 1 + \sum_{i\in \S_3\backslash \S'_3}3(\widetilde{X}_i^{t+1} + X_i^t)(\widetilde{X}_i^{t+1}-X_i^t) \\
&\quad + \sum_{i\in \S'_3}3(\widetilde{X}_i^{t+1} + X_i^t)(\widetilde{X}_i^{t+1} - X_i^t) \\
&\le \sqrt{3}\biggl(2\big((\X^{(2)})^\top\x^\star\big) + \sum_{i\in \S_3\backslash \S'_3}2x^\star_i (\widetilde{X}_i^{t+1} - X_i^t)\biggr) \\
&\quad + \sum_{i\in \S'_3}3(R_t + 1)X_i^t(\widetilde{X}_i^{t+1} - X_i^t) \\
&\le \sqrt{3}\cdot 2\big((\X^{(3)})^\top\x^\star\big) + \sum_{i\in \S'_3}\Big(\sqrt{3} + 3 - \big(\sqrt{3} + 0.15\big)\Big)X_i^t\big(\widetilde{X}_i^{t+1} - X_i^t\big),
\end{align*}
where in the last line we used that $2|x^\star_i| \le (\sqrt{3} + 0.15)|X_i^t|$ for $i\in \S'_3$.
Finally, we have
\begin{align*}
\sum_{i\in \S'_3}(3 - 0.15)X_i^t(\widetilde{X}_i^{t+1} - X_i^t) &\le -\sum_{i\in \S'_3}2.85X_i^t \eta\cdot 2\big((\X^t)^\top\x^\star\big) \frac{c_\star}{3\sqrt{k}}X_i^t\\
&\le -2.85\eta \cdot 2\big((\X^t)^\top\x^\star\big) \frac{c_\star}{3\sqrt{k}} \cdot \frac{1}{3}\\
&=: -\Delta,
\end{align*}
that is 
\begin{equation}\label{eq:buffer}
3\|\X^{(3)}\|_2^2 - 1 \le \sqrt{3}\cdot 2\big((\X^{(3)})^\top\x^\star\big) - \Delta.
\end{equation}

\textbf{Exchanging large coordinates with empirical gradient}\\
Let $i\in \S_2\cup \S_3$ with $x^\star_i>0$. Analogous bounds for the case $x^\star_i<0$ can be established the same way.
 
In the following, we write $a=\frac{1}{2}(X_i^t + \sqrt{(X_i^t)^2+\beta^2})$, $b = \frac{1}{2}(-X_i^t + \sqrt{(X_i^t)^2 + \beta^2})$, $G = \nabla F(\X^t)_i$ and $g = \nabla f(\X^t)_i$ for notational brevity. We can bound the ratio
\begin{align*}
\frac{X_i^{t+1}}{\widetilde{X}_i^{t+1}} &= \frac{ae^{-\eta G} - be^{\eta G}}{ae^{-\eta g} - be^{\eta g}} \\
&\le \frac{ae^{-\eta G}}{ae^{-\eta g} - be^{\eta g}}\\
&=\frac{ae^{-\eta G}}{ae^{-\eta g}} + \frac{ae^{-\eta G} be^{\eta g}}{ae^{-\eta g}(ae^{-\eta g} - be^{\eta g})}\\
&\le \exp\biggl(\eta \frac{0.01c_\star((\X^t)^\top\x^\star)}{\sqrt{k}}\biggr) + \frac{3\sqrt{k}}{c_\star}\beta \\
&\le 1 + 0.02c_\star\eta \frac{((\X^t)^\top\x^\star)}{\sqrt{k}},
\end{align*}
where for the second last line we used the inequality $-x + \sqrt{x^2 + \beta^2} \le \beta$ for $x>0$, the fact that $|X_i^t| \ge |x^\star_i|/3 \ge c_\star/(3\sqrt{k})$ for $i\in \S_2\cup\S_3$, and that, as in (\ref{eq:bound_gradient_difference}), we use Lemma \ref{lemma:support1} to bound, with probability $1-c_pn^{-10}$,
\begin{equation*}
|\nabla F(\X^t)_i - \nabla f(\X^t)_i| \le \frac{0.01c_\star((\X^t)^\top\x^\star)}{\sqrt{k}}.
\end{equation*}
Similarly, we can also bound $X_i^{t+1} / \widetilde{X}_i^{t+1} \ge 1 - 0.02c_\star\eta \frac{((\X^t)^\top\x^\star)}{\sqrt{k}}$. With this, we have the bounds
\begin{equation*}
\|\X^{(4)}\|_2^2 \le \biggl(1 + 0.02c_\star\eta \frac{((\X^t)^\top\x^\star)}{\sqrt{k}}\biggr)^2\|\X^{(3)}\|_2^2 \le \biggl(1 + 2.1\cdot 0.02c_\star\eta \frac{((\X^t)^\top\x^\star)}{\sqrt{k}}\biggr)\|\X^{(3)}\|_2^2,
\end{equation*}
provided that $\eta \le 5/(\sqrt{2}c_\star)$, and similarly
\begin{equation*}
\big((\X^{(4)})^\top\x^\star\big) \ge \biggl(1 - 0.02c_\star\eta \frac{((\X^t)^\top\x^\star)}{\sqrt{k}}\biggr)\big((\X^{(3)})^\top\x^\star\big).
\end{equation*}
Recalling the definition of $\Delta$ and using (\ref{eq:buffer}), we can now bound
\begin{align*}
3\|\X^{(4)}\|_2^2 - 1 &\le 3\|\X^{(3)}\|_2^2 - 1 + 3\|\X^{(3)}\|_2^2 \cdot 2.1\cdot 0.02c_\star\eta \frac{((\X^t)^\top\x^\star)}{\sqrt{k}} \\
&\le 3\|\X^{(3)}\|_2^2 - 1 + \frac{\Delta}{2} \\
&\le \sqrt{3} \cdot 2\big((\X^{(3)})^\top\x^\star\big) - \frac{\Delta}{2}\\
&\le \sqrt{3} \cdot 2\big((\X^{(4)})^\top\x^\star\big) + \sqrt{3} \cdot 2\big((\X^{(3)})^\top\x^\star\big)\cdot 0.02c_\star\eta \frac{((\X^t)^\top\x^\star)}{\sqrt{k}} - \frac{\Delta}{2}\\
&\le \sqrt{3} \cdot 2\big((\X^{(4)})^\top\x^\star\big) - \frac{\Delta}{4},
\end{align*}
where the last inequality holds because we can show $\|\X^{(3)}\|_2^2 \le 2$ the same way as in the proof of the upper bound of (\ref{eq:claim2}).

\textbf{Exchanging off-support coordinates}\\
We have $\|\X^{(5)}\|_2^2 - \|\X^{(4)}\| \le n\delta^2\le \Delta/4$ for $\delta\le c_1/n$ small enough, so
\begin{equation*}
3\|\X^{(5)}\|_2^2 - 1 \le \sqrt{3} \cdot 2\big((\X^{(5)})^\top\x^\star\big).
\end{equation*}

\textbf{Exchanging small coordinates}\\
Following the same steps for small coordinates in the previous case, we can show that also
\begin{equation*}
3\|\X^{(6)}\|_2^2 - 1 \le \sqrt{3} \cdot 2\big((\X^{(6)})^\top\x^\star\big),
\end{equation*}
which completes the proof that $R_{t+1} \ge 1/\sqrt{3}$ if $1/\sqrt{3}\le R_t<1/\sqrt{3} + 0.05$.
\end{proof}

%% file: technical_lemmas.tex
%%%%%%%%%%%%%%%%%%%%%%%%%%%%%%%%%%%%%%%%%%%%%%
%%%% Technical lemmas:
\section{Technical lemmas}
\label{supp:technical_lemmas}
In this section, we collect technical lemmas and concentration bounds used in the proofs of Theorems \ref{theorem} and \ref{theorem_discrete} and the supporting Lemmas \ref{lemma:support1}--\ref{lemma:support3}.

\begin{theorem}\label{thm:ref1}(Proposition 34 \cite{V12})  
Let $g:\R^n\rightarrow \R$ be a Lipschitz continuous function with Lipschitz constant $\lambda$, i.e.\ $|g(\x) - g(\mathbf{y})|\le \lambda\|\x-\mathbf{y}\|_2$ for all $\x,\mathbf{y}\in\R^n$. Let $\A\in\R^n$ be a standard normal random vector. Then, for any $\epsilon>0$, we have
\begin{equation*}
\P\left[|g(\A) - \E[g(\A)]| \ge \epsilon \right] \le 2\exp\biggl(-\frac{\epsilon^2}{2\lambda^2}\biggr)
\end{equation*}
\end{theorem}

\begin{theorem}\label{thm:ref2}(Theorems 3.6, 3.7 \cite{CL06})  
Let $X_i$ be independent random variables satisfying $|X_i|\le \lambda$ for all $i\in [n]$. Let $X = \sum_{i=1}^nX_i$ and $\|X\| = \sqrt{\sum_{i=1}^n\E[X_i^2]}$. Then, for any $\epsilon > 0$, we have
\begin{equation*}
\P\left[|X - \E[X]| \ge \epsilon \right] \le \exp\biggl(-\frac{\epsilon^2}{2(\|X\|^2 + \lambda\epsilon/3)}\biggr)
\end{equation*}
\end{theorem}

We state the following Lemma from \cite{CLM16} without proof. While the first and last inequality were not shown in \cite{CLM16}, it can be done the same way as in the proof of Lemma A.5 in \cite{CLM16}. Convexity follows from the convexity of the operator norm.

\begin{lemma}\label{lemma:tech1} (Lemma A.5 \cite{CLM16})
Let $\{\mathbf{A}_{j}\}_{j=1}^m$ be a collection of i.i.d.\ $\gauss(0,\mathbf{I}_k)$ random vectors. For any $t>0$, let $\mathcal{A}\subseteq \R^{m\times k}$ be the set consisting of all $\{\mathbf{a}_{j}\}_{j=1}^m\in \R^{m\times k}$ satisfying the following:
\begin{align}
\|\mathbf{a}\|_{2\rightarrow 2} &\le \sqrt{m} + \sqrt{k} + t, \label{eq:tech1_1}\\
\|\mathbf{a}\|_{2\rightarrow 4} &\le (3m)^{\frac{1}{4}} + \sqrt{k} + t, \label{eq:tech1_2}\\
\|\mathbf{a}\|_{2\rightarrow 6} &\le (15m)^{\frac{1}{6}} + \sqrt{k} + t, \label{eq:tech1_3}\\
\|\mathbf{a}\|_{2\rightarrow 8} &\le (105m)^{\frac{1}{8}} + \sqrt{k} + t, \label{eq:tech1_4}
\end{align}
where we write
\begin{equation*}
\|\mathbf{a}\|_{2\rightarrow p} = \sup_{\|x\|_2 \le 1} \|\mathbf{a}\x\|_p.
\end{equation*}
Then, we have $P[\{\mathbf{A}_j\}_{j=1}^m\in \mathcal{A}] \ge 1 - 4\exp(-t^2/2)$. Further, the set $\mathcal{A}$ is convex.
\end{lemma}

The following lemma is a slight modification of the previous result.

\begin{lemma}\label{lemma:tech2}
Let $\{\mathbf{A}_{j}\}_{j=1}^m$ be a collection of i.i.d.\ $\gauss(0,\mathbf{I}_k)$ random vectors. For any $t>0$, let $\mathcal{A}\subseteq \R^{m\times k}$ be the set consisting of all $\{\mathbf{a}_{j}\}_{j=1}^m\in \R^{m\times k}$ satisfying the following: For all $\x\in\R^k$ with $\|\x\|_2= 1$,
\begin{align}
\Biggl(\sum_{j=1}^m (\mathbf{a}_j^\top\x)^4\Biggr)^{\frac{1}{4}} &\le (3m)^{\frac{1}{4}} + \sqrt{8\log (2k)}\|\x\|_1 + t, \label{eq:tech2_1} \\
\Biggl(\sum_{j=1}^m (\mathbf{a}_j^\top\x)^6\Biggr)^{\frac{1}{6}} &\le (15m)^{\frac{1}{6}} + \sqrt{8\log (2k)}\|\x\|_1 + t, \label{eq:tech2_2} \\
\Biggl(\sum_{j=1}^m (\mathbf{a}_j^\top\x)^8\Biggr)^{\frac{1}{8}} &\le (105m)^{\frac{1}{8}} + \sqrt{8\log (2k)}\|\x\|_1 + t. \label{eq:tech2_3} 
\end{align}
Then, we have $\P[\{\mathbf{A}_j\}_{j=1}^m\in\mathcal{A}]\ge 1 - 3\lceil \log_2\sqrt{k}\rceil\exp(-t^2/2)$. Further, the set $\mathcal{A}$ is convex.
\end{lemma}

\begin{lemma}\label{lemma:tech3}
Let $\x^\star\in \R^n$ be a $k$-sparse vector with $\|\x^\star\|_2=1$. Let $\{\mathbf{A}_j\}_{j=1}^m$  be a collection of i.i.d.\ $\gauss(0,\mathbf{I}_n)$ random vectors and $\{\varepsilon_j\}_{j=1}^m$ a collection of independent centered sub-exponential random variables with maximum sub-exponential norm $\sigma = \max_j\|\varepsilon_j\|_{\psi_1}$. There exist universal constants $c,c_s,c_p>0$ such that if $m\ge c_s\max\{k^2\log^2n, \;\log^5n\}$, then, with probability at least $1 - c_pn^{-13}$,
\begin{align}
\Bigg|\frac{1}{m}\sum_{j=1}^mA_{ji}^2A_{js}A_{jl} - \E\bigl[A_{1i}^2A_{1s}A_{1l}\bigr]\Bigg|&\le c\sqrt{\frac{\log n}{m}} && \text{for all } i,l,s\in [n], \label{eq:tech3_1}\\
\Bigg|\frac{1}{m}\sum_{j=1}^mA_{ji}^2A_{jl}(\A_{j,-i}^\top\x^\star_{-i}) - \E\bigl[A_{1i}^2A_{1l}(\A_{1,-i}^\top\x^\star_{-i})\bigr]\Bigg| &\le c\sqrt{\frac{\log n}{m}} && \text{for all } i,l\in[n], \label{eq:tech3_2} \\
\Bigg|\frac{1}{m}\sum_{j=1}^mA_{ji}A_{jl}(\A_{j,-i}^\top\x^\star_{-i})^2 - \E\bigl[A_{1i}A_{1l}(\A_{1,-i}^\top\x^\star_{-i})^2\bigr]\Bigg| &\le c\sqrt{\frac{\log n}{m}} && \text{for all } i,l\in[n], \label{eq:tech3_3} \\
\Bigg|\frac{1}{m}\sum_{j=1}^mA_{ji}^8 - 105\Bigg| &\le c\sqrt{\frac{\log^5 n}{m}} && \text{for all } i\in[n] \label{eq:tech3_4}, \\
\Bigg|\frac{1}{m}\sum_{j=1}^m\varepsilon_j A_{ji}A_{jl}\Bigg| &\le c\sigma\sqrt{\frac{\log n}{m}} && \text{for all } i,l\in[n], \label{eq:tech3_5} \\
\Bigg|\frac{1}{m}\sum_{j=1}^m\varepsilon_j A_{jl}(\A_{j}^\top\x^\star) \Bigg| &\le c\sigma\sqrt{\frac{\log n}{m}} && \text{for all } l\in[n]. \label{eq:tech3_6} 
\end{align}
\end{lemma}

\begin{lemma}\label{lemma:tech4}
Let $\x^\star\in \R^n$ be a $k$-sparse vector with $\|\x^\star\|_2=1$, and let $\{\mathbf{A}_{j}\}_{j=1}^m$ be a collection of i.i.d.\ $\gauss(0,\mathbf{I}_n)$ random vectors. There exist universal constants $c, c_s, c_p>0$ such that the following holds. Let $\mathcal{A}\subseteq \R^{m\times n}$ be the set consisting of all $\{\mathbf{a}_j\}_{j=1}^m\in\R^{m\times n}$ satisfying the following:
\begin{align*}
\Bigg|\frac{1}{m}\sum_{j=1}^m(\mathbf{a}_j^\top\x^\star)^4 - 3\Bigg| &\le c\sqrt{\frac{\log n}{m}} \\
\Bigg|\frac{1}{m}\sum_{j=1}^m(\mathbf{a}_j^\top\x^\star)^6 - 15\Bigg| &\le c\sqrt{\frac{\log^3 n}{m}} \\
\Bigg|\frac{1}{m}\sum_{j=1}^m(\mathbf{a}_j^\top\x^\star)^8 - 105\Bigg| &\le c\sqrt{\frac{\log^5 n}{m}}.
\end{align*}
Then, if $m\ge c_s\log^5n$, we have $\P[\{\mathbf{A}_j\}_{j=1}^m\in \mathcal{A}] \ge 1 - c_pn^{-12}$. Further, the set $\mathcal{A}$ is convex.
\end{lemma}

\begin{lemma}\label{lemma:tech5}
Let $\{\mathbf{A}_j\}_{j=1}^m$ be a collection of i.i.d.\ $\gauss(0,\mathbf{I}_k)$ random vectors. There exist universal constants $c,c_s,c_p>0$ such that if $m\ge c_s\max\{k^2\log^2 n, \; \log^5n\}$, where $n\ge k$ is any natural number, then the set $\mathcal{A}\subseteq \R^{m\times k}$ defined by
\begin{equation*}
\mathcal{A} = \Biggl\{\{\mathbf{a}_j\}_{j=1}^m\in \R^{m\times k}: \frac{1}{m}\sum_{j=1}^m\Biggl(\sum_{l\in\mathcal{L}}a_{jl}\Biggr)^4\le c|\mathcal{L}|^2 \text{ for all } \mathcal{L}\subseteq[k] \Biggr\}
\end{equation*}
satisfies $\P[\{\mathbf{A}_j\}_{j=1}^m\in\mathcal{A}]\ge 1 - c_pn^{-11}$.
Further, the set $\mathcal{A}$ is convex.
\end{lemma}

\begin{lemma}\label{lemma:tech6}
Let $\x^\star\in\R^n$ be a $k$-sparse vector, and let $\S = \{i\in [n]: x^\star_i\neq 0\}$ be its support. Let $\{\mathbf{A}_j\}_{j=1}^m$ be a collection of i.i.d.\ $\gauss(0,\mathbf{I}_n)$ random variables and $\{\varepsilon_j\}_{j=1}^m$ a collection of independent centered sub-exponential random variables with maximum sub-exponential norm $\sigma = \max_{j}\|\varepsilon_j\|_{\psi_1}$. There exist universal constants $c_s, c_p>0$ such that, for any constant $c_1>0$, there is a $c>0$ such that if $m\ge c_s\max\{k^2\log n, \;\log^5n \}$, then, with probability at least $1 - c_pn^{-10}$, the following holds:
For any $i\in [n]$, $\mathcal{K}\subseteq \S\backslash \{i\}$ and $\epsilon \ge c_1/m$, let $N_{\epsilon}^{\mathcal{K}}$ be a smallest $\epsilon$-net of $\{\x\in\R^n:\x_{\mathcal{K}^c} = \mathbf{0}, \|\x\|_2 = 1\}$. Then, for any $\mathcal{L}\subseteq \mathcal{K}$ with $|\mathcal{L}|\ge \frac{1}{2}|\mathcal{K}|$, and $\x\in N_{\epsilon}^{\mathcal{K}}$,
\begin{align}
\Bigg|\sum_{l\in \mathcal{L}}\frac{1}{m}\sum_{j=1}^mA_{ji}A_{jl}(\mathbf{A}_{j}^\top\x)^2\Bigg| &\le c |\mathcal{L}|\frac{\log n}{\sqrt{m}}, \label{eq:tech6_1}, \\
\Bigg|\sum_{l\in \mathcal{L}}\frac{1}{m}\sum_{j=1}^mA_{ji}^2A_{jl}(\mathbf{A}_{j, -l}^\top\x_{-l})\Bigg| &\le c |\mathcal{L}|\sqrt{\frac{\log n}{m}} \label{eq:tech6_2}, \\
\Bigg|\sum_{l\in \mathcal{L}}\frac{1}{m}\sum_{j=1}^m\varepsilon_jA_{jl}(\mathbf{A}_{j, -l}^\top\x_{-l})\Bigg| &\le c\sigma |\mathcal{L}|\sqrt{\frac{\log n}{m}}, \label{eq:tech6_3} \\ 
\Bigg|\sum_{l\in \mathcal{L}}\frac{1}{m}\sum_{j=1}^mA_{ji}A_{jl}(\mathbf{A}_{j}^\top\x)(\mathbf{A}_{j, -i}^\top\x^\star_{-i})\Bigg| &\le c |\mathcal{L}|\frac{\log n}{\sqrt{m}} \label{eq:tech6_4}.
\end{align}
\end{lemma}

\begin{lemma}\label{lemma:tech7}
Let $\x^\star\in\R^n$ be a $k$-sparse vector, and let $\S = \{i\in [n]: x^\star_i\neq 0\}$ be its support. Let $\{\mathbf{A}_j\}_{j=1}^m$  be a collection of i.i.d.\ $\gauss(0,\mathbf{I}_n)$ random vectors and $\{\varepsilon_j\}_{j=1}^m$ a collection of independent centered sub-exponential random variables with maximum sub-exponential norm $\sigma = \max_{j}\|\varepsilon_j\|_{\psi_1}$. Let $\mathcal{X} = \{\x\in \R^n: \x_{\S^c} = \mathbf{0}, \|\x\|_2\le 1\}$. There exist universal constants $c, c_s, c_p >0$ such that if $m\ge c_s\max\{k^2\log^2n, \log^5n\}$, then the following holds with probability at least $1 - c_pn^{10}$:
For all $i\in [n]$ and $\x\in\mathcal{X}$,
\begin{align}
\Bigg|\frac{1}{m}\sum_{j=1}^m A_{ji}(\A_{j,-i}^\top\x_{-i})^3\Bigg| &\le c\|\x\|_1 \frac{\log n}{\sqrt{m}}, \label{eq:tech7_1} \\
\Bigg|\sum_{l\neq i} x_l \sum_{r\neq l,i}x_r \frac{1}{m}\sum_{j=1}^mA_{ji}^2A_{jl}A_{jr}\Bigg| &\le c\|\x\|_1 \sqrt{\frac{\log n}{m}},\label{eq:tech7_2} \\
\Bigg|\sum_{l\neq i} x_l \sum_{r\neq l,i}x_r \frac{1}{m}\sum_{j=1}^m\varepsilon_jA_{jl}A_{jr}\Bigg| &\le c\sigma\|\x\|_1 \sqrt{\frac{\log n}{m}}, \label{eq:tech7_3} \\
\Bigg|\frac{1}{m}\sum_{j=1}^mA_{ji}(\mathbf{A}_{j,-i}^\top\x_{-i})^2(\mathbf{A}_{j,-i}^\top\x^\star_{-i})\Bigg| &\le c\|\x\|_1 \frac{\log n}{\sqrt{m}},\label{eq:tech7_4}. 
\end{align}
\end{lemma}

In the following, we prove Lemmas \ref{lemma:tech2}--\ref{lemma:tech7}.
%%%%%%%%%%%%%%%%%%%%Proofs%%%%%%%%%%%%%%%%%%%%
\begin{proof}[Proof of Lemma \ref{lemma:tech2}]
The proof follows that of Lemma A.5 of \cite{CLM16} closely. Since all three inequalities can be shown the same way, we only present the proof for (\ref{eq:tech2_3}). 
For any fixed $\lambda \ge 1$, define
\begin{equation*}
\|\mathbf{a}\|_{2\rightarrow 8, \lambda} = \max \{\|\mathbf{a}\x\|_8 : \|\x\|_2 = 1, \|\x\|_1\le \lambda\}.
\end{equation*}
We will show that $\|\A\|_{2\rightarrow 8, \lambda}\le (105m)^{\frac{1}{8}} + \sqrt{2\log (2k)}\lambda + t$ with probability $1-2\exp(-t^2/2)$. The result then follows by taking the union bound over $\lambda = 2^1,2^2,\dots,2^{\lceil \log_2\sqrt{k} \rceil}$ for each of the three inequalities (\ref{eq:tech2_1})--(\ref{eq:tech2_3}), since we have $1\le \|\x\|_1\le \sqrt{k}$ and the union bound is chosen such that there always exists a $\lambda$ with $\|\x\|_1 \le \lambda \le 2\|\x\|_1$.

Define $X_{\mathbf{u},\mathbf{v}} = \langle \A \mathbf{u}, \mathbf{v} \rangle$ on the set
\begin{equation*}
T_{\lambda} = \{(\mathbf{u},\mathbf{v}) : \mathbf{u}\in \R^k, \|\mathbf{u}\|_2= 1, \|\mathbf{u}\|_1\le \lambda, \mathbf{v}\in \R^m, \|\mathbf{v}\|_{8/7} \le 1\}.
\end{equation*}
Then, by H\"{o}lder's inequality, we have $\|\A\|_{2\rightarrow 8, \lambda} = \max_{(\mathbf{u},\mathbf{v})\in T_{\lambda}} X_{\mathbf{u}, \mathbf{v}}$.

Define $Y_{\mathbf{u}, \mathbf{v}} = \langle \mathbf{G}, \mathbf{u}\rangle + \langle \mathbf{H}, \mathbf{v} \rangle$, where $\mathbf{G}\in \R^k$ and $\mathbf{H}\in\R^m$ are independent standard normal random vectors.

We have, for any $(\mathbf{u},\mathbf{v}),(\mathbf{u'},\mathbf{v'})\in T_{\lambda}$, 
\begin{equation*}
\E[|X_{\mathbf{u},\mathbf{v}} - X_{\mathbf{u'},\mathbf{v'}}|^2] = \|\mathbf{v}\|_2^2 + \|\mathbf{v'}\|_2^2 - 2 \langle \mathbf{u}, \mathbf{u}'\rangle \langle \mathbf{v}, \mathbf{v}'\rangle,
\end{equation*}
and 
\begin{equation*}
\E[|Y_{\mathbf{u},\mathbf{v}} - Y_{\mathbf{u'},\mathbf{v'}}|^2] = 2 + \|\mathbf{v}\|_2^2 + \|\mathbf{v'}\|_2^2 - 2 \langle \mathbf{u}, \mathbf{u}'\rangle - 2\langle \mathbf{v}, \mathbf{v}'\rangle,
\end{equation*}
where we used $\|\mathbf{u}\|_2= \|\mathbf{u}'\|_2 = 1$. Therefore, we have
\begin{equation*}
\E[|Y_{\mathbf{u},\mathbf{v}} - Y_{\mathbf{u'},\mathbf{v'}}|^2] - \E[|X_{\mathbf{u},\mathbf{v}} - X_{\mathbf{u'},\mathbf{v'}}|^2] = 2(1 - \langle \mathbf{u}, \mathbf{u}'\rangle ) (1 - \langle \mathbf{v}, \mathbf{v}'\rangle ) \ge 0,
\end{equation*}
where we used $\|\mathbf{v}\|_2 \le \|\mathbf{v}\|_{8/7} \le 1$ and $\|\mathbf{v'}\|_2 \le \|\mathbf{v'}\|_{8/7} \le 1$. By Proposition 33 of \cite{V12}, this implies
\begin{equation*}
\E\biggl[\max_{(\mathbf{u},\mathbf{v})\in T_{\lambda}} X_{\mathbf{u}, \mathbf{v}}\biggr] \le \E\biggl[\max_{(\mathbf{u},\mathbf{v})\in T_{\lambda}} Y_{\mathbf{u}, \mathbf{v}}\biggr],
\end{equation*}
and hence
\begin{align*}
\E[\|\mathbf{A}\|_{2\rightarrow 8, \lambda}] &\le \E\biggl[\max_{(\mathbf{u},\mathbf{v})\in T_{\lambda}} Y_{\mathbf{u}, \mathbf{v}}\biggr] \\
&\le \E\biggl[\max_{(\mathbf{u},\mathbf{v})\in T_{\lambda}} \|\mathbf{G}\|_{\infty}\|\mathbf{u}\|_1 + \|\mathbf{H}\|_8\|\mathbf{v}\|_{8/7}\biggr] \\
&\le \sqrt{2\log (2k)}\lambda + (105m)^{\frac{1}{8}},
\end{align*}
where we used H\"{o}lder's inequality in the second line, and for the last inequality the bound
\begin{align*}
\exp(t\E[\|\mathbf{G}\|_{\infty}]) &\le \E[\exp(t\|\mathbf{G}\|_{\infty})]\le \sum_{i=1}^k \E[\exp(t|G_i|)]\le 2k\exp(t^2/2)
\end{align*}
with $t = \sqrt{2\log (2k)}$.

Finally, $\|\cdot \|_{2\rightarrow 8,\lambda}$ is a $1$-Lipschitz function: let $\mathbf{a}, \mathbf{b} \in \R^{m\times k}$ and, without loss of generality, $\|\mathbf{a}\|_{2\rightarrow 8, \lambda}\ge\|\mathbf{b}\|_{2\rightarrow 8, \lambda}$. Then, 
\begin{align*}
\|\mathbf{a}\|_{2\rightarrow 8, \lambda}-\|\mathbf{b}\|_{2\rightarrow 8, \lambda} & = \max_{\|\x\|_2=1, \|\x\|_1\le \lambda} \|\mathbf{a}\x\|_8 - \max_{\|\mathbf{y}\|_2=1, \|\mathbf{y}\|_1\le \lambda} \|\mathbf{b}\mathbf{y}\|_8 \\
&\le \max_{\|\x\|_2=1, \|\x\|_1\le \lambda} \|\mathbf{a}\x\|_8 - \|\mathbf{b}\x\|_8 \\
&\le \max_{\|\x\|_2=1, \|\x\|_1\le \lambda} \|(\mathbf{a} - \mathbf{b})\x\|_8 \\
&\le \max_{\|\x\|_2=1, \|\x\|_1\le \lambda} \|(\mathbf{a} - \mathbf{b})\x\|_2 \\
&\le \|\mathbf{a}-\mathbf{b}\|_F,
\end{align*}
where we write $\|\cdot\|_F$ for the Frobenius norm and used the fact that it is an upper bound to the $\ell_2$-operator norm. Hence, an application of Theorem \ref{thm:ref1} yields
\begin{equation*}
\P\Bigl[\|\A\|_{2\rightarrow 8, \lambda} < \sqrt{2\log (2k)}\lambda + (105m)^{\frac{1}{8}} + t \Bigr] \ge 1 - \exp(-t^2/2).
\end{equation*}
Taking the union bound over $\lambda = 2^1,\dots,2^{\lceil \log_2\sqrt{k} \rceil}$ completes the proof that (\ref{eq:tech2_3}) holds with probability $1 - \lceil \log_2 \sqrt{k}\rceil\exp(-t^2/2)$, since the union bound is chosen such that for any $\x\in\R^k$ with $\|\x\|_2=1$ there is a $\lambda$ satisfying $\|\x\|_1\le \lambda \le 2\|\x\|_1$. The inequalities (\ref{eq:tech2_1}) and (\ref{eq:tech2_2}) can be treated the same way, and another union bound gives the desired result.

Finally, convexity of $\mathcal{A}$ follows from an application of the Minkowski inequality. Let $\mathbf{a}, \mathbf{b} \in \R^{m\times k}$ satisfy (\ref{eq:tech2_3}), and $\alpha \in (0,1)$. Then, for any $\x\in\R$ with $\|\x\|_2=1$, 
\begin{align*}
\Biggl(\sum_{j=1}^m \big(\alpha \mathbf{a}^\top\x + (1-\alpha)\mathbf{b}^\top\x\big)^8\Biggr)^{\frac{1}{8}} &\le \alpha \Biggl(\sum_{j=1}^m (\mathbf{a}^\top\x)^8\Biggr)^{\frac{1}{8}} + (1-\alpha) \Biggl(\sum_{j=1}^m (\mathbf{b}^\top\x)^8\Biggr)^{\frac{1}{8}}\\
&\le (105m)^\frac{1}{8} + \sqrt{8\log (2k)} \|\x\|_1 + t.
\end{align*}
Since (\ref{eq:tech2_1}) and (\ref{eq:tech2_2}) can be treated the same way, this means that $\alpha\mathbf{a} + (1-\alpha)\mathbf{b}\in \mathcal{A}$.
\end{proof}

\begin{proof}[Proof of Lemma \ref{lemma:tech3}]
The proof of Lemma \ref{lemma:tech3} relies on Theorem \ref{thm:ref2} and the following truncation, which allows us to consider bounded random variables. We begin by showing the inequality (\ref{eq:tech3_1}).

Let $i,l,s\in [n]$.
Writing $B_j = \{\max\{|A_{ji}|, |A_{jl}|, |A_{js}|\}\le \sqrt{64\log n}\}$, we split the term $A_{ji}^2A_{jl}A_{js} = A_{ji}^2A_{jl}A_{js} \Eins(B_j) + A_{ji}^2A_{jl}A_{js}  (1 - \Eins(B_j)) = Y_j + Z_j$, where by $\Eins(\cdot)$ we denote the indicator function. We will show that each of the two terms concentrates around its mean. 

Since $\frac{1}{m}|Y_j|\le \frac{1}{m}64^2\log^2n$ is bounded, we can apply Theorem \ref{thm:ref2}. We can bound
\begin{equation*}
\sqrt{\sum_{j=1}^m \frac{1}{m^2}\E[Y_j^2]} = \sqrt{\sum_{j=1}^m\frac{1}{m^2}\E\Bigl[A_{ji}^4A_{jl}^2A_{js}^2\Eins(B_j)\Bigr]} \le \sqrt{\frac{105}{m}}.
\end{equation*}
Then, since $m\ge c_s\log^5 n$, we have, for $c>0$ sufficiently large,
\begin{equation*}
\P\Biggl[\Bigg|\frac{1}{m}\sum_{j=1}^mY_j - \E[Y_j]\Bigg| > \frac{c}{2}\sqrt{\frac{\log n}{m}}\Biggr] \le 2\exp\Biggl(- \frac{c^2\log n/(4m)}{2(\frac{105}{m} + 64^2c\frac{\log^{5/2}n}{6m^{3/2}})}\Biggr) \le 2n^{-16}.
\end{equation*}
For the second term $Z_j$, we can use the Chebyshev inequality to bound
\begin{align*}
\operatorname{Var}\Biggl(\frac{1}{m}\sum_{j=1}^mZ_j\Biggr) &\le \frac{1}{m}\E\bigl[A_{1i}^4A_{1l}^2A_{1s}^2\Eins(B_j^c)\bigr] \nonumber\\
&\le \frac{1}{m}\sqrt{\E\bigl[A_{1i}^{8}A_{1l}^4A_{1s}^4\bigr] \P\Bigl[\max\{|A_{1i}|, |A_{1l}|, |A_{1s}|\}>\sqrt{64\log n}\Bigr]} \nonumber\\
&\le \frac{c'}{mn^{16}},
\end{align*}
for an absolute constant $c'>0$, and hence, by the Chebyshev inequality,
\begin{equation*}
\P\Biggl[\Bigg|\frac{1}{m}\sum_{j=1}^mZ_j - \E[Z_j]\Bigg| > \frac{c}{2}\sqrt{\frac{\log n}{m}}\Biggr] \le \frac{\frac{c'}{mn^{16}}}{\frac{c^2\log n}{4m}} = \frac{4c'}{c^2\log n} n^{-16} \le n^{-16},
\end{equation*}
for $c \ge \sqrt{\frac{4c'}{\log n}}$.

This completes the proof that 
\begin{equation*}
\P\Biggl[\Bigg|\frac{1}{m}\sum_{j=1}^mA_{ji}^2A_{jl}A_{js} - \E\bigl[A_{1i}^2A_{1l}A_{1s}\bigr]\Bigg| > c\sqrt{\frac{\log n}{m}}\Biggr] \le 3n^{-16}.
\end{equation*}
Taking the union bound over all $i,l,s\in [n]$ shows that (\ref{eq:tech3_1}) holds with probability at least $1-3n^{-13}$.

The inequalities (\ref{eq:tech3_2})--(\ref{eq:tech3_6}) can be shown following the same steps. To show (\ref{eq:tech3_4}), we need to control higher order terms, which is the reason for the additional logarithmic factor. The bounds (\ref{eq:tech3_5}) and (\ref{eq:tech3_6}) can be shown the same way, as, for sub-exponential random variables, we have the standard tail bound $\P[|\varepsilon|> t]\le \exp(1-\tilde{c}t/\sigma)$ for all $t\ge 0$, where $\tilde{c}>0$ is an absolute constant, and the bound on the second moment $\E[\varepsilon^2] \le 4\sigma^2$. As the proof of each bound follows exactly the same steps, we omit the details to avoid repetition.
\end{proof}

\begin{proof}[Proof of Lemma \ref{lemma:tech4}]
Using the fact that $\A_j^\top\x^\star \sim \gauss(0, \|\x^\star\|_2^2)$ are independent random variables, this lemma can be shown following the same steps as in the proof of Lemma \ref{lemma:tech3}. The convexity of $\mathcal{A}$ follows, as in Lemma \ref{lemma:tech2}, from the Minkowski inequality. 
\end{proof}

\begin{proof}[Proof of Lemma \ref{lemma:tech5}]
We prove Lemma \ref{lemma:tech5} via induction over the size $|\mathcal{L}|$. For any $1\le \lambda\le k$, define
\begin{equation*}
\mathcal{A}_\lambda = \Biggl\{\{\mathbf{a}_j\}_{j=1}^m\in \R^{m\times k}: \frac{1}{m}\sum_{j=1}^m\Biggl(\sum_{l\in\mathcal{L}}a_{jl}\Biggr)^4\le c|\mathcal{L}|^2 \text{ for all } \mathcal{L}\subseteq[k] \text{ with } |\mathcal{L}| = \lambda \Biggr\}.
\end{equation*}
We will show that $\P[\{\mathbf{A}_j\}_{j=1}^m\in\mathcal{A}_\lambda]\ge 1 - c_1\lambda n^{-13}$ for some constant $c_1>0$, from which the result $\P[\{\mathbf{A}_j\}_{j=1}^m\in\mathcal{A}] \ge 1-c_pn^{-11}$ follows by taking the union bound over all possible sizes $\lambda$.
As in the proof of Lemma \ref{lemma:tech2}, convexity of $\mathcal{A}_\lambda$ follows from the Minkowski inequality. As the intersection of convex sets, $\mathcal{A}$ is also convex.

The base case $\lambda = 1$ follows from the bound (\ref{eq:tech3_2}) of Lemma \ref{lemma:tech3}. For the induction step, assume that we have already shown $\P[\{\mathbf{A}_j\}_{j=1}^m\in\mathcal{A}_\lambda]\ge 1 - c_1\lambda n^{-13}$ for some $1\le\lambda < k$.
Let $r\in [k]$ and $\mathcal{L}\subseteq [k]\backslash\{r\}$ be a subset of coordinates with $|\mathcal{L}|=\lambda$. We have
\begin{align}\label{eq:indstep}
\frac{1}{m}\sum_{j=1}^m\Biggl(A_{jr} + \sum_{l\in \mathcal{L}} A_{jl}\Biggr)^4 &= \frac{1}{m}\sum_{j=1}^mA_{jr}^4 + \frac{4}{m}\sum_{j=1}^m\Biggl(\sum_{l\in \mathcal{L}} A_{jl}\Biggr) A_{jr}^3 + \frac{6}{m}\sum_{j=1}^m\Biggl(\sum_{l\in \mathcal{L}} A_{jl}\Biggr)^2 A_{jr}^2 
\nonumber\\*
&\quad + \frac{4}{m}\sum_{j=1}^m\Biggl(\sum_{l\in \mathcal{L}} A_{jl}\Biggr)^3 A_{jr} + \frac{1}{m}\sum_{j=1}^m\Biggl(\sum_{l\in \mathcal{L}} A_{jl}\Biggr)^4 
\end{align}
We need to show that, with probability $1-c_1(\lambda+1)n^{-13}$, the sum (\ref{eq:indstep}) is bounded by $c(|\mathcal{L}|+1)^2$ for all $r\in [k]$ and $\mathcal{L}\subseteq [k]\backslash\{r\}$ with $|\mathcal{L}|=\lambda$.

By the induction hypothesis, we can bound the last term with probability $1-c_1\lambda n^{-13}$,
\begin{equation*}
\frac{1}{m}\sum_{j=1}^m\Biggl(\sum_{l\in \mathcal{L}} A_{jl}\Biggr)^4 \le c|\mathcal{L}|^2.
\end{equation*}
The first, second and third term in (\ref{eq:indstep}) can be bounded using Lemma \ref{lemma:tech3}. The following holds with probability at least $1-c_2n^{-13}$, where $c_2$ is the universal constant from Lemma \ref{lemma:tech3}. Using the assumption $m\ge c_sk^2\log^2 n$, we can bound the first term by
\begin{equation*}
\frac{1}{m}\sum_{j=1}^mA_{jr}^4 \le 4,
\end{equation*}
the second term by
\begin{equation*}
\frac{4}{m}\sum_{j=1}^m\Biggl(\sum_{l\in \mathcal{L}} A_{jl}\Biggr) A_{jr}^3 = 4\sum_{l\in\mathcal{L}} \frac{1}{m} \sum_{j=1}^mA_{jl}A_{jr}^3 \le 4 \biggl( 4 + \frac{|\mathcal{L}| - 1}{k} \biggr) \le 20,
\end{equation*}
and the third term in (\ref{eq:indstep}) can be bounded, using the Cauchy-Schwarz inequality, by
\begin{align*}
\frac{6}{m}\sum_{j=1}^m\Biggl(\sum_{l\in \mathcal{L}} A_{jl}\Biggr)^2 A_{jr}^2 &\le 6\sqrt{\frac{1}{m}\sum_{j=1}^m\Biggl(\sum_{l\in \mathcal{L}} A_{jl}\Biggr)^4}\sqrt{\frac{1}{m}\sum_{j=1}^mA_{jr}^4} \\
&\le 6 \sqrt{c|\mathcal{L}|^2} \sqrt{4} \\
&\le 0.4 c |\mathcal{L}|,
\end{align*}
for $c \ge 900$, where we used the induction hypothesis in the second line. 

To bound the fourth term in (\ref{eq:indstep}), fix any $r\in [k]$ and subset $\mathcal{L}\subseteq [k]\backslash\{r\}$ with cardinality $|\mathcal{L}| = \lambda$, and consider the function
\begin{equation*}
h_{r, \mathcal{L}}(\mathbf{a}) = \frac{4}{m}\sum_{j=1}^ma_{jr}\Biggl(\sum_{l\in\mathcal{L}}a_{jl}\Biggr)^3.
\end{equation*}
Then, we need to bound the probability of the event
\begin{align*}
B_r = \{|h_{r, \mathcal{L}}(\A)| \le 1.6c|\mathcal{L}| \text{ for all } \mathcal{L}\subseteq [k]\backslash \{r\} \text{ with } |\mathcal{L}|=\lambda\}
\end{align*}
 for all $r\in[k]$. To this end, we condition on $\A_{\cdot r} = \mathbf{a}_{\cdot r}$ using the formula  
\begin{equation*}
\P[B] = \int \P[B|\A_{\cdot r} = \mathbf{a}_{\cdot r}] \mu(\mathbf{a}_{\cdot r})d\mathbf{a}_{\cdot r},
\end{equation*}
which holds for any event $B$, where we write $\mu$ for the standard normal density. 

To show that the conditional probability $\P[B_r|\A_{\cdot r} = \mathbf{a}_{\cdot r}]$ is close to one, we begin by showing that, for any fixed $\mathcal{L}$, $h_{r, \mathcal{L}}(\mathbf{A})|_{\mathbf{A}_{\cdot r} = \mathbf{a}_{\cdot r}}$ (making explicit the fact that we are conditioning on $\mathbf{A}_{\cdot r} = \mathbf{a}_{\cdot r}$) concentrates around its expectation $\E[h_{r, \mathcal{L}}(\A)|\mathbf{A}_{\cdot r} = \mathbf{a}_{\cdot r}] = 0$. 
For the sake of brevity, we will omit explicitly writing the condition $\mathbf{A}_{\cdot r} = \mathbf{a}_{\cdot r}$ in what follows.

\textbf{Step 1: Bound the fourth term in (\ref{eq:indstep}) conditioned on $r$, $\A_{\cdot r}$ and $\mathcal{L}$}\\
Fix an $r\in[k]$, a vector $\mathbf{a}_{\cdot r}$ satisfying $\max_j |a_{jr}|\le 6\sqrt{\log n}$ and a subset $\mathcal{L}\subseteq [k]\backslash \{r\}$ with cardinality $|\mathcal{L}| = \lambda$. 
The idea is to apply concentration of Lipschitz functions of Gaussian random variables to show that the fourth term in (\ref{eq:indstep}), $h_{r, \mathcal{L}}(\A)$, is close to its expectation $\E[h_{r, \mathcal{L}}(\A)] = 0$. However, as $h_{r, \mathcal{L}}$ is not globally Lipschitz continuous, we cannot directly apply Theorem \ref{thm:ref1}. Similar to the proof of Theorem 3 in \cite{F18}, we will first restrict $h_{r, \mathcal{L}}$ to a high probability event on which $h_{r, \mathcal{L}}$ is Lipschitz continuous. Then, we extend this restricted function to a function $\tilde{h}_{r, \mathcal{L}}$ on the entire space in a way such that $\tilde{h}_{r, \mathcal{L}}$ is globally Lipschitz continuous, and apply Theorem \ref{thm:ref1} to $\tilde{h}_{r, \mathcal{L}}$. This also provides a high probability bound for $h_{r, \mathcal{L}}(\A)$, since, by construction, $\tilde{h}_{r, \mathcal{L}}(\A)=h_{r, \mathcal{L}}(\A)$ with high probability. 

\textbf{Step 1, part (a): Bound the Lipschitz constant of $h_{r, \mathcal{L}}$ restricted to $\mathcal{A}_\lambda$}\\ 
Restricted to $\mathcal{A}_\lambda$ and conditioned on $\mathbf{A}_{\cdot r} = \mathbf{a}_{\cdot r}$, we can bound the Lipschitz constant of $h_{r, \mathcal{L}}$ by the norm of its gradient. The norm of the gradient is an upper bound for the Lipschitz constant by the mean-value theorem, since $\mathcal{A}_\lambda$ is a convex set. For any $l\in\mathcal{L}$, we have
\begin{equation*}
\frac{\partial}{\partial a_{jl}}h_{r, \mathcal{L}}(\mathbf{a}) = \frac{12}{m}a_{jr}\Biggl(\sum_{l'\in\mathcal{L}}a_{jl'}\Biggr)^2.
\end{equation*}
Hence, we can bound
\begin{align*}
\|\nabla h_{r, \mathcal{L}}(\mathbf{a})\|_2^2 &= \sum_{l\in \mathcal{L}}\sum_{j=1}^m \biggl(\frac{\partial}{\partial a_{jl}}h_{r, \mathcal{L}}(\mathbf{a})\biggr)^2 \\
&\le \frac{144|\mathcal{L}|\max_ja_{jr}^2}{m} \frac{1}{m}\sum_{j=1}^m\Biggl(\sum_{l\in\mathcal{L}}a_{jl}\Biggr)^4 \\
&\le \frac{5184|\mathcal{L}|\log n}{m} c |\mathcal{L}|^2
\end{align*}
on $\mathcal{A}_\lambda$, where we used the assumption $\max_j a_{jr}^2\le 36\log n$ and the induction hypothesis.

\textbf{Step 1, part (b): Construct a globally Lipschitz continuous extension of $h_{r, \mathcal{L}}$}\\
Consider the following Lipschitz extension of the function $h_{r, \mathcal{L}}$:
\begin{equation*}
\tilde{h}_{r, \mathcal{L}}(\mathbf{a}) = \inf_{\mathbf{a'}\in \mathcal{A}_\lambda} \bigl((h_{r, \mathcal{L}}(\mathbf{a'}) + \operatorname{Lip}(h_{r, \mathcal{L}}) \|\mathbf{a} - \mathbf{a'} \|_2\bigr),
\end{equation*}
where we write $\operatorname{Lip}(h_{r, \mathcal{L}}) = \sqrt{\frac{5184c|\mathcal{L}|^3\log n}{m}}$. By definition, we have $\tilde{h}_{r, \mathcal{L}} = h_{r, \mathcal{L}}$ on $\mathcal{A}_\lambda$, and it follows from an application of the triangle inequality that $\tilde{h}_{r, \mathcal{L}}$ is globally Lipschitz continuous with Lipschitz constant $\operatorname{Lip}(h_{r, \mathcal{L}})$ (see e.g.\ Theorem 7.2 of \cite{M95}). 

We will show that $\tilde{h}_{r, \mathcal{L}}$ concentrates around its mean, which can potentially differ from the mean of $h_{r, \mathcal{L}}$.
Since $h_{r, \mathcal{L}}$ and $\tilde{h}_{r, \mathcal{L}}$ differ only on $\mathcal{A}_\lambda^c$, which has probability less than $c_1\lambda n^{-13}$, we can bound, using the Cauchy-Schwarz inequality,
\begin{equation*}
\E[|h_{r, \mathcal{L}}(\A)|\Eins_{\mathcal{A}_\lambda^c}(\A)] \le \sqrt{\E[h_{r, \mathcal{L}}(\A)^2]} \sqrt{\E[\Eins_{\mathcal{A}_\lambda^c}(\A)]} \le \frac{c'}{n^5},
\end{equation*}
where we used that $\E[h_{r, \mathcal{L}}(\A)^2]= 15|\mathcal{L}|^3/m$. Using $\tilde{h}_{r, \mathcal{L}}(\mathbf{a}) \le \operatorname{Lip}(h_{r, \mathcal{L}}) \|\mathbf{a}\|_2$, we have
\begin{equation*}
\E[|\tilde{h}_{r, \mathcal{L}}(\A)|\Eins_{\mathcal{A}_\lambda^c}(\A)] \le \operatorname{Lip}(h_{r, \mathcal{L}})\sqrt{\E[\|\A\|_2^2]} \sqrt{\E[\Eins_{\mathcal{A}_\lambda^c}(\A)]} \le \frac{c''}{n^4}.
\end{equation*}
All in all, this shows that
\begin{equation*}
\big|\E[h_{r, \mathcal{L}}(\A)] - \E[\tilde{h}_{r, \mathcal{L}}(\A)]\big| \le \frac{c_3}{n^4}
\end{equation*}
for a constant $c_3 > 0$.
Hence, we can apply Theorem \ref{thm:ref1} to obtain
\begin{align*}
\P\bigl[|\tilde{h}_{r, \mathcal{L}}(\A)| > 1.6c|\mathcal{L}| \;\big|\; \A_{\cdot r} = \mathbf{a}_{\cdot r} \bigr] &\le 2\exp\Biggl(- \frac{(1.6c|\mathcal{L}| - c_3/n^4)^2}{2\frac{5184c|\mathcal{L}|^3\log n}{m}}\Biggr) \\
&\le 2 \exp\biggl(- c_4 \frac{m}{|\mathcal{L}|\log n}\biggr),
\end{align*} 
for a constant $c_4 \le \frac{c}{4050} - \frac{c_3}{3240n^4}$.

\textbf{Step 2: Unravel the conditions: take union bounds and integrate over $\mathbf{a_{\cdot r}}$}\\
Let
\begin{equation*}
\widetilde{B}_r = \Bigl\{|\tilde{h}_{r, \mathcal{L}}(\A)| \le 1.6c|\mathcal{L}| \text{ for all } \mathcal{L}\subseteq [k]\backslash \{r\} \text{ with } |\mathcal{L}|=\lambda\Bigr\}.
\end{equation*}
Since we assume $m\ge c_sk^2\log^2 n$, we can take the union bound over all possible subsets $\mathcal{L}\subseteq [k]\backslash \{r\}$ with $|\mathcal{L}| = \lambda$ to obtain, using the upper bound ${\binom{k}{\lambda}} \le (\frac{ek}{\lambda})^\lambda$,
\begin{align*}
\P\bigl[\widetilde{B}_r \;\big|\; \A_{\cdot r} = \mathbf{a}_{\cdot r}\bigr] &\ge 1 - 2\exp\biggl(- c_4 \frac{m}{|\mathcal{L}|\log n} + |\mathcal{L}|\log \frac{ek}{|\mathcal{L}|}\biggr) \\
&\ge 1 - 2\exp(-(c_sc_4-1)k\log n).
\end{align*}
Next, we integrate over all $\mathbf{a}_{\cdot r}$ satisfying $\max_j |a_{jr}|\le 6\sqrt{\log n}$:
\begin{align*}
\P\bigl[\widetilde{B}_r \bigr] &\ge \int_{\{\max_j |a_{jr}|\le 6\sqrt{\log n}\}}\P\bigl[B_r \;\big|\; \A_{\cdot r} = \mathbf{a}_{\cdot r}\bigr]\mu(\mathbf{a}_{\cdot r})d\mathbf{a}_{\cdot r} \\
&\ge \bigl(1 - 2\exp(-(c_sc_4-1)k\log n)\bigr) \P\Bigl[\max_j |a_{jr}| \le 6\sqrt{\log n}\Bigr] \\
&\ge 1 - (m + 2)n^{-18},
\end{align*}
where we write $\mu$ for the standard normal density and the last line follows from standard Gaussian tail bounds and a union bound, provided that $(c_sc_4-1)k\ge 18$. 

By construction, we have $h_{r, \mathcal{L}}(\mathbf{a}) = \tilde{h}_{r, \mathcal{L}}(\mathbf{a})$ for all $r\in [k]$ and subsets $\mathcal{L} \subseteq [k]\backslash\{r\}$ with cardinality $|\mathcal{L}| = \lambda$ on the set $\mathcal{A}_\lambda\cap\{\max_{j,r}|a_{jr}|\le 6\sqrt{\log n}\}$. Hence, a union bound gives a lower bound on the probability that the fourth term in (\ref{eq:indstep}) is bounded by $1.6c|\mathcal{L}|$ for all $r\in [k]$ and $\mathcal{L}\subseteq [k]\backslash\{r\}$ with $|\mathcal{L}| = \lambda$:
\begin{equation*}
\P\Biggl[\bigcap\limits_{r=1}^k B_r\Biggr] \ge \P\Biggl[\bigcap\limits_{r=1}^k \widetilde{B}_r \cap \mathcal{A}_\lambda\cap\Bigl\{\max_{j,r}|a_{jr}|\le 6\sqrt{\log n}\Bigr\}\Biggr] \ge 1 - (m + 2)n^{-17} - c_1\lambda n^{-13}.
\end{equation*}
Putting everything together, this completes the induction step (\ref{eq:indstep}),
\begin{equation*}
\frac{1}{m}\sum_{j=1}^m\Biggl(\sum_{l\in \mathcal{L}} A_{jl} + A_{jr}\Biggr)^4 \le c(|\mathcal{L}|+1)^2 \quad \text{for all } r\in[k], \mathcal{L}\subseteq [k]\backslash \{r\} \text{ with } |\mathcal{L}|=\lambda ,
\end{equation*}
holds with probability at least $1 - c_1(\lambda+1)n^{-13}$, provided that $c_1 \ge c_2 + (m + 2)n^{-4}$, which is satisfied for a universal constant $c_1$ if $m \lesssim n^4$. The case $m\gtrsim n^4$ is simpler and can be shown following the same steps, writing probabilities in terms of $m$ instead of $n$.
\end{proof}

\begin{proof}[Proof of Lemma \ref{lemma:tech6}] We will show that each of the inequalities (\ref{eq:tech6_1})--(\ref{eq:tech6_4}) is satisfied with probability at least $1 - \frac{c_p}{4}n^{-10}$. The lemma then follows by taking a union bound. Throughout this proof, we will assume $m\le n^{3/2}$. The other case $m > n^{3/2}$ is simpler and can be shown following the same steps, writing probabilities in terms of $m$ instead of $n$. 

\textbf{Proof that (\ref{eq:tech6_1}) is satisfied with high probability}\\
In order to show that (\ref{eq:tech6_1}) holds with high probability, we proceed as in the induction step of the proof of Lemma \ref{lemma:tech5}: for any index $i\in [n]$, any subset $\mathcal{K}\subseteq \S\backslash \{i\}$, any subset $\mathcal{L}\subseteq \mathcal{K}$ with $|\mathcal{L}| \ge \frac{1}{2}|\mathcal{K}|$, and any vector $\x\in N^{\mathcal{K}}_{\epsilon}$, we consider the function
\begin{equation*}
h_{i, \mathcal{K},\mathcal{L},\x}(\mathbf{a}) = \sum_{l\in \mathcal{L}} \frac{1}{m}\sum_{j=1}^m a_{ji}a_{jl}(\mathbf{a}_{j}^\top\x)^2.
\end{equation*} 
Then, we define $B_i$ as the following event:
\begin{align*}\label{eq:event}
B_i = \Big\{|h_{i,\mathcal{K},\mathcal{L},\x}(\mathbf{A})|\le c|\mathcal{L}|\frac{\log n}{\sqrt{m}} \quad &\text{for all subsets } \mathcal{K}\subseteq \S\backslash \{i\}, \mathcal{L}\subseteq \mathcal{K} \text{ with } \nonumber\\
&|\mathcal{L}|\ge \frac{1}{2}|\mathcal{K}| \text{ and } \x\in N^{\mathcal{K}}_{\epsilon} \Big\}.
\end{align*}
To bound the probability of $B_i$, we condition on $\A_{\cdot i} = \mathbf{a}_{\cdot i}$ using the formula
\begin{equation*}
\P[B] = \int \P[B|\A_{\cdot i}=\mathbf{a}_{\cdot i}]\mu(\mathbf{a}_{\cdot i})d\mathbf{a}_{\cdot i},
\end{equation*}
which holds for any event $B$, where we write $\mu$ for the standard normal density.

To show that the conditional probability $\P[B_i|\A_{\cdot i} = \mathbf{a}_{\cdot i}]$ is close to one, we begin by showing that, for any fixed $\mathcal{K}$, $\mathcal{L}$ and $\x$ as described above, $h_{i, \mathcal{K},\mathcal{L},\x,}(\mathbf{A})|_{\mathbf{A}_{\cdot i} = \mathbf{a}_{\cdot i}}$ (making explicit the fact that we are conditioning on $\mathbf{A}_{\cdot i} = \mathbf{a}_{\cdot i}$) concentrates around its expectation $\E[h_{i,\mathcal{K},\mathcal{L}, \x}(\A)|\mathbf{A}_{\cdot i} = \mathbf{a}_{\cdot i}] = 0$. To simplify notation, we will omit the subscripts and write $h(\A)$ in what follows.

\textbf{Step 1: Bound $h(\A)$ conditioned on $i, \mathcal{K},\mathcal{L},\x$ and $\A_{\cdot i}$}\\
Let $\mathcal{A}_1$ be as in Lemma \ref{lemma:tech5}, and $\mathcal{A}_2$ as in Lemma \ref{lemma:tech1} with $t = 5\sqrt{\log n}$ (more precisely, we have $\mathcal{A}_1,\mathcal{A}_2\subset \R^{m\times n}$, and we require the projections onto $\R^{m\times \S}$ to be as in the respective lemmas). Then, by these two Lemmas, we have $\P[\mathcal{A}_1\cap\mathcal{A}_2] \ge 1 - c_2n^{-11}$ for a universal constant $c_2>0$, and as the intersection of two convex sets, $\mathcal{A}_1\cap\mathcal{A}_2$ is also convex.

Fix an $i\in [n]$, a subset $\mathcal{K}\subseteq \S\backslash\{i\}$, a subset $\mathcal{L}\subseteq\mathcal{K}$ with $|\mathcal{L}|\ge \frac{1}{2}|\mathcal{K}|$ and a vector $\x\in N^{\mathcal{K}}_{\epsilon}$.
We begin by conditioning on $\A_{\cdot i} = \mathbf{a}_{\cdot i}$ for a vector $\mathbf{a}_{\cdot i}$ satisfying $\max_j|a_{ji}|\le 5\sqrt{\log n}$.
As in the proof of Lemma \ref{lemma:tech5}, the idea is use Theorem \ref{thm:ref1} to show that $h(\A)$ is close to its expectation. However, $h(\A)$ is not globally Lipschitz continuous. We will show that the function $h$ restricted to $\mathcal{A}_1\cap \mathcal{A}_2$ is Lipschitz continuous. Then, we extend this restricted function to a globally Lipschitz continuous function $\tilde{h}$ on the entire space, and apply Theorem \ref{thm:ref1} to $\tilde{h}$. By construction, we have $h(\A) = \tilde{h}(\A)$ with high probability, which therefore also yields a high probability bound for $h(\A)$.

\textbf{Step 1, part (a): Bound the Lipschitz constant of $h$ restricted to $\mathcal{A}_1\cap\mathcal{A}_2$}\\
Restricted to $\mathcal{A}_1\cap\mathcal{A}_2$, we can bound the Lipschitz constant of $h$ by the norm of its gradient. The norm of the gradient is an upper bound for the Lipschitz constant by the mean-value theorem, since $\mathcal{A}_1\cap\mathcal{A}_2$ is a convex set. Note that $h$ only depends on $a_{jr}$ for $r\in\mathcal{K}$, since we have $x_r = 0$ for $r\notin \mathcal{K}$. We can compute
\begin{equation*}
\frac{\partial}{\partial a_{jr}} h(\mathbf{a})= \begin{cases} 
\frac{2}{m} a_{ji}x_r(\mathbf{a}_{j}^\top\x) \sum_{l\in\mathcal{L}}a_{jl}& r\in \mathcal{K}\backslash\mathcal{L}\\
\frac{2}{m} a_{ji}x_r(\mathbf{a}_{j}^\top\x) \sum_{l\in\mathcal{L}}a_{jl} +  \frac{1}{m}a_{ji}(\mathbf{a}_{j}^\top\x)^2 \quad & r\in \mathcal{L} \end{cases}
\end{equation*}
Using the inequality $(a+b)^2 \le 2(a^2+b^2)$, we can bound
\begin{align}\label{eq:lipconst}
\|\nabla h(\mathbf{a})\|_2^2 = \sum_{r\in\mathcal{K}}\sum_{j=1}^m\biggl(\frac{\partial}{\partial a_{jr}} h(\mathbf{a})\biggr)^2 \le & 8\sum_{r\in \mathcal{K}}\sum_{j=1}^m \frac{1}{m^2} x_r^2a_{ji}^2\Biggl(\sum_{l\in \mathcal{L}}a_{jl}\Biggr)^2(\mathbf{a}_{j}^\top\x)^2 \nonumber\\
& + 2\sum_{r\in \mathcal{L}} \sum_{j=1}^m\frac{1}{m^2}a_{ji}^2(\mathbf{a}_{j}^\top\x)^4.
\end{align}
We bound the two sums separately. For the first sum, we have
\begin{align*}
&\quad \sum_{r\in \mathcal{K}}\sum_{j=1}^m \frac{1}{m^2}x_r^2a_{ji}^2\Biggl(\sum_{l\in \mathcal{L}}a_{jl}\Biggr)^2(\mathbf{a}_{j}^\top\x)^2 \\
&= \frac{1}{m} \sum_{r\in \mathcal{K}}x_r^2\sum_{j=1}^m \frac{1}{m}a_{ji}^2\Biggl(\sum_{l\in \mathcal{L}}a_{jl}\Biggr)^2(\mathbf{a}_{j}^\top\x)^2 \\
&\le \Bigl(\max_j a_{ji}^2\Bigr)\frac{1}{m} \sum_{r\in \mathcal{K}}x_r^2\sum_{j=1}^m \frac{1}{m}\Biggl(\sum_{l\in \mathcal{L}}a_{jl}\Biggr)^2(\mathbf{a}_{j}^\top\x)^2 \\
&\le \frac{25\log n}{m} \sum_{r\in \mathcal{K}}x_r^2 \sqrt{\frac{1}{m}\sum_{j=1}^m \Biggl(\sum_{l\in \mathcal{L}}a_{jl}\Biggr)^4 }\sqrt{\frac{1}{m}\sum_{j=1}^m (\mathbf{a}_{j}^\top\x)^4},
\end{align*}
where we used H\"{o}lder's inequality in the second and the Cauchy-Schwarz inequality in the last line. 
Since $\mathbf{a}\in \mathcal{A}_1$, we can apply Lemma \ref{lemma:tech5} to bound
\begin{equation*}
\sqrt{\frac{1}{m}\sum_{j=1}^m \Biggl(\sum_{l\in \mathcal{L}}a_{jl}\Biggr)^4} \le c_3|\mathcal{L}|,
\end{equation*}
where $c_3$ is the universal constant from Lemma \ref{lemma:tech5}. Since also $\mathbf{a}\in \mathcal{A}_2$, Lemma \ref{lemma:tech1} yields  
\begin{equation*}
\sqrt{\frac{1}{m}\sum_{j=1}^m(\mathbf{a}_{j}^\top\x)^4} \le \sqrt{\frac{1}{m}\|\mathbf{a}_{\cdot \mathcal{K}}\|_{2\rightarrow 4}^4} \le 11,
\end{equation*}
as $\|\x\|_2 = 1$.
Putting this together, we can bound the first sum in (\ref{eq:lipconst}),
\begin{equation*}
\sum_{r\in \mathcal{K}}\sum_{j=1}^m \frac{1}{m^2} x_r^2a_{ji}^2\Biggl(\sum_{l\in \mathcal{L}}a_{jl}\Biggr)^2(\mathbf{a}_{j}^\top\x)^2 \le \frac{5\log n}{m}\sum_{r\in \mathcal{K}}11c_3|\mathcal{L}|x_r^2 \le 55c_3\frac{|\mathcal{L}|\log n}{m},
\end{equation*}
where we again used $\|\x\|_2= 1$.
For the second sum in (\ref{eq:lipconst}), we can write
\begin{equation*}
\sum_{r\in \mathcal{L}} \sum_{j=1}^m\frac{1}{m^2}a_{ji}^2(\mathbf{a}_{j}^\top\x)^4 \le \frac{5\log n}{m} |\mathcal{L}| \frac{1}{m}\sum_{j=1}^m(\mathbf{a}_{j}^\top\x)^4 \le 605\frac{|\mathcal{L}|\log n}{m},
\end{equation*}
where we again used Lemma \ref{lemma:tech1} to bound $\frac{1}{m}\sum_{j=1}^m(\mathbf{a}_{j}^\top\x)^4\le 121$.
All in all, we have shown 
\begin{equation*}
\|\nabla h(\mathbf{a})\|_2 = \Biggl(\sum_{r\in \mathcal{K}}\sum_{j=1}^m\biggl(\frac{\partial}{\partial a_{jr}} h(\mathbf{a})\biggr)^2\Biggr)^{\frac{1}{2}} \le \sqrt{c_4\frac{|\mathcal{L}|\log n}{m}}
\end{equation*}
for $\mathbf{a}\in \mathcal{A}_1\cap \mathcal{A}_2$ and $c_4 = 440c_3 + 1210$.

\textbf{Step 1, part (b): Construct a globally Lipschitz continuous extension of $h$}\\
Consider the following Lipschitz extension of the function $h$:
\begin{equation*}
\tilde{h}(\mathbf{a}) = \inf_{\mathbf{a}'\in \mathcal{A}_1\cap \mathcal{A}_2} \big(h(\mathbf{a'}) + \operatorname{Lip}(h) \|\mathbf{a}-\mathbf{a'}\|_2\big),
\end{equation*}
where we write $\operatorname{Lip}(h) = \sqrt{c_4\frac{|\mathcal{L}|\log n}{m}}$. By definition, we have $\tilde{h}=h$ on $\mathcal{A}_1\cap\mathcal{A}_2$, and it follows from an application of the triangle inequality that $\tilde{h}$ is globally Lipschitz continuous with Lipschitz constant $\operatorname{Lip}(h)$ (see e.g.\ Theorem 7.2 of \cite{M95}). 

We will show that $\tilde{h}$ concentrates around its mean, which can potentially differ from the mean of $h$.
Since $h$ and $\tilde{h}$ differ only on $(\mathcal{A}_1\cap\mathcal{A}_2)^c$, which has probability at most $c_2n^{-11}$, we can bound, using the Cauchy-Schwarz inequality,
\begin{equation*}
\E[|h(\A)|\Eins_{(\mathcal{A}_1\cap\mathcal{A}_2)^c}(\A)] \le \sqrt{\E[h(\A)^2]} \sqrt{\E[\Eins_{(\mathcal{A}_1\cap\mathcal{A}_2)^c}(\A)]} \le \frac{c'}{n^5},
\end{equation*}
where we used that $\E[h(\A)^2] \le \sqrt{315}k^2/m$. Using $\tilde{h}(\mathbf{a}) \le \operatorname{Lip}(h) \|\mathbf{a}\|_2$, we can bound
\begin{equation*}
\E[|\tilde{h}(\A)|\Eins_{(\mathcal{A}_1\cap\mathcal{A}_2)^c}(\A)] \le \operatorname{Lip}(h)\sqrt{\E[\|\A\|_2^2]} \sqrt{\E[\Eins_{(\mathcal{A}_1\cap\mathcal{A}_2)^c}(\A)]} \le \frac{c''}{n^3}.
\end{equation*}
All in all, this shows that
\begin{equation*}
\big|\E[h(\A)] - \E[\tilde{h}(\A)]\big| \le \frac{c_5}{n^3}.
\end{equation*}
for a constant $c_5 > 0$. 
Hence, we can apply Theorem \ref{thm:ref1} to obtain
\begin{align*}
\P\biggl[|\tilde{h}(\A)| > c\frac{|\mathcal{L}|\log n}{\sqrt{m}} \;\Big|\;\A_{\cdot i} = \mathbf{a}_{\cdot i} \biggr] &\le 2\exp\Biggl(- \frac{(c\frac{|\mathcal{L}|}{\sqrt{m}}\log n - c_5/n^3)^2}{2c_4\frac{|\mathcal{L}|\log n}{m}}\Biggr) \\
&\le 2 \exp(- c_6 |\mathcal{L}|\log n),
\end{align*} 
for a universal constant $c_6 \le \frac{c}{2c_4}- \frac{cc_5}{c_4}\frac{\sqrt{m}}{n^3}$, where we used the assumption $m\le n^{3/2}$.

\textbf{Step 2: Unravel the conditions: take union bounds and integrate over $\mathbf{a}_{\cdot i}$}\\
Let
\begin{align*}
\widetilde{B}_{i, \lambda} = \Big\{|\tilde{h}_{i,\mathcal{K},\mathcal{L},\x}(\A)|\le c\frac{|\mathcal{L}|\log n}{\sqrt{m}} \quad &\text{for all subsets } \mathcal{K}\subseteq \S\backslash \{i\} \text{ with } |\mathcal{K}| = \lambda, \\
&\mathcal{L}\subseteq \mathcal{K} \text{ with } |\mathcal{L}|\ge \frac{1}{2}|\mathcal{K}| \text{ and } \x\in N^{\mathcal{K}}_{\epsilon} \Big\},
\end{align*}
and let $\widetilde{B}_i$ be the same event without the restriction $|\mathcal{K}| = \lambda$. The cardinality of the $\epsilon$-net can be bounded by $|N^{\mathcal{K}}_{\epsilon}|\le (3/\epsilon)^\lambda\le (3m/c_1)^\lambda$.
Taking union bounds over all $\x\in N^{\mathcal{K}}_{\epsilon}$, $\mathcal{L}\subseteq \mathcal{K}$ with $|\mathcal{L}|\ge \frac{1}{2}|\mathcal{K}|$ and $\mathcal{K}\subseteq \S\backslash \{i\}$ with fixed cardinality $|\mathcal{K}|=\lambda$, we obtain (using the upper bound ${\binom{k}{\lambda}} \le (\frac{ek}{\lambda})^\lambda$)
\begin{align*}
\P\bigl[\widetilde{B}_{i,\lambda} \;\big||\; \A_{\cdot i} = \mathbf{a}_{\cdot i}\bigr] &\ge 1 - 2\exp\biggl(- \frac{c_6}{2}\lambda\log n + \lambda\log \frac{3m}{c_1} + \lambda\log 2 + \lambda\log\frac{ek}{\lambda}\biggr) \\
&\ge 1 - 2\exp(-c_7\lambda\log n)
\end{align*}
for a constant $c_7>0$ if $c_6 \ge 5 + 2\frac{\log (6e/c_1)}{\log n}$, where we used the assumption $m\le n^{3/2}$.  
Taking the union bound over all possible choices for the cardinalities $\lambda=1,...,k$ gives
\begin{align*}
\P\bigl[\widetilde{B}_i \;|\; \A_{\cdot i} = \mathbf{a}_{\cdot i}\bigr] &\ge 1 - \frac{2}{1 - 2\exp(-c_7 \log n)}\exp(-c_7\log n) \ge 1 - 2.1 \exp(-c_7\log n),
\end{align*}
provided that $n^{-c_7} \le 0.1/4.2$.
Next, we integrate over all $\mathbf{a}_{\cdot i}$ satisfying $\max_j |a_{ji}|\le 5\sqrt{\log n}$:
\begin{align*}
\P\bigl[\widetilde{B}_i\bigr] &\ge \int_{\{\max_j |a_{ji}|\le 5\sqrt{\log n}\}}\P\bigl[\widetilde{B}_i \;\big|\; \A_{\cdot i} = \mathbf{a}_{\cdot i} \bigr]\mu(\mathbf{a}_{\cdot i})d\mathbf{a}_{\cdot i} \nonumber\\
&\ge \Bigl(1 - 2.1e^{-c_7\log n}\Bigr) \P\Bigl[\max_j |a_{ji}| \le 5\sqrt{\log n}\Bigr] \\
&\ge 1 - 4n^{-11},
\end{align*}
where we write $\mu$ for the Gaussian density. The last line follows from standard Gaussian tail bounds and a union bound where we used the assumption $m \le n^{3/2}$, provided that $c_7 \ge 11$.

On $\mathcal{A}_1\cap\mathcal{A}_2\cap \{\max_j |a_{ji}|\le 5\sqrt{\log n}\}$, which is independent of $\mathcal{K}$, $\mathcal{L}$ and $\x$, we have, by construction, $h_{i,\mathcal{K},\mathcal{L},\x} = \tilde{h}_{i,\mathcal{K},\mathcal{L},\x}$ for all subsets $\mathcal{K}\subseteq \S\backslash \{i\}$, $\mathcal{L}\subseteq \mathcal{K}$ with $|\mathcal{L}|\ge \frac{1}{2}|\mathcal{K}|$ and $\x\in N^{\mathcal{K}}_{\epsilon}$. In this case, the bound on $\tilde{h}_{\mathcal{K},\mathcal{L},\x}(\A)$ also applies to $h_{\mathcal{K},\mathcal{L},\x}(\A)$, so that
\begin{align*}
\P[B_i] &\ge \P\Bigl[\widetilde{B}_i \cap \mathcal{A}_1\cap \mathcal{A}_2\cap \Bigl\{\max_j |a_{ji}|\le 5\sqrt{\log n}\Bigr\}\Bigr] \\
&\ge 1 - 4n^{-11} - c_2n^{-11} - n^{-11}\\
&\ge 1- \frac{c_p}{4}n^{-11},
\end{align*}
for $c_p = 4c_2 + 20$.
Finally, taking the union bound over all $i\in [n]$ establishes the bound (\ref{eq:tech6_1}) with probability at least $1-\frac{c_p}{4}n^{-10}$.

\textbf{Proof that (\ref{eq:tech6_2})--(\ref{eq:tech6_4}) are satisfied with high probability}\\
Inequalities (\ref{eq:tech6_2}) and (\ref{eq:tech6_4}) can be shown following the exact same steps as above for (\ref{eq:tech6_1}). The former is slightly simpler because of the lower order of the expression we need to control, which is also the reason for the term $\sqrt{\log n}$ in the bound (\ref{eq:tech6_2}) instead of $\log n$. The proof of inequality (\ref{eq:tech6_3}) is identical to the one of (\ref{eq:tech6_2}), as we can condition on the random variables $\{\varepsilon_j\}_{j=1}^m$ (the same way we conditioned on $\mathbf{A}_{\cdot i} = \mathbf{a}_{\cdot i}$ in Step 1 above) and use the standard tail bound $\P[|\varepsilon_j|>t]\le \exp(1-\tilde{c}t/\sigma^2)$ for sub-exponential random variables, where $\tilde{c}>0$ is an absolute constant. We omit the details of the proof to avoid repetition.
\end{proof}

\begin{proof}[Proof of Lemma \ref{lemma:tech7}]
We will show that each of the inequalities (\ref{eq:tech7_1})--(\ref{eq:tech7_3}) is satisfied with probability at least $1 - \frac{c_p}{4}n^{-10}$. The lemma then follows by taking a union bound. 

We begin by showing the bound (\ref{eq:tech7_1}). To show that (\ref{eq:tech7_1}) holds with high probability, we will show that it is satisfied for all $\{\mathbf{a}_j\}_{j=1}^m\in\mathcal{A}\subseteq \R^{m\times n}$, where $\mathcal{A}$ is defined as the set where the bounds from Lemmas \ref{lemma:tech3} and \ref{lemma:tech6} are satisfied, and such that the projection of $\mathcal{A}$ onto $\R^{m\times \mathcal{S}}$ satisfies Lemma \ref{lemma:tech1} with $t = 5\sqrt{\log n}$. By the aforementioned lemmas, we have $\P[\{\mathbf{A}_j\}_{j=1}^m\in\mathcal{A}]\ge 1-\frac{c_p}{4}n^{-10}$ for a universal constant $c_p>0$.

\textbf{Proof that (\ref{eq:tech7_1}) is satisfied with high probability}\\
Let $\{\mathbf{a}_j\}_{j=1}^m\in \mathcal{A}$ and fix any $i\in [n]$. 
We denote the term that we need to bound in (\ref{eq:tech7_1}) by
\begin{equation*}
g(\x) = \frac{1}{m}\sum_{j=1}^ma_{ji}(\mathbf{a}_{j,-i}^\top\x_{-i})^3,
\end{equation*}
and we consider the constrained optimization problem
\begin{equation}
\label{eq:optconstr}
\max_{\x} g(\x) \quad \text{s.t.} \quad \x\in\mathcal{X}_b,
\end{equation}
where we write $\mathcal{X}_b = \{\x\in\R^n: \x_{\S^c} = \mathbf{0}, \|\x\|_2^2 \le 1, \|\x\|_1\le b\}$ for some $0\le b \le \sqrt{k}$. We will show that any optimizer $\x'\in\mathcal{X}_b$ satisfies $g(\x')\le cb \sqrt{\frac{\log n}{m}}$. By symmetry, we can also obtain the lower bound $g(\x')\ge -cb \sqrt{\frac{\log n}{m}}$ by considering the corresponding minimization problem. Since the above holds for an arbitrary $0\le b\le \sqrt{k}$, this would then complete the proof that inequality (\ref{eq:tech7_1}) is satisfied for all $\{\mathbf{a}_j\}_{j=1}^m\in \mathcal{A}$ and $i\in [n]$.

Since we are maximizing a continuous function over a compact set, a global maximum is attained at some feasible point $\x'\in\mathcal{X}_b$. Using the KKT conditions at $\x'$, we will bound the Lagrange multiplier $\mu_2^\star$ corresponding to the constraint $\|\x\|_1\le b$. This controls how the maximum attainable value $g(\x')$ can increase if we relax the constraint $\|\x\|_1\le b$ (for details see e.g.\ Section 5.6 of \cite{BV04}), and integrating over $b$ gives the desired result.

\textbf{Step 1: Establish KKT conditions}\\
Let $\x'\in\mathcal{X}_b$ be an optimal point of the constrained optimization problem (\ref{eq:optconstr}).
In order to establish the KKT conditions at $\x'$, we verify that the linear independence constraint qualification (LICP) (see e.g.\ \cite{P73}) holds in this problem: the gradients of all active inequality constraints are linearly independent at any point $\x$. We restrict our attention to $b\neq \sqrt{s}$ for $s=1,...,k$, which is a set of measure zero and can be ignored when integrating. If only one inequality constraint is binding, there is nothing to show. Otherwise, the gradients of the constraints, $\frac{\partial}{\partial \x}\|\x\|_2^2 = 2\x$ and $\frac{\partial}{\partial \x}\|\x\|_1 = \operatorname{sign}(\x)$ can only be linearly dependent if $x_i\in \{-\tilde{c},0,\tilde{c}\}$ for all coordinates $i\in [n]$, where $\tilde{c}>0$ is some constant. But then $\|\x\|_2^2 = s\tilde{c}^2$ and $\|\x\|_1 = s\tilde{c}$, where $s$ is the number of non-zero coordinates of $\x$. Since $b\neq \sqrt{s}$, the two constraints cannot be simultaneously binding.

From LICP it follows that the following KKT conditions are satisfied at $\x'$: writing 
\begin{equation*}
\mathcal{L}(\x,\boldsymbol{\lambda},\boldsymbol{\mu}) = g(\x) - \mu_1 \big(\|\x\|_2^2-1\big) - \mu_2 \big(\|\x\|_1-b\big) + \sum_{r\in \S^c}\lambda_rx_r
\end{equation*}
for the Lagrangian, there exist Lagrange multipliers $\boldsymbol{\lambda}^*$ and $\boldsymbol{\mu}^\star$ satisfying 
\begin{align}
\nabla_{\x}\mathcal{L}(\x',\boldsymbol{\lambda}^\star,\boldsymbol{\mu}^\star) &= \mathbf{0} \quad &&\text{stationarity} \label{eq:kkt1}\\
\x_{\S^c} = \mathbf{0}, \quad \|\x'\|_2^2 \le 1, \quad \|\x'\|_1&\le b &&\text{primal feasibility}\label{eq:kkt2}\\
\mu_1^\star\ge 0, \quad \mu_2^\star&\ge 0 &&\text{dual feasibility}\label{eq:kkt3}\\
\mu_1^\star\big(\|\x'\|_2^2-1\big) = 0, \quad \mu_2^\star\big(\|\x'\|_1-b\big) &= 0 &&\text{complementary slackness}\label{eq:kkt4}
\end{align}

\textbf{Step 2: Bounding the Lagrange multiplier $\mu_2^\star$}\\
Using the stationarity condition (\ref{eq:kkt1}), we can compute, for any $l\in \S\backslash\{i\}$, 
\begin{align*}
\frac{\partial}{\partial x_l} \mathcal{L}(\x',\boldsymbol{\lambda}^\star,\boldsymbol{\mu}^\star) &= \frac{3}{m}\sum_{j=1}^ma_{ji}a_{jl}(\mathbf{a}_{j,-i}^\top\x'_{-i})^2 - 2\mu_1^\star x'_l - \mu_2^\star\operatorname{sign}(x'_l) = 0 \\
\Rightarrow \hspace{5em} h_l(\x') &:= \frac{1}{m}\sum_{j=1}^ma_{ji}a_{jl}(\mathbf{a}_{j,-i}^\top\x'_{-i})^2 = \frac{2}{3}\mu_1^\star x'_l + \frac{1}{3}\mu_2^\star\operatorname{sign}(x'_l),
\end{align*}
This identity implies that $h_l(\x')$ has the same sign as $x'_{l}$ because of dual feasibility (\ref{eq:kkt3}), and that $x'_{l}=0$ if and only if $h_l(\x') = 0$. Rearranging for $\mu_2^\star$, we obtain, for any $l$ with $x'_{l}\neq 0$,
\begin{equation}\label{eq:shadowprice}
\mu_2^\star = 3|h_l(\x')| - 2\mu_1^\star|x'_l| \le 3|h_l(\x')|,
\end{equation}
where we again used that $\mu_1^\star\ge 0$. Next, we show that
\begin{equation}\label{eq:shadowbound}
\min_{l:x'_l\neq 0} |h_l(\x')| \le \frac{c\log n}{3\sqrt{m}}.
\end{equation}
To this end, let $\mathcal{K} = \{r\in \S: x'_r\neq 0\}$ and let $\mathcal{L} = \{l\in \mathcal{K}: x'_l>0\}$. Note that we must have $x'_i = 0$, as $x_i$ does not contribute to the value of $g(\x)$. We can assume without loss of generality that $|\mathcal{L}|\ge \frac{1}{2}|\mathcal{K}|$, as we can otherwise consider the set defined by $x'_l<0$. Further, assume for notational simplicity that $\|\x'\|_2 = 1$; otherwise, we can replace $\x'$ by $\x' / \|\x'\|_2$, which makes every $|h_l(\x')|$ larger, since we have $\|\x'\|_2\le 1$ by primal feasibility (\ref{eq:kkt2}). 

Let $N_{\epsilon}^{\mathcal{K}}$ be an $\epsilon$-net of the set $\{\x\in\R^n: \x_{\mathcal{K}^c} = \mathbf{0}, \|\x\|_2 = 1\}$, where $\epsilon = c\log n/(465\sqrt{m})$. Let $\x\in N_{\epsilon}^{\mathcal{K}}$ with $\|\x-\x'\|_2 \le \epsilon$. Then, recalling $m\ge c_s\max\{k^2\log^2n, \; \log^5n\}$, we have by inequality (\ref{eq:tech6_1}) of Lemma \ref{lemma:tech6},
\begin{equation*}
\sum_{l\in \mathcal{L}} h_l(\x) \le c_1\mathcal{L}|\frac{\log n}{\sqrt{m}} \le \frac{c|\mathcal{L}|\log n}{6\sqrt{m}},
\end{equation*}
provided that $c\ge 6c_1$, where $c_1$ is the universal constant from Lemma \ref{lemma:tech6}. Since $h_l(\x)>0$ for all $l\in\mathcal{L}$, there must be an index $l\in \mathcal{L}$ with $h_l(\x) \le c\log n/(6\sqrt{m})$ by the pigeonhole principle. By the Cauchy-Schwarz inequality, we can bound
\begin{align*}
|h_l(\x)-h_l(\x')| &= \Bigg|\frac{1}{m}\sum_{j=1}^m a_{ji}a_{jl}\bigl(\mathbf{a}_{j,-i}^\top(\x_{-i}+\x_{-i}')\bigr)\bigl(\mathbf{a}_{j,-i}^\top(\x_{-i}-\x_{-i}')\bigr)\Bigg| \nonumber\\
&\le \Biggl(\frac{1}{m}\sum_{j=1}^ma_{ji}^2a_{jl}^2\bigl(\mathbf{a}_{j,-i}^\top(\x_{-i}+\x_{-i}')\bigr)^2\Biggr)^{\frac{1}{2}} \Biggl( \frac{1}{m}\sum_{j=1}^m\bigl(\mathbf{a}_{j,-i}^\top(\x_{-i}-\x_{-i}')\bigr)^2\Biggr)^\frac{1}{2} \nonumber\\
&\le \Biggl(\frac{1}{m}\sum_{j=1}^ma_{ji}^8\Biggr)^{\frac{1}{8}}\Biggl(\frac{1}{m}\sum_{j=1}^ma_{jl}^8\Biggr)^{\frac{1}{8}}\Biggl(\frac{1}{m}\sum_{j=1}^m\bigl(\mathbf{a}_{j,-i}^\top(\x_{-i}+\x_{-i}')\bigr)^4\Biggr)^{\frac{1}{4}} \\*
&\quad \cdot \Biggl(\frac{1}{m}\sum_{j=1}^m\bigl(\mathbf{a}_{j,-i}^\top(\x_{-i}-\x_{-i}')\bigr)^2\Biggr)^\frac{1}{2}\nonumber\\
&\le 106^\frac{1}{4}  \frac{(3m)^{\frac{1}{4}} + \sqrt{k} + 5\sqrt{\log n}}{m^\frac{1}{4}}\|\x+\x'\|_2  \frac{\sqrt{m} + \sqrt{k} + 5\sqrt{\log n}}{\sqrt{m}}\|\x-\x'\|_2\\
&\le \frac{c\log n}{6\sqrt{m}},
\end{align*}
where we used Lemmas \ref{lemma:tech1} and \ref{lemma:tech3}, and $\|\x-\x'\|\le c\log n/(465\sqrt{m})$. Since $h_l(\x')\ge 0$ is non-negative, this completes the proof of (\ref{eq:shadowbound}).
Combining (\ref{eq:shadowprice}) with (\ref{eq:shadowbound}), we obtain
\begin{equation*}
\mu_2^\star \le c\frac{\log n}{\sqrt{m}}.
\end{equation*}

\textbf{Step 3: Showing the inequality (\ref{eq:tech7_1})}\\
In order to show the inequality (\ref{eq:tech7_1}), define the value function
\begin{equation*}
v(b) = \max_{\x} \bigl\{g(\x): \x_{\S^c}= \mathbf{0}, \|\x\|_2^2\le 1, \|\x\|_1\le b\bigr\}.
\end{equation*}
We use that $\mu_2^\star$ is the shadow price of the constraint $\|\x\|_1\le b$ (see e.g.\ Section 5.6 of \cite{BV04}):
\begin{equation*}
\frac{\partial}{\partial b}v(b) = \mu_2^\star(b).
\end{equation*}
By definition, we have $v(0) = 0$. For any $0\le b_0\le \sqrt{k}$, the fundamental theorem of calculus yields
\begin{align*}
v(b_0) = \int_0^{b_0} \frac{\partial}{\partial b}v(b)db
= \int_0^{b_0} \mu_2^\star(b)db \le \int_0^{b_0}c\frac{\log n}{\sqrt{m}}db 
< cb_0 \frac{\log n}{\sqrt{m}},
\end{align*}
By definition, we have $g(\x)\le v(b_0)$ for all $\x\in\mathcal{X}$ with $\|\x\|_1\le b_0$.
Hence, 
\begin{equation*}
g(\x) \le cb \frac{\log n}{\sqrt{m}}
\end{equation*}
must hold for any $\x$ with $\|\x\|_2^2\le 1$ and $\|\x\|_1\le b$,
which completes the proof of (\ref{eq:tech7_1}).

\textbf{Proof that (\ref{eq:tech7_2})--(\ref{eq:tech7_4}) are satisfied with high probability}\\
The proofs of (\ref{eq:tech7_2}), (\ref{eq:tech7_3}) and (\ref{eq:tech7_4}) follow the same steps as the proof of (\ref{eq:tech7_1}), using the bounds (\ref{eq:tech6_2}), (\ref{eq:tech6_3}) and (\ref{eq:tech6_4}) of Lemma \ref{lemma:tech6}, respectively, to bound the shadow price $\mu_2^\star$. We omit the details of the proof to avoid repetition.
\end{proof}

%% file: main.bbl
\begin{thebibliography}{52}
\providecommand{\natexlab}[1]{#1}
\providecommand{\url}[1]{\texttt{#1}}
\expandafter\ifx\csname urlstyle\endcsname\relax
  \providecommand{\doi}[1]{doi: #1}\else
  \providecommand{\doi}{doi: \begingroup \urlstyle{rm}\Url}\fi

\bibitem[Arora et~al.(2019)Arora, Cohen, Hu, and Luo]{ACHL19}
S.~Arora, N.~Cohen, W.~Hu, and Y.~Luo.
\newblock Implicit regularization in deep matrix factorization.
\newblock In \emph{Advances in Neural Information Processing Systems},
  volume~32, pages 7411--7422, 2019.

\bibitem[Audibert et~al.(2013)Audibert, Bubeck, and Lugosi]{ABL13}
J.-Y. Audibert, S.~Bubeck, and G.~Lugosi.
\newblock Regret in online combinatorial optimization.
\newblock \emph{Mathematics of Operations Research}, 39\penalty0 (1):\penalty0
  31--45, 2013.

\bibitem[Beck and Teboulle(2003)]{BT03}
A.~Beck and M.~Teboulle.
\newblock Mirror descent and nonlinear projected subgradient methods for convex
  optimization.
\newblock \emph{Operations Research Letters}, 31\penalty0 (3):\penalty0
  167--175, 2003.

\bibitem[Boyd and Vandenberghe(2004)]{BV04}
S.~Boyd and L.~Vandenberghe.
\newblock \emph{Convex optimization}.
\newblock Cambridge University Press, Cambridge, 2004.

\bibitem[Bubeck(2015)]{B15}
S.~Bubeck.
\newblock Convex optimization: Algorithms and complexity.
\newblock \emph{Foundations and Trends in Machine Learning}, 8:\penalty0
  231--358, 2015.

\bibitem[Cai et~al.(2016)Cai, Li, and Ma]{CLM16}
T.~Cai, X.~Li, and Z.~Ma.
\newblock Optimal rates of convergence for noisy sparse phase retrieval via
  thresholded {W}irtinger flow.
\newblock \emph{Annals of Statistics}, 44\penalty0 (5):\penalty0 2221--2251,
  2016.

\bibitem[Cand\`{e}s et~al.(2015)Cand\`{e}s, Li, and Soltanolkotabi]{CLS15}
E.~J. Cand\`{e}s, X.~Li, and M.~Soltanolkotabi.
\newblock Phase retrieval via {W}irtinger flow: Theory and algorithms.
\newblock \emph{IEEE Transactions on Information Theory}, 61\penalty0
  (4):\penalty0 1985--2007, 2015.

\bibitem[Chafa{\"\i} et~al.(2012)Chafa{\"\i}, Gu\'{e}don, Lecu\'{e}, and
  Pajor]{CGLP12}
D.~Chafa{\"\i}, O.~Gu\'{e}don, G.~Lecu\'{e}, and A.~Pajor.
\newblock \emph{Interactions between compressed sensing random matrices and
  high dimensional geometry}.
\newblock Soci\'{e}t\'{e} Math\'{e}matique de France, Paris, 2012.

\bibitem[Chen et~al.(2015)Chen, Chi, and Goldsmith]{CCG15}
Y.~Chen, Y.~Chi, and A.~J. Goldsmith.
\newblock Exact and stable covariance estimation from quadratic sampling via
  convex programming.
\newblock \emph{IEEE Transactions on Information Theory}, 61\penalty0
  (7):\penalty0 4034--4059, 2015.

\bibitem[Chung and Lu(2006)]{CL06}
F.~Chung and L.~Lu.
\newblock Concentration inequalities and martingale inequalities: A survey.
\newblock \emph{Internet Mathematics}, 3\penalty0 (1):\penalty0 79--127, 2006.

\bibitem[Dainty and Fienup(1987)]{DF87}
J.~C. Dainty and J.~R. Fienup.
\newblock Phase retrieval and image reconstruction for astronomy.
\newblock In H.~Stark, editor, \emph{Image Recovery: Theory and Application},
  chapter~7, pages 231--275. Academic Press, New York, 1987.

\bibitem[Fresen(2018)]{F18}
D.~J. Fresen.
\newblock Variations and extensions of the gaussian concentration inequality.
\newblock \emph{arXiv preprint arXiv:1812.10938}, 2018.

\bibitem[Ghadimi and Lan(2013)]{GL13}
S.~Ghadimi and G.~Lan.
\newblock Stochastic first- and zeroth-order methods for nonconvex stochastic
  programming.
\newblock \emph{SIAM Journal on Optimization}, 23\penalty0 (4):\penalty0
  2341--2368, 2013.

\bibitem[Ghai et~al.(2020)Ghai, Hazan, and Singer]{GHS20}
U.~Ghai, E.~Hazan, and Y.~Singer.
\newblock Exponentiated gradient meets gradient descent.
\newblock In \emph{International Conference on Algorithmic Learning Theory},
  volume 117, pages 386--407, 2020.

\bibitem[Gunasekar et~al.(2017)Gunasekar, Woodworth, Bhojanapalli, Neyshabur,
  and Srebro]{GWBNS17}
S.~Gunasekar, B.~Woodworth, S.~Bhojanapalli, B.~Neyshabur, and N.~Srebro.
\newblock Implicit regularization in matrix factorization.
\newblock In \emph{Advances in Neural Information Processing Systems},
  volume~30, pages 6151--6159, 2017.

\bibitem[Hand and Voroninski(2016)]{HV16}
P.~Hand and V.~Voroninski.
\newblock Compressed sensing from phaseless {G}aussian measurements via linear
  programming in the natural parameter spaces.
\newblock \emph{arXiv preprint arXiv:1611.05985}, 2016.

\bibitem[Hoff(2017)]{H17}
P.~D. Hoff.
\newblock Lasso, fractional norm and structured sparse estimation using a
  {H}adamard product parametrization.
\newblock \emph{Computational Statistics \& Data Analysis}, 115:\penalty0
  186--198, 2017.

\bibitem[Jaganathan et~al.(2016)Jaganathan, Eldar, and Hassibi]{JEH16}
K.~Jaganathan, Y.~C. Eldar, and B.~Hassibi.
\newblock Phase retrieval: An overview of recent developments.
\newblock In A.~Stern, editor, \emph{Optical Compressive Imaging}, chapter~13,
  pages 263--296. Taylor Francis Group, Boca Raton, FL, 2016.

\bibitem[Jagatap and Hegde(2019)]{JH19}
G.~Jagatap and C.~Hegde.
\newblock Sample-efficient algorithms for recovering structured signals from
  magnitude-only measurements.
\newblock \emph{IEEE Transactions on Information Theory}, 65\penalty0
  (7):\penalty0 4434--4456, 2019.

\bibitem[Kivinen and Warmuth(1997)]{KW97}
J.~Kivinen and M.~K. Warmuth.
\newblock Exponentiated gradient versus gradient descent for linear predictors.
\newblock \emph{Information and Computation}, 132\penalty0 (1):\penalty0 1--63,
  1997.

\bibitem[Kot\l{l}owski and Neu(2019)]{KN19}
W.~Kot\l{l}owski and G.~Neu.
\newblock Bandit principal component analysis.
\newblock In \emph{Conference on Learning Theory}, volume~99, pages 1994--2024,
  2019.

\bibitem[Lecu\'{e} and Mendelson(2015)]{LM15}
G.~Lecu\'{e} and S.~Mendelson.
\newblock Minimax rate of convergence and the performance of empirical risk
  minimization in phase recovery.
\newblock \emph{Electronic Journal of Probability}, 20:\penalty0 1--29, 2015.

\bibitem[Li and Voroninski(2013)]{LV13}
X.~Li and V.~Voroninski.
\newblock Sparse signal recovery from quadratic measurements via convex
  programming.
\newblock \emph{SIAM Journal on Mathematical Analysis}, 45\penalty0
  (5):\penalty0 3019--3033, 2013.

\bibitem[Li et~al.(2018)Li, Ma, and Zhang]{LMZ18}
Y.~Li, T.~Ma, and H.~Zhang.
\newblock Algorithmic regularization in over-parametrized matrix sensing and
  neural networks with quadratic activation.
\newblock In \emph{Conference on Learning Theory}, volume~75, pages 2--47,
  2018.

\bibitem[Mattila(1995)]{M95}
P.~Mattila.
\newblock \emph{Geometry of sets and measures in Euclidean spaces. Fractals and
  rectifiability}.
\newblock Cambridge University Press, Cambridge, 1995.

\bibitem[Millane(1990)]{M90}
R.~Millane.
\newblock Phase retrieval in crystallography and optics.
\newblock \emph{Journal of the Optical Society of America A}, 7\penalty0
  (3):\penalty0 394--411, 1990.

\bibitem[Nemirovski and Yudin(1983)]{NY83}
A.~Nemirovski and D.~B. Yudin.
\newblock \emph{Problem complexity and method efficiency in optimization}.
\newblock Wiley, New York, 1983.

\bibitem[Nemirovski et~al.(2009)Nemirovski, Juditsky, Lan, and Shapiro]{NJLS09}
A.~Nemirovski, A.~Juditsky, G.~Lan, and A.~Shapiro.
\newblock Robust stochastic approximation approach to stochastic programming.
\newblock \emph{SIAM Journal on Optimization}, 19\penalty0 (4):\penalty0
  1574--1609, 2009.

\bibitem[Netrapalli et~al.(2015)Netrapalli, Jain, and Sanghavi]{NJS15}
P.~Netrapalli, P.~Jain, and S.~Sanghavi.
\newblock Phase retrieval using alternating minimization.
\newblock \emph{IEEE Transactions on Signal Processing}, 63\penalty0
  (18):\penalty0 4814--4826, 2015.

\bibitem[Ohlsson et~al.(2012)Ohlsson, Yang, Dong, and Sastry]{OYDS12}
H.~Ohlsson, A.~Y. Yang, R.~Dong, and S.~S. Sastry.
\newblock {CPRL}--an extension of compressive sensing to the phase retrieval
  problem.
\newblock In \emph{Advances in Neural Information Processing Systems},
  volume~25, pages 1367--1375, 2012.

\bibitem[Peterson(1973)]{P73}
D.~W. Peterson.
\newblock A review of constraint qualifications in finite-dimensional spaces.
\newblock \emph{SIAM Review}, 15\penalty0 (3):\penalty0 639--654, 1973.

\bibitem[Raskutti et~al.(2014)Raskutti, Wainwright, and Yu]{RWY14}
G.~Raskutti, M.~J. Wainwright, and B.~Yu.
\newblock Early stopping and non-parametric regression: an optimal
  data-dependent stopping rule.
\newblock \emph{Journal of Machine Learning Research}, 15\penalty0
  (1):\penalty0 335--366, 2014.

\bibitem[Schechtman et~al.(2014{\natexlab{a}})Schechtman, Beck, and
  Eldar]{SBE14}
Y.~Schechtman, A.~Beck, and Y.~C. Eldar.
\newblock {GESPAR}: Efficient phase retrieval of sparse signals.
\newblock \emph{IEEE Transactions on Signal Processing}, 62\penalty0
  (4):\penalty0 928--938, 2014{\natexlab{a}}.

\bibitem[Schechtman et~al.(2014{\natexlab{b}})Schechtman, Eldar, Cohen,
  Chapman, Miao, and Segev]{SEC+14}
Y.~Schechtman, Y.~C. Eldar, O.~Cohen, H.~N. Chapman, J.~Miao, and M.~Segev.
\newblock Phase retrieval with application to optical imaging: A contemporary
  overview.
\newblock \emph{IEEE Signal Processing Magazine}, 32\penalty0 (3):\penalty0
  87--109, 2014{\natexlab{b}}.

\bibitem[Schniter and Rangan(2015)]{SR15}
P.~Schniter and S.~Rangan.
\newblock Compressive phase retrieval via generalized approximate message
  passing.
\newblock \emph{IEEE Transactions on Signal Processing}, 63\penalty0
  (4):\penalty0 1043--1055, 2015.

\bibitem[Shalev-Schwartz(2011)]{S11}
S.~Shalev-Schwartz.
\newblock Online learning and online convex optimization.
\newblock \emph{Foundations and Trends in Machine Learning}, 4\penalty0
  (2):\penalty0 107--194, 2011.

\bibitem[Sun et~al.(2018)Sun, Qu, and Wright]{SQW18}
J.~Sun, Q.~Qu, and J.~Wright.
\newblock A geometric analysis of phase retrieval.
\newblock \emph{Foundations of Computational Mathematics}, 18\penalty0
  (5):\penalty0 1131--1198, 2018.

\bibitem[Va\v{s}kevi\v{c}ius et~al.(2019)Va\v{s}kevi\v{c}ius, Kanade, and
  Rebeschini]{VKR19}
T.~Va\v{s}kevi\v{c}ius, V.~Kanade, and P.~Rebeschini.
\newblock Implicit regularization for optimal sparse recovery.
\newblock In \emph{Advances in Neural Information Processing Systems},
  volume~32, pages 2968--2979, 2019.

\bibitem[Va\v{s}kevi\v{c}ius et~al.(2020)Va\v{s}kevi\v{c}ius, Kanade, and
  Rebeschini]{VKR20}
T.~Va\v{s}kevi\v{c}ius, V.~Kanade, and P.~Rebeschini.
\newblock The statistical complexity of early-stopped mirror descent.
\newblock In \emph{Advances in Neural Information Processing Systems},
  volume~33, pages 253--264, 2020.

\bibitem[Vershynin(2012)]{V12}
R.~Vershynin.
\newblock Introduction to the non-asymptotic analysis of random matrices.
\newblock In Y.~Eldar and G.~Kutyniok, editors, \emph{Compressed sensing,
  theory and applications}, chapter~5, pages 210--268. Cambridge University
  Press, Cambridge, 2012.

\bibitem[Wang et~al.(2017)Wang, Giannakis, and Eldar]{WGE17}
G.~Wang, G.~B. Giannakis, and Y.~C. Eldar.
\newblock Solving systems of random quadratic equations via truncated amplitude
  flow.
\newblock \emph{IEEE Transactions on Information Theory}, 64\penalty0
  (2):\penalty0 773--794, 2017.

\bibitem[Wang et~al.(2018)Wang, Zhang, Giannakis, Ak\c{c}akaya, and
  Chen]{WZGAC18}
G.~Wang, L.~Zhang, G.~B. Giannakis, M.~Ak\c{c}akaya, and J.~Chen.
\newblock Sparse phase retrieval via truncated amplitude flow.
\newblock \emph{IEEE Transactions on Signal Processing}, 66\penalty0
  (2):\penalty0 479--491, 2018.

\bibitem[Wei et~al.(2019)Wei, Yang, and Wainwright]{WYW19}
Y.~Wei, F.~Yang, and M.~J. Wainwright.
\newblock Early stopping for kernel boosting algorithms: a general analysis
  with localized complexities.
\newblock \emph{IEEE Transactions on Information Theory}, 65\penalty0
  (10):\penalty0 6685--6703, 2019.

\bibitem[Wu and Rebeschini(2020{\natexlab{a}})]{WR20b}
F.~Wu and P.~Rebeschini.
\newblock A continuous-time mirror descent approach to sparse phase retrieval.
\newblock In \emph{Advances in Neural Information Processing Systems},
  volume~33, pages 20192--20203, 2020{\natexlab{a}}.

\bibitem[Wu and Rebeschini(2020{\natexlab{b}})]{wu2020continuoustime}
Fan Wu and Patrick Rebeschini.
\newblock A continuous-time mirror descent approach to sparse phase retrieval.
\newblock \emph{arXiv 2010.10168}, 2020{\natexlab{b}}.

\bibitem[Wu and Rebeschini(2021)]{WR20}
Fan Wu and Patrick Rebeschini.
\newblock Hadamard wirtinger flow for sparse phase retrieval.
\newblock In \emph{Proceedings of The 24th International Conference on
  Artificial Intelligence and Statistics}, volume 130 of \emph{Proceedings of
  Machine Learning Research}, pages 982--990, 2021.

\bibitem[Yao et~al.(2007)Yao, Rosasco, and Caponnetto]{YLC07}
Y.~Yao, L.~Rosasco, and A.~Caponnetto.
\newblock On early stopping in gradient descent learning.
\newblock \emph{Constructive Approximation}, 26\penalty0 (2):\penalty0
  289--315, 2007.

\bibitem[Yuan et~al.(2019)Yuan, Wang, and Wang]{YWW19}
Z.~Yuan, H.~Wang, and Q.~Wang.
\newblock Phase retrieval via sparse {W}irtinger flow.
\newblock \emph{Journal of Computational and Applied Mathematics},
  355:\penalty0 162--173, 2019.

\bibitem[Zhang et~al.(2018)Zhang, Wang, Giannakis, and Chen]{ZWGC18}
L.~Zhang, G.~Wang, G.~B. Giannakis, and J.~Chen.
\newblock Compressive phase retrieval via reweighted amplitude flow.
\newblock \emph{IEEE Transactions on Signal Processing}, 66\penalty0
  (19):\penalty0 5029--5040, 2018.

\bibitem[Zhang and Yu(2005)]{ZY05}
T.~Zhang and B.~Yu.
\newblock Boosting with early stopping: Convergence and consistency.
\newblock \emph{Annals of Statistics}, 33\penalty0 (4):\penalty0 1538--1579,
  2005.

\bibitem[Zhao et~al.(2019)Zhao, Yang, and He]{ZYH19}
P.~Zhao, Y.~Yang, and Q.-C. He.
\newblock Implicit regularization via {H}adamard product over-parametrization
  in high-dimensional linear regression.
\newblock \emph{arXiv preprint arXiv:1903.09367}, 2019.

\bibitem[Zhou et~al.(2020)Zhou, Mertikopoulos, Bambos, Boyd, and
  Glynn]{ZMBBG20}
Z.~Zhou, P.~Mertikopoulos, N.~Bambos, S.~P. Boyd, and P.~W. Glynn.
\newblock On the convergence of mirror descent beyond stochastic convex
  programming.
\newblock \emph{SIAM Journal on Optimization}, 30\penalty0 (1):\penalty0
  687--716, 2020.

\end{thebibliography}
